\PassOptionsToPackage{unicode}{hyperref}
\PassOptionsToPackage{hyphens}{url}
\documentclass[12pt,leqno]{article}

\usepackage{enumitem}
\usepackage{dsfont}
\usepackage{geometry}
\usepackage{setspace}
\usepackage{amsmath}
\usepackage{amsfonts}
\usepackage{amssymb}
\usepackage{graphicx}%
\usepackage{booktabs}
\usepackage{natbib}
\usepackage{array}
\usepackage{hyperref}
\usepackage{placeins}
\usepackage{xcolor}
\usepackage{chngcntr}
\usepackage[utf8]{inputenc}
\usepackage{pdfsync} 
\usepackage{accents}
\usepackage{arydshln} 
\usepackage{nicematrix}

\usepackage{geometry}
\usepackage{amsmath,amssymb}
\usepackage{cases}
\usepackage{lmodern}
\usepackage{bbm}
\usepackage{iftex}
\usepackage{lipsum}
\usepackage{xcolor}
\usepackage{dsfont}
\usepackage{longtable,booktabs,array}
\usepackage{calc} 
\usepackage{etoolbox}
\usepackage{graphicx}
\usepackage{lscape}
\usepackage{amssymb}
\usepackage{caption}
\usepackage{multirow}

\usepackage[title]{appendix}
\usepackage{setspace}
\usepackage[T1]{fontenc}
\usepackage{hanging} 
\usepackage{bigstrut}
\usepackage{footmisc}
\usepackage{microtype}
\usepackage{xcolor}
\usepackage{xurl}
\usepackage{bookmark}
\usepackage{subcaption}

\newtheorem{assumption}{Hypothesis}
\newtheorem{assumptionalt}{Hypothesis}[assumption]

\newenvironment{assumptionp}[1]{
  \renewcommand\theassumptionalt{#1}
  \assumptionalt
}{\endassumptionalt}

\newtheorem{corollary}{Corollary}

\newtheorem{definition}{Definition}
\newtheorem{example}{Example}

\newtheorem{lemma}{Lemma}

\newtheorem{proposition}{Proposition}

\renewcommand{\cite}{\citeasnoun}
\newenvironment{proof}[1][Proof]{\textbf{#1.} }{\ \rule{0.5em}{0.5em}}

\newcommand{\mycomment}[1]{}

\newcommand{\f}{_{50}}
\newcommand{\e}{_{80}}

\newcommand{\pD}{\rho_{D}}
\newcommand{\pND}{\rho_{ND}}
\newcommand{\qD}{q_D}
\newcommand{\qND}{q_{ND}}
\newcommand{\END}{E_{NC}}
\newcommand{\DND}{D_{ND}}
\newcommand{\ELND}{EL_{ND}}
\newcommand{\eps}{\varepsilon}

\newcommand{\vCa}{v_{C^\alpha}}
\newcommand{\vDa}{v_{D^\alpha}}

\newcommand{\Crn}{C^1_*}
\newcommand{\Drn}{D^1_*}
\newcommand{\Da}{D^\alpha}
\newcommand{\Ca}{C^\alpha}
\newcommand{\Do}{D^1}
\newcommand{\Co}{C^1}

\graphicspath{{/Users/tvuphu/Dropbox/Certification/}}

\date{\today}

\newcounter{ccc} 

\begin{document}

\title{{Noisy Certification in a Duopolistic Setting with Loss-Averse Buyers}\thanks{\protect\linespread{1}\protect\selectfont We are grateful to Lawrence Ausubel, Emel Filiz-Ozbay, Doruk Iris, Ginger Jin, Yusufcan Masatlioglu, Erkut Ozbay, Daniel Vincent as well as attendants of various economic seminars for their helpful feedback and comments. This work was supported by the Creative-Pioneering Researchers Program at Seoul National University and by
the Research Grant of the Center for Distributive Justice at the Institute of Economic Research, Seoul National University.}}
\author{Dmitry A. Shapiro\thanks{\protect\linespread{1}\protect\selectfont Department of Economics and SNU Institute of Economic
Research, College of Social Sciences, Seoul National University, 1 Gwanak-ro, Gwanak-gu,
Seoul 08826, South Korea. Email:
\texttt{dmitry.shapiro@snu.ac.kr.} } \and
Tri Phu Vu\thanks{\protect\linespread{1}\protect\selectfont  Department of Economics, University of Maryland, 3114 Tydings Hall, 7343 Preinkert Dr., College Park, MD 20742. Email: \texttt{tvuphu@umd.edu.} }
}

\maketitle

\begin{abstract}
\linespread{1.0}\selectfont
This paper studies how noise in certification technology affects seller profits in a duopoly with unobservable product quality. We identify two opposing effects of noisy certification. First, it reduces the informativeness of certification outcomes, homogenizing buyers' beliefs and limiting the scope for vertical differentiation. Second, it introduces randomness into buyer perceptions, endogenously generating differentiation between otherwise similar products. When buyers are risk-neutral, the first effect dominates, reducing seller profits. However, when buyers are loss averse, the negative impact of reduced informativeness is mitigated, and noisy certification can increase profits relative to accurate certification. Experimentally, treatments with inaccurate certification are more profitable than those with accurate certification, particularly in settings with intense competition.
\end{abstract}



\section{Introduction}

When quality is difficult to verify, third-party certification plays a crucial role in reducing uncertainty and guiding economic decisions. In pharmaceuticals, FDA approval distinguishes safe, effective drugs from potential hazards. In financial markets, credit ratings from agencies like Moody's or S\&P shape investment flows by affecting perceptions of risk and creditworthiness. In academia, peer review certifies research quality, determining publication outcomes and scholarly impact. In consumer markets, certified labels such as `organic' or `fair trade' build trust, and retailers' quality seals, certified pre-owned programs, and platform badges influence discovery, choice, and prices.

Yet, while certification helps establish credibility and inform decision-making, it is far from perfect. For example, some articles published in top journals fail to garner citations, while groundbreaking papers can initially struggle to pass peer review.\footnote{Card and DellaVigna (2013) found that among articles published in top-five economics journals from 1970 to 2012, 176 (1.3\%) received no citations. In contrast, the 2023 Nobel Prize-winning paper on mRNA vaccines (Karik\'{o}, Buckstein, Ni, and Weissman, 2005) was initially rejected by prestigious journals like \textit{Nature}, \textit{Science}, and \textit{Cell} due to its perceived incremental contribution (see \url{https://www.science.org/content/blog-post/nobel-modified-mrna}).} Similarly,
Moody's and S\&P erroneously assigned AAA ratings to mortgage-backed securities, a misjudgment that Bloomberg (2008) identified as a critical contributor to ``\textit{the world financial system's biggest crisis since the Great Depression.}'' Beyond formal certification, buyers also rely on noisy quality signals such as platform-assigned badges and algorithmic labels (for example, `Amazon's Choice' on Amazon, `Top Rated Seller' on eBay, and `Superhost' on Airbnb), as well as coarse letter grades, star ratings, and online product and service reviews.

Such inaccuracies frequently arise from systemic features of certification and rating systems rather than isolated anomalies. Unreliable consumer feedback and other forms of noisy data can distort evaluation accuracy (Resnick and Zeckhauser, 2002). Conflicts of interest, such as rating agencies being paid by issuers, consistently yield biased ratings (Beaver, Shakespeare, and Soliman, 2006). Competition among certifiers can weaken standards, as evidenced by declines in bond-rating accuracy during periods of heightened rivalry (Becker and Milbourn, 2011). Moreover, certifiers may prioritize reputation over precision, and well-meaning efforts such as disclosures of conflicts of interest can unintentionally exacerbate bias (Benabou and Laroque, 1992; Cain, Loewenstein, and Moore, 2005).

These observations motivate us to investigate the following questions. Do sellers sometimes have incentives to pursue certifications that may misrepresent quality? How does certification noise affect price competition and profitability in markets with competing sellers? Can industries, perhaps counterintuitively, benefit from noisy certification relative to perfectly accurate certification?

We develop a theoretical model and a series of laboratory experiments to study these questions. Two sellers have private information about product quality. Each can obtain certification from a noisy technology, for example a retailer or platform `verified' badge awarded or withheld based on threshold rules, or a multi-level rating such as a 1 to 5-star score or an A to C grade derived from a noisy underlying assessment. After obtaining a result, the seller chooses whether to disclose it. The two sellers then observe either the disclosed outcome or the competitor's non-disclosure and set prices. Buyers, who may be risk-neutral or loss-averse, make purchase decisions based on prices and any disclosed information, or its absence.

Noisy certification has two primary effects on seller profits. The first is a reduction in informativeness, the \textit{blurring effect}. Noise weakens the distinction between the posterior expected qualities of the two products, which narrows the difference in buyers' willingness to pay (WTP). In the limit, when certification conveys no information, willingness to pay is equalized and Bertrand competition drives profits to zero. The second is endogenous differentiation, the \textit{differentiation effect}. Randomness in outcomes can make the two products look different even when underlying qualities are similar. For example, if both sellers have the same quality and both certify, accurate certification yields equal WTP and zero profits, whereas noisy certification creates a positive probability that outcomes differ, which yields positive expected profits.

Noisy certification may also create non-equilibrium differentiation through belief heterogeneity. With accurate certification, disclosed outcomes make product quality common knowledge, so buyers share beliefs. With noisy certification, buyers can hold different beliefs about quality. This belief dispersion can soften price competition and increase profits. Although our theoretical analysis adopts a standard equilibrium framework with buyers holding correct and identical on-equilibrium beliefs, belief heterogeneity likely contributes to some of the patterns observed in our experimental results.

The theoretical analysis delivers two main results. First, with risk-neutral buyers, accurate certification is always more profitable in equilibrium than any equilibrium with noisy certification. The informativeness loss dominates the differentiation gain. Second, with loss-averse buyers, this ranking can reverse. Loss aversion attenuates the blurring effect because willingness to pay reflects both expected quality and expected loss. Noise may reduce differences in expected quality while increasing differences in expected loss, so the willingness-to-pay gap can be larger under loss aversion than under loss (and risk) neutrality. With a weaker blurring effect, the differentiation effect can dominate, making noisy certification more profitable. This can occur for moderate loss aversion and arbitrarily small noise.

Whether noisy certification raises profitability is a priori ambiguous. Outcomes depend on loss aversion, certification accuracy and cost, and the distribution of quality. Our theoretical analysis isolates the mechanisms and the conditions under which each dominates but does not deliver a universal ranking, so we complement it with laboratory experiments that compare profitability of accurate and noisy certification environments.

The experimental design mirrors the theoretical framework and introduces three key variations. First, we manipulate the availability and accuracy of certification.
We consider treatments with no certification (D1), accurate certification (D2), and noisy certification (D3).  Second, we consider two different quality distributions: a discrete uniform distribution with support $\{50,51,\dots,100\}$, labeled D$_{50}$, and another with support $\{80,81,\dots,100\}$, labeled D$\e$.\footnote{One can interpret D$\e$
as a narrow-dispersion environment with relatively homogeneous quality, and D$\f$
as a broad-dispersion environment with heterogeneous quality. Examples include phone chargers, HDMI cables, and mid-tier hotel rooms for the former; used cars and boutique electronics for the latter.} The latter case, where sellers' qualities are closer to each other, is expected to amplify the differentiation effect and boost profitability under noisy certification. Third, where applicable, we vary the certification fee, $c$, and the certification precision, $\alpha$.

Our main experimental result is that treatments with noisy certification (D3) generate higher seller profits than treatments with accurate certification (D2). This pattern holds under both quality distributions and for most values of the certification fee and the accuracy parameter. The gain is modest in the broad-dispersion environment D\(_{50}\), and sizeable in the narrow-dispersion environment D\(_{80}\).

Several factors explain why D2 is less profitable. First, sellers in D2 tend to over-certify, acquiring costly certification more frequently than their counterparts in D3. However, this alone does not account for the disparity. Even when certification costs are excluded, D2 remains less profitable than D3 across both quality distributions. Another key factor is that D3 reduces competitive intensity relative to D1 and D2. Specifically, D3$_l$ sellers ($l\in{50,80}$) charge the highest prices and capture the largest share of total surplus (the sum of seller and buyer profits) among the three D$_l$ treatments.

The effect of noisy certification on competition is especially evident in the more competitive D$\e$ treatments. Despite D$\e$ offering a higher expected quality (90 versus 75 in D$\f$), average prices in D1$\e$ and D2$\e$ fall below those in their D$\f$ counterparts. By contrast, D3$\e$ prices exceed those in D3$\f$. A similar pattern emerges with profits: earnings in D1$\e$ and D2$\e$ are roughly equivalent to those in D1$\f$ and D2$\f$, whereas D3$\e$ profits are substantially higher than those in D3$\f$. Thus, while accurate certification in D2$\e$ fails to ease competitive pressure, noisy certification in D3$\e$ does, allowing sellers to capitalize on the higher expected quality in D$\e$.

\section{Literature review}
\label{sec:review}
Product certification is a common practice for sellers to convey quality information to uninformed buyers (Dranove and Jin, 2010). This paper contributes to the literature examining the impact of certification on market outcomes. While many studies in this area analyze settings with a monopolistic seller or supplier (Strausz, 2005; Peyrache and Quesada, 2011; Van Der Schaar and Zhang, 2015; Pollrich and Wagner, 2016; Stahl and Strausz, 2017; Chen and Lee, 2017; Marinovic et al., 2018; Huh et al., 2023), our work falls within the smaller group focusing on oligopoly markets. We \textcolor{black}{theoretically and experimentally} show that the impact of certification in such markets leads to distinctly different outcomes compared to single-seller environments.

Within the \textcolor{black}{limited} literature on certification under oligopoly, our paper is related to De and Nabar (1991), Bottega and Freitas (2019), and Zhang and Li (2020). In De and Nabar's framework, the market is perfectly competitive and oligopolistic sellers are price-takers. Our setting is different as sellers can strategically set prices. Zhang and Li (2020) consider an environment with two risk-neutral sellers and loss-averse buyers like ours.  However, their work assumes fully accurate certification for determining private quality, which we relax by introducing noisy certification. Similarly, Bottega and Freitas (2019) explore inaccurate certification in a duopoly market with binary quality levels and risk-neutral buyers. Our paper differs in that the seller's quality is pure private information, i.e., it is unobserved by the competitor, and in allowing for loss-averse buyers. As we demonstrate, the latter significantly alters the impact of noisy certification.



Our study builds upon Huh, Shapiro and Ham (2023) by sharing similarities in buyer utility and the underlying certification technology. However, our paper is different in two key aspects. First, we move beyond the monopolistic setting of Huh et al. (2023) to explore a duopoly market. \textcolor{black}{In a monopoly, the seller's profit depends on the \textit{levels} of the buyers' willingness to pay. In contrast, in a duopoly, the seller's profit depends on the \textit{difference} in buyers' willingness to pay. In Section \ref{sec:theory}, we demonstrate that this distinction leads to dramatically different predictions. For instance, consider a simple environment in which buyers are risk neutral. Under this scenario, duopolistic sellers earn zero profit when certification is unavailable, as buyers hold identical beliefs about product quality. When certification becomes available, it allows duopolistic sellers to differentiate themselves and increase their profits. Conversely, a monopolistic seller earns a positive profit even without certification, and the introduction of certification reduces their profit.} Second, we introduce variable certification costs and precisions \textcolor{black}{in our experiments}, enabling us to analyze their impact on seller profits and market outcomes. This element is absent in Huh et al. (2023), where both cost and precision are fixed.

Our paper is also related to the experimental literature on Bertrand competition and quality disclosure. An important paper in this literature, Dufwenberg and Gneezy (2000), examines a simplified setting where participants choose integers between 2 and 100, with the lowest number winning. They observe the limitations of the Bertrand solution with two sellers, but not with three or more.  Our study differs in that we introduce seller quality heterogeneity and a richer strategic environment where sellers choose both prices and quality messages, \textcolor{black}{allowing us to study sellers' strategic decisions on information acquisition.} More recent papers investigate Bertrand competition under various aspects, including increasing marginal costs (Abbink and Brandts, 2008), risk aversion (Anderson et al., 2012), price complexity (Kalayci, 2015), and dishonesty about production cost (Feltovich, 2019). With regard to experimental literature on quality disclosure, quality information in our paper can be communicated via either cheap talk (non-verifiable disclosure), correct certification (verifiable disclosure), or noisy certification (partly verifiable disclosure). This departs from studies like S\'{a}nchez-Pag\'{e}s and Vorsatz (2009) and Sheth (2021) that allow for seller silence.
Additionally, our experiment permits dishonest quality reporting when sellers do not certify or choose not to reveal certification outcomes conditional on certifying. The dishonest disclosure characteristic aligns with the settings in Cai and Wang (2006) and Wang et al. (2010) and distinguishes our experiment from environments in Jin et al. (2021), or Sheth (2021), which focus solely on truthful cheap-talk messages.

\section{Theoretical Framework}
\label{sec:theory}

\subsection{Model}
\label{subsec:existence}

Consider a game between two risk-neutral sellers and a unit mass of price-taking buyers with unit demand. Each seller offers a product of quality $v\in\mathbb{V}=\{v_1,v_2,\dots,v_n\}$ with $0<v_1<\dots<v_n$. Product quality is the seller's private information and is unobservable to buyers and to the rival seller. The prior probability of type $v_i$ is $q_i\in(0,1)$ and is common knowledge. Production cost is zero, regardless of quality.

A third-party intermediary offers certification for a strictly positive fee $c$. The certification outcome $s$ takes values in the quality support, $s\in\mathbb{V}$. The technology is accurate with probability $\alpha$, in which case the outcome equals the true quality, $\Pr(s=v_i\mid v=v_i,\text{ success})=1$. If accuracy fails with probability $1-\alpha$, the outcome is drawn from the prior distribution, $\Pr(s=v_j\mid v=v_i,\text{ fail})=q_j$.\footnote{This noise structure preserves the mean and yields simple posteriors while allowing correlation between quality and outcome; see Klemperer, 1987, Chen et al., 2001, Shin and Sudhir, 2010, and Shin and Yu, 2021. Because failed outcomes are drawn from the prior, the conditional distribution of outcomes under failure does not depend on the true quality, which avoids boundary problems that arise with additive noise $v+\varepsilon$ that can leave the support of $\mathbb{V}$.} Hence,
$\Pr(s=v_i\mid v=v_i)=\alpha+(1-\alpha)q_i$ and $\Pr(s=v_j\mid v=v_i)=(1-\alpha)q_j$ for $j\ne i$. If $v^\prime>v^{\prime\prime}$ and $\alpha>0$, the distribution of certification outcomes conditional on $v^\prime$ first-order stochastically dominates that conditional on $v^{\prime\prime}$. We call certification accurate if $\alpha=1$ and noisy otherwise.

Buyers have loss-averse preferences, where the reference point is endogenously determined by buyers' beliefs regarding expected quality.\footnote{Among the variety of reference-dependent models with endogenous reference points (see, e.g., Gul, 1991; K\H{o}szegi and Rabin, 2006), this reference-dependent model ``\textit{has proven quite popular in applications, as the reference point is neither stochastic nor recursively defined, but is simply the expected consumption utility of the lottery}'' (Masatlioglu and Raymond, 2016, p. 2765).}
Let $\mu$ be buyers' beliefs regarding the distribution of product quality,
and \(v\) be the actual quality of the purchased product. Let $p$ be the paid price. Buyers' ex post utility is
\begin{eqnarray*}
u_B(v,p)=\underbrace{v-p}_{\text{consumption utility}}+\underbrace{b \cdot \min\{v-E_\mu v, 0\}}_{\text{loss utility}}.
\label{eq:utilityLAprimitive}
\end{eqnarray*}
where $b\ge 0$ is common across buyers. When $b=0$ buyers are loss neutral. When $b>0$ they experience a utility loss whenever realized quality falls short of the reference point $E_\mu v$.

The timing is as follows. First, both sellers observe their own types. Second, each seller decides whether to purchase certification and, conditional on the outcome, whether to disclose it publicly. Third, after observing rivals' disclosed outcomes or non-disclosure, sellers set prices. Finally, buyers choose which product to purchase based on prices and observed messages.

Let $C\subseteq \mathbb{V}$ denote the set of types who choose to certify their products, and $D\subseteq \mathbb{V}$ denote the set of outcomes that are disclosed to buyers. We will refer to $C$ as the certification strategy and to $D$ as the disclosure strategy.\footnote{We assume disclosure depends only on the realized outcome, not directly on the seller's type. Unless disclosing a given outcome yields zero payoff, all types optimally make the same disclosure decision for that outcome, so this restriction is without loss of generality for our analysis.} If $s\notin D$, the outcome is not disclosed. Write $ND$ for non-disclosure and let $S$ be the set of all on-equilibrium messages. If $D=\mathbb{V}$, then $S=D$; otherwise $S=D\cup{ND}$.

Buyers' willingness to pay after observing message $s_i\in S$ is
$$
p(v|s_i)=E(v|s_i)+b\cdot EL_i,
$$
where $E(v|s_i)$ is the posterior expected quality given $s_i$, and $EL_i$ is the expected loss conditional on $s_i$
\begin{equation}
EL_i=E_v[\min\{v-E(v|s_i),0\}|s_i].
\label{eq:ELi}
\end{equation}
Buyers' WTP depends on the profile $(C,D)$, the precision $\alpha$, and the degree of loss aversion $b$. If $s\in S$, beliefs are updated by Bayes' rule. If $s\in(\mathbb{V}\cup{ND})\setminus S$, then $s$ is an off-equilibrium message: $E(v|s)$ and $p(v|s)$ are determined given buyers' off-equilibrium beliefs.

Using backward induction, let $\pi_i(s_i,s_j)$ denote seller $i$'s equilibrium profit in the pricing subgame that follows messages $(s_i,s_j)\in S\times S$, where $s_i$ is the message of seller $i$ and $s_j$ is the message of its competitor, seller $j$. We first characterize $\pi_i(s_i,s_j)$ in Proposition \ref{pro:WTPs}.
\begin{proposition}
In the pricing subgame that follows $(s_i,s_j)$, if $p(v|s_i)\ge p(v|s_j)$, then $\pi_i(s_i,s_j)=p(v|s_i)-p(v|s_j)$ and $\pi_j(s_i,s_j)=0$.
\label{pro:WTPs}
\end{proposition}
\noindent The logic in Proposition \ref{pro:WTPs} is standard for unit-demand Bertrand with heterogeneous reservation values. The seller for whom buyers have the higher WTP sets the price equal to the difference in WTPs, which is $p(v| s_i)-p(v| s_j)$, while the rival sets the price $0$. When WTPs are equal, the equilibrium price and profit are zero.

Let $E\pi_i(s_i)=E_{s_j}[\pi_i(s_i,s_j)]$ be seller $i$'s expected profit after disclosing $s_i\in S$. Disclosure of $s_i$ is optimal if and only if
\begin{equation*}
E\pi_i(s_i)\ge E\pi_i(ND).
\end{equation*}

Consider a seller of type $v$. Let $\Pr(s_k|v)$ denote the probability of receiving outcome $s_k$ conditional on having quality $v$. Given disclosure strategy $D$, the expected profit from purchasing certification is
$$
E\Pi_i(v)=\sum_{s_k\in D}\Pr(s_k|v) E\pi_i(s_k)+\Pr(ND) E\pi_i(ND)-c,
$$
where $\Pr(ND)=1-\sum_{s_k\in D}\Pr(s_k|v)$. Buying certification is optimal if and only if $E\Pi_i(v)\ge E\pi_i(ND)$. We are now ready to define an equilibrium.

\begin{definition}
A pure-strategy symmetric equilibrium consists of $(C,D)$, a pricing strategy, and a buyers' purchasing strategy such that:

i) Certification and disclosure decisions are optimal.

ii) The pricing strategy is a Nash equilibrium in every subgame $(s_i,s_j)$.

iii) For any on-equilibrium $s_i\in S$, buyers' posteriors follow Bayes' rule given $(C,D)$.

iv) Buyers' purchase decision is optimal.\footnote{In our theoretical analysis we assume full market coverage. This assumption is standard in duopoly analyses such as Villas-Boas, 1999 and Shin and Sudhir, 2010. In our experiments, the purchase rate was 85.2\%.}
\label{def:eq1}
\end{definition}

We impose the following restrictions on equilibria. First, for accurate certification, we assume that if buyers observe an off-equilibrium message $s\not\in S$, their beliefs are $\Pr(v=s)=1$. This ensures a unique equilibrium in the case of accurate certification.\footnote{Without this refinement, the equilibrium in Proposition \ref{pro:D2RN} need not be unique; however, as shown in its proof, it remains the most profitable one.} Second, we focus on threshold strategies in which all types above a certification threshold seek certification and all outcomes above a disclosure threshold are disclosed. Although non-threshold equilibria exist,
both our experimental results and prior studies (e.g., Jovanovic, 1982; Levin et al., 2009; and Zhang and Li, 2020) indicate that higher-quality sellers are more likely to certify, and better certification outcomes are more likely to be disclosed. Finally, when indifferent between certifying and not certifying, sellers certify with probability 1.

\subsection{Loss-Averse Buyers: Two Types}
\label{subsec:rankingLA}

We use a two-type benchmark to illustrate the two forces that determine how certification noise affects profits. The \textit{blurring effect} is informational: as accuracy falls, disclosed certification outcomes become less informative and the gap in buyers' WTP shrinks. The \textit{differentiation effect} is strategic: noise randomizes which certification outcomes sellers receive, creating the possibility that the difference in buyers' WTP exceeds the underlying quality difference. By Proposition \ref{pro:WTPs}, blurring pushes profits down, whereas differentiation can push them up. Loss aversion can weaken blurring. With loss-averse buyers, the WTP gap reflects differences in both expected qualities and expected losses. Noise reduces the former but can amplify the latter, thereby increasing WTP differences and, in turn, expected profits.

Consider a two-type setting $\mathbb{V}=\{v_L,v_H\}$ with $0<v_L<v_H$ and $\Pr(v_i)=q_i$. Let $s_H$ and $s_L$ denote the high and low certification outcomes. Let $(C,D)=(\{v_H\},\{s_H\})$, so only high types certify and only a high outcome is disclosed. When $\alpha<1$, false negatives arise because some $v_H$ sellers fail to obtain $s_H$.\footnote{$(C,D)=(\{v_L,v_H\},\{s_H\})$ is an example of an equilibrium with both false negatives and false positives.} If buyers observe $s_H$, they infer $v_H$ with certainty. If they observe $ND$, then, unless $\alpha=1$, the product may be either high or low quality.

As an example, consider platform badges that are binary pass/fail and enforced by procedural thresholds that only imperfectly track quality. On Airbnb, Superhost status depends on rolling-window targets for response, cancellations, and ratings; a single emergency cancellation can push an excellent host below a cutoff, so a provider with quality $v_H$ appears without the badge, corresponding to $ND$ in our model. On eBay, Top-Rated Seller status hinges on on-time shipping and defect rates; temporary supply shocks or courier delays can strip the badge from a high-quality seller. These mechanisms make it possible that two otherwise similar high-quality sellers present different quality signals in the same period, one with a badge and one without. In our model, this shifts the pricing subgame from symmetric $(s_H,s_H)$, where the WTP difference is zero, to asymmetric $(s_H,ND)$, where the WTP difference is strictly positive, capturing the differentiation effect.\footnote{Our model features seller-initiated certification with voluntary disclosure, whereas platform badges are awarded and displayed by the platform. The two settings are nonetheless closely related. The decision to seek certification corresponds to a seller's effort to meet badge thresholds, which entails a cost. Disclosure yields the same information sets for buyers: observing a badge signals high quality, as observing $s_H$ does in our setting; the absence of a badge does not imply low quality, as with $ND$. From the seller's perspective, it is always optimal to disclose $s_H$, so voluntary disclosure of $s_H$ and platform assignment of a high-quality badge are strategically identical in the two-type case.}

By Proposition \ref{pro:WTPs}, the only subgame where seller $i$ earns a positive profit is when $(s_i=s_H, s_j=ND)$, with $\pi_i(s_H,ND)=p(v|s_H)-p(v|ND)$.
Hence seller $i$'s expected profit is
\begin{equation}
\Pi_i=\underbrace{\Pr(s_i=s_H, s_j=ND|\alpha)}_{\text{prob of\ non-Bertrand}}\cdot \underbrace{[p(v|s_H)-p(v|ND)]}_{\text{difference in\ buyer's WTP}}- \underbrace{c q_H}_{\text{certification cost}}.
\label{eq:Pi}
\end{equation}

As the example below shows, lowering $\alpha$ moves the two components in opposite directions. It reduces $p(v|s_H)-p(v|ND)$ because certification outcomes become less informative; this is the blurring effect. It can raise $\Pr(s_i=s_H, s_j=ND|\alpha)$ by making asymmetric messages more likely; this is the differentiation effect. When $b=0$, the blurring effect dominates and noisy certification reduces profits. When $b>0$, loss aversion changes how buyers value non-disclosure signal, $ND$. At $\alpha=1$, $ND$ perfectly reveals $v_L$, so $p(v|ND)=v_L$. At $\alpha<1$, $ND$ pools $v_L$ with some $v_H$, so quality under $ND$ is uncertain. Although $E(v|ND)$ rises, loss aversion penalizes the downside: $p(v|ND)=E(v|ND)+b\cdot EL_{ND}<E(v|ND)$ with $EL_{ND}<0$. Since $p(v|s_H)=v_H$ is unchanged, $p(v|s_H)-p(v|ND)$ is larger with loss aversion than under loss (and risk) neutrality. This attenuates the blurring effect and can make noisy certification more profitable.

\begin{example} \label{ex:same} \rm  Suppose $v_L=1, v_H=3, q_H=2/3$, and $b=1$. Seller $i$ earns a positive profit only when he has type $v_H$ and discloses $s_H$, while seller $j$ does not disclose. This occurs  with probability
\[
\Pr(s_i=s_H, s_j=ND|\alpha)=\underbrace{q_H}_{\text{$Pr(v_i=v_H)$}} \underbrace{(\alpha+(1-\alpha) q_H)}_{\text{$Pr(s_i=s_H|v_i=v_H)$}} \underbrace{[q_L+q_H(1-\alpha)q_L]}_{\text{$Pr(s_j=ND)$}}.
\]
It is straightforward to verify that $\Pr(s_i=s_H, s_j=ND|\alpha)>\tfrac{2}{9}=\Pr(s_i=s_H, s_j=ND|\alpha=1)$ for every $\alpha\in[0,1)$. When $\alpha<1$, even if both sellers are high quality, their certification outcomes can differ, leading to asymmetric disclosures with positive probability; this is the differentiation effect.

Under accurate certification, the WTP difference is $p(v| s_H)-p(v| ND)=2$. Under noisy certification with risk-neutral buyers (i.e., $b=0$), the difference is lower as
$p(v|s_H)-p(v|ND)=\frac{6}{5-2\alpha}<2$, which is due to the blurring effect.
With $b=1$, the difference in WTPs equals $\frac{6}{5-2\alpha}\frac{7-4\alpha}{5-2\alpha}$.
This remains below two but is strictly larger than in the risk-neutral case for any $\alpha<1$.
With $ND$-product being risky and $s_H$-product being certain, loss aversion increases $p(v|s_H)-p(v|ND)$; thereby weakening the blurring effect.
When $c\in[4/9,2/3]$, $(C,D)=(\{v_H\},\{s_H\})$ is an equilibrium for any $\alpha\in[0,1]$. In this range,
\[
\Pi^{\alpha=1}_i=\frac{4}{9}-\frac{2c}{3} \text{< } \Pi^{\alpha}_i=\frac{16\alpha^2+4\alpha-56}{54\alpha-135}-\frac{2c}{3} \mbox{\qquad for all } \alpha\in(0.25, 1).
\]

\begin{figure}[!htb]
\minipage{0.32\textwidth}
  \includegraphics[width=\linewidth]{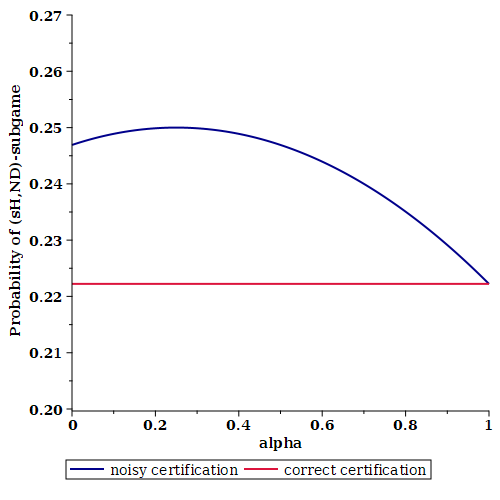}
\endminipage\hfill
\minipage{0.32\textwidth}
 \includegraphics[width=\linewidth]{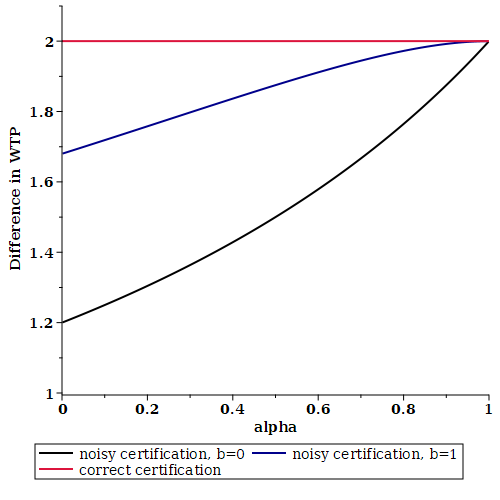}
\endminipage\hfill
\minipage{0.32\textwidth}
  \includegraphics[width=\linewidth]{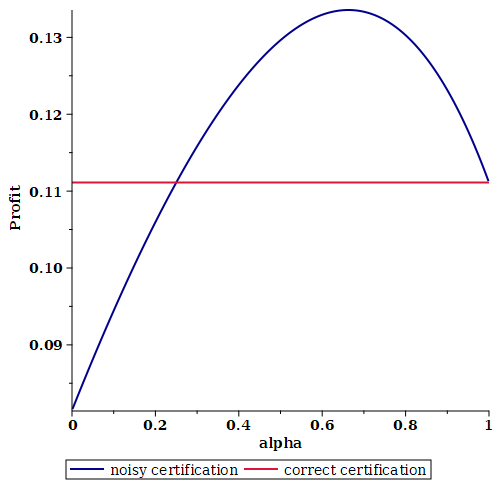}
\endminipage

\medskip
{\footnotesize \rm NOTES: Left panel shows the probability of a positive-profit outcome under accurate and noisy certification. The middle panel shows $p(v\mid s_H)-p(v\mid ND)$ under accurate certification and under noisy certification for $b=0$ and $b=1$. The right panel shows $\Pi_i^\alpha$ and $\Pi_i^1$ when $c=1/2$.}
\end{figure}
\end{example}

Proposition \ref{pro:2Tsame} below formalizes the ideas in Example \ref{ex:same} and provides necessary and sufficient conditions for noisy certification to be more profitable than accurate certification in the two-type setting.

\begin{proposition} \label{pro:2Tsame} Suppose $(C,D)=(\{v_H\},\{s_H\})$ is equilibrium under both accurate and noisy certifications.\footnote{Equilibrium existence conditions are given in Proposition A1 in the Appendix.} The environment with noisy certification is more profitable if and only if:
\begin{equation}
q_H(b+1)> \frac{1}{1-(1-q_H)(1-\alpha)}.
\label{eq:sameEq}
\end{equation}
\end{proposition}

Several remarks follow. First, the value of $c$ does not enter (\ref{eq:sameEq}) because the paid certification cost is the same in both environments. Second, for (\ref{eq:sameEq}) to hold, $b$ must be strictly greater than 0. Noisy certification cannot be more profitable when buyers are risk-neutral. Third, $q_H$ must be sufficiently high for noisy certification to be more profitable. There are two reasons. High $q_H$ is needed for a positive differentiation effect: when $q_H$ is large, lowering $\alpha$ raises the probability of the only asymmetric subgame $(s_H,ND)$. In addition, a higher $q_H$ increases expected loss under $ND$ which, with loss-averse buyers, reduces $p(v|ND)$ and increases $\pi(s_H,ND)=p(v|s_H)-p(v|ND)$. Finally, restrictions on $\alpha$ are mild. There exist parameters for which (\ref{eq:sameEq}) holds for all $\alpha\in[0,1)$. If $q_H(b+1)>1$, values of $\alpha$ arbitrarily close to 1 can satisfy (\ref{eq:sameEq})\footnote{Proposition \ref{pro:2Tsame} applies to $\alpha<1$. In particular, the limit of (\ref{eq:sameEq}) as $\alpha\rightarrow 1$ cannot be used to compare $\lim_{\alpha\rightarrow 1}\Pi_i^\alpha$ and $\Pi_i^1$, which are equal for any parameter values.}, implying that even small amounts of certification noise can raise profits.

\subsection{General Case}
\label{subsec:generalRNLA}

In what follows, let $v_C$ be the lowest certifying type, so $C=\{v\in\mathbb{V}: v\ge v_C\}$, and let $v_D$ be the lowest disclosed outcome, so $D=\{s\in\mathbb{V}: s\ge v_D\}$. To distinguish environments, we use superscript $1$ for variables and outcomes under accurate certification and superscript $\alpha$ for those under noisy certification. When referring to environments in text, we will say ``\textit{100\% certification}'' versus ``$\alpha$ \textit{certification}.''

\subsubsection{No-Certification Benchmark}
\label{subsec:nocert}

Without certification, buyers' WTP for the two products is identical. By Proposition \ref{pro:WTPs}, Bertrand competition then yields zero profit.

\begin{corollary}
In the absence of certification, sellers' profits are zero.
\label{coro:D1}
\end{corollary}

Corollary \ref{coro:D1} highlights a key distinction between the monopoly setting in Huh et al. (2023) and our duopoly framework. In a monopoly, expected profit equals the buyers' ex ante WTP minus certification costs. With risk-neutral buyers, ex ante WTP does not depend on certification availability or precision because $E[E[v|s]]=E[v]$ by the law of iterated expectations; hence, costly certification is weakly dominated by no certification (and remains unattractive under low loss aversion).

In a duopoly, profitability depends on how much a seller's WTP exceeds the rival's, net of certification costs. Without certification, products are perceived as identical and profits are zero. With certification, disclosed outcomes differentiate products in WTP space, allowing positive expected profits. Thus, regardless of buyers' loss aversion, equilibria with certification are at least as profitable as the no-certification benchmark.

\subsubsection{Risk-Neutral Buyers}
\label{subsec:rankingRN}

We now analyze the risk-neutral setting. Under accurate certification there is a unique threshold equilibrium. With noisy certification multiple equilibria may exist, but every such equilibrium yields lower profits than the accurate-certification benchmark. Propositions \ref{pro:D2RN} and \ref{pro:RankingRN} formalize these statements.

\begin{proposition}
Suppose there exists a lowest type $v_m$ such that $\sum_{j=1}^{m} q_j(v_m-v_j)-c\ge 0$. Let $(\Crn,\Drn)$ be such that $v_C=v_D=v_m$. Then $(\Crn,\Drn)$ is the unique equilibrium under accurate certification.
If no such $v_m$ exists, there is no equilibrium under accurate certification in which sellers certify.
\label{pro:D2RN}
\end{proposition}

Proposition \ref{pro:D2RN} characterizes the unique equilibrium under accurate certification. For intuition, consider a certification strategy (equilibrium or not) where all types certify, $v_C=v_1$, so buyers face no uncertainty. By Proposition \ref{pro:WTPs}, a type $v_k$ earns expected profit $\sum_{j=1}^{k} q_j(v_k-v_j)-c$. This is because the rival is type $v_j$ with probability $q_j$, and profit is $v_k-v_j$ if $j\le k$ and $0$ otherwise. One can show that a certifying type's expected profit does not depend on the certification threshold $v_C$. Hence, in equilibrium only types $v\ge v_m$ choose to certify.

\begin{proposition}
Consider an $\alpha$-equilibrium $(\Ca,\Da)$ with $\vCa=\vDa$. If type $v_m$ in Proposition \ref{pro:D2RN} exists, then $(\Crn,\Drn)$ is more profitable than $(\Ca,\Da)$.\footnote{The result also holds when $\vCa\ne\vDa$. Earlier drafts contained a proof; for brevity it is omitted here and provided in a separate (non-publication) appendix. Intuitively, equilibrium conditions rule out $\vCa>\vDa$ because observing $s=\vDa$ signals a certification error, inducing beliefs $E[v| v\in C]$ that are higher than $E[v|s=\vCa]$. While equilibria with $\vCa > \vDa$ are theoretically possible under certain off-equilibrium beliefs, such strategies yield weakly lower aggregate profits than disclosure equilibria, thus reinforcing Proposition \ref{pro:RankingRN}'s profitability ranking.}
\label{pro:RankingRN}
\end{proposition}

To understand how certification noise changes sellers' aggregate profits, consider an equilibrium with $v_{\Ca}=v_{\Da}=v_l$. Let $\pND$ be the unconditional probability of non-disclosure, and $\rho_i$ (for $i\ge l$) denote the unconditional probability of a seller disclosing $s=v_i$. The two sellers' joint expected profit (excluding certification costs) is
\begin{equation*}
\underbrace{\sum_{i,j=l}^n \rho_i \rho_j |p(v|s_i)-p(v|s_j)|}_{\mbox{\footnotesize both sellers disclose}} +\underbrace{2\sum_{j=l}^n \rho_j \pND(p(v|s_j)-p(v|ND))}_{\mbox{\footnotesize one seller discloses}},
\end{equation*}
where Proposition \ref{pro:WTPs} implies $\pi_i(s_i,s_j)+\pi_j(s_i,s_j)=|p(v|s_i)-p(v|s_j)|$, and
equilibrium requires $p(v|s_i)\ge p(v|ND)$ for all $s_i\in D$.

Lowering $\alpha$ changes joint profits through the same two forces identified earlier: blurring and differentiation. First, as $\alpha$ falls, certification becomes less informative and the WTP difference narrows as
\[
\lvert p(v | s_i) - p(v | s_j) \rvert
= \varphi(\alpha) \lvert v_i - v_j \rvert,
\qquad \text{ with } \qquad \varphi(0) = 0, \ \varphi(1) = 1, \ \varphi^\prime(\alpha) > 0,
\]
for any $s_i, s_j\in D$. This is the \textit{blurring effect}, which lowers the expected joint profit: by Proposition \ref{pro:WTPs}, smaller WTP differences lead to lower profits.

The second force can raise profits. Certification noise shifts the distribution of disclosed outcomes from $(\qND,q_l,\dots,q_n)$ to $(\rho_{ND},\rho_l,\dots,\rho_n)$, potentially increasing product differentiation.
For example, when certifying sellers have identical qualities ($v_i=v_j=v$), the only pricing subgame under perfect certification is $(s_i,s_j)=(v,v)$, where each seller earns zero profit. If $\alpha<1$, subgames with $s_i\ne s_j$ occur with positive probability, yielding positive expected profit. This positive \textit{differentiation effect} becomes stronger when the probability of identical quality types is high, especially when this probability is high for the highest-quality sellers. Under perfect certification, these sellers tie with each other and earn zero profit. Certification noise not only breaks such ties, allowing them to be perceived differently, but also preserves their advantage, as their distribution of observed outcomes still first-order stochastically dominates that of all lower-quality sellers.


Taken together, these forces move profits in opposite directions. With risk-neutral buyers, the blurring effect dominates: for any $\alpha<1$, expected profits are lower than under accurate certification (Proposition \ref{pro:RankingRN}).

\subsubsection{Loss-Averse Buyers: General Case}
\label{subsec:LAgeneral}

Assume buyers are loss averse with coefficient $b$.\footnote{Loss aversion is not strictly necessary for our results. In the supplementary appendix, which is not intended for publication, we consider the two-type setting from Section \ref{subsec:rankingLA} with risk-averse buyers (CARA utility) and show that noisy certification can also dominate accurate certification, but only under implausibly high risk aversion. However, it may require unrealistic values
of risk aversion. For the parameters in Example \ref{ex:same}, the absolute risk-aversion coefficient would need to exceed $0.797$; at that level, the certainty equivalent of a 50\% chance of winning \$5{,}000 or \$10{,}000 is \$5{,}000.87.} For any message $s_i$, buyers' WTP is
$p(v|s_i)=E_i(\alpha)+b\cdot EL_i(\alpha)$, where $E_i(\alpha)$ is expected quality conditional on $s_i$ and $EL_i(\alpha)$ is the expected-loss term defined in (\ref{eq:ELi}). We make the dependence on $\alpha$ explicit in the notations, since the focus of this subsection is on certification precision.

Let $(C,D)$ be such that $v_C=v_D=v_l$, and with a slight abuse of notation, let $ND<l<\dots<n$. As before, let $\rho_i$ and $\rho_{ND}$ denote the unconditional probabilities of a seller disclosing certification outcome $s=v_i$ and non-disclosing, respectively. The seller's expected profit is:
\begin{eqnarray*}
\Pi^\alpha&=&\sum_{s_i,s_j\in S: i<j} \rho_i\rho_j (p(v|s_j)-p(v|s_i))\\
&=&\underbrace{\sum_{s_i,s_j\in S: i<j}\rho_i\rho_j (E_j(\alpha)-E_i(\alpha)))}_{K_0(\alpha)}
+b \underbrace{\sum_{s_i,s_j\in S: i<j}\rho_i\rho_j (EL_j(\alpha)-EL_i(\alpha))}_{K_1(\alpha)}=K_0(\alpha)+bK_1(\alpha),
\end{eqnarray*}
which is linear in $b$. From Proposition \ref{pro:RankingRN},
we know that $K_0(\alpha)<K_0(1)$ when $0\le\alpha<1$. Thus, for inaccurate certification to be
more profitable, there must exist $\alpha$ such that $K_1(\alpha)>K_1(1)$ and $b$ is sufficiently large. A sufficient condition for the former is $\frac{\partial K_1(\alpha)}{\partial\alpha}\big|_{\alpha=1}<0$.

\begin{proposition}
Assume that the product qualities $\{v_i\}$ and their associated probabilities $\{q_i\}$ satisfy: (1) $v_{i+1}-v_i\ge v_i-v_{i-1}$ for $i\in\{2,\dots,n-1\}$;
and (2) $q_k>2(q_1+\dots+q_{k-1})$ for every $k\ge 2$. Let $(C,D)$ be such that
$v_C=v_D=v_l$, and assume that $(C,D)$ is an equilibrium in both the accurate and inaccurate certification settings. Then, for $\alpha$ close to $1$ and sufficiently large $b$, noisy certification is more profitable than accurate certification.
\label{pro:LAgeneral}
\end{proposition}

The intuition is as follows. Consider two families of subgames: $S_1$, in which both sellers disclose, and $S_2$, in which exactly one seller discloses. At $\alpha=1$, any disclosed signal pins down quality, so disclosed options carry no loss penalty and only $ND$ is uncertain: $EL_i(1)=0$ for any $i\ne ND$. As $\alpha$ falls slightly, $ND$ begins to pool some high types with low types. The reference point for $ND$ rises while the downside remains, so loss-averse buyers penalize $ND$ more heavily. Disclosed outcomes remain comparatively safe. The WTP gap between a disclosing and a non-disclosing seller therefore widens, which raises profits in $S_2$. This mirrors the intuition from Example \ref{ex:same}, now in the general case.

For $S_1$ subgames, as $\alpha$ falls slightly below $1$, noise harms lower signals more than the top signal. The highest disclosed message, $s_n$, is associated with a product that is almost certainly high quality, so its downside and loss penalty are relatively small. Lower disclosed messages face more meaningful downside relative to their own references, so loss-averse buyers discount them more. The WTP gap between $s_n$ and any lower disclosed message therefore widens compared to $\alpha=1$, which increases profits in $S_1$.

The condition in Proposition \ref{pro:LAgeneral} is sufficient but not necessary. It is deliberately strong so that it applies for any values of $l$ and $n$. In more specific settings, as in Proposition \ref{pro:2Tsame}, much weaker parameter restrictions can ensure that noisy certification is more profitable. Likewise, Example \ref{ex:same} shows that ``sufficiently high $b$'' need not be implausibly large and $\alpha$ need not be sufficiently close to $1$ for noisy certification to yield higher profits.

\section{Experimental Design and Hypotheses}
\label{sec:design}

\subsection{Experimental design}
\label{subsec:design}

The experimental environment mirrors the duopoly model in Section~\ref{sec:theory}. In each round, we used stranger matching to form four-person groups. Roles were randomly assigned so that two participants became sellers and two became buyers; groups and roles were redrawn every round, independently of past assignments. Each buyer had unit demand. Each seller offered two units of a product with quality unknown to buyers. A seller's product quality was independently drawn each round from a discrete uniform distribution with support $\{l,l+1,\dots,100\}$, where $l$ varied by treatment and was common knowledge.

Each round proceeded as follows. After the quality draw, if certification was available, each seller decided whether to seek certification. Sellers who certified paid cost $c$ and then observed the certification outcome. Next, they chose whether to disclose the certified value or, instead, to send a public cheap-talk message $m\in\{l,l+1,\dots,100\}$. Sellers who did not certify (because certification was unavailable or they opted out) sent a public cheap-talk message $m\in\{l,l+1,\dots,100\}$.\footnote{Although the theoretical framework does not model cheap talk, sellers often rely on soft claims alongside or in place of hard certification results. Moreover, experimental papers have demonstrated that cheap talk can be more informative than what the theoretical framework would predict. In the supplementary materials (not intended for publication) we show that, if cheap talk is added to Section~\ref{sec:theory}, then for any on-equilibrium cheap-talk messages $m$ and $m^\prime$ one obtains $p^{ct}(v|m)=p^{ct}(v|m^\prime)$, so the theoretical informativeness of cheap talk is limited. Empirically, cheap talk had no significant effect on profits in D1 and D2, but had a positive, 5\% significant effect in D3$\f$ and a positive, 10\% significant effect in D3$\e$.} Buyers could distinguish certification outcomes from cheap talk. After messages were observed, both sellers simultaneously set prices. Buyers then chose seller 1, seller 2, or no purchase. Finally, payoffs were displayed and true qualities revealed.

Let $c$ be the certification cost, $Price$ the transaction price, $v$ product quality, and $Sold$ the number of units sold by a seller. A seller's per-round profit was $Sold\cdot Price - \mathds{1}_{certified}\cdot c$, while the buyer's profit per round was $v - Price$ if he purchased the product and $0$ otherwise.

We implemented both between- and within-subject variation. Between subjects, we varied (i) the lower bound $l$ of the quality support and (ii) the certification environment. We used $l\in\{50,80\}$. When $l=80$, sellers' qualities are more likely to be close, which, as argued in Section~\ref{sec:theory}, tends to magnify the potential differentiation effect of message dispersion. Certification environments were D1 (no certification), D2 (accurate certification), and D3 (noisy certification), yielding six treatments: D1$\f$, D1$\e$, D2$\f$, D2$\e$, D3$\f$, and D3$\e$. Each subject participated in exactly one treatment. We use D$_l$ (where $l\in\{50,80\}$) to refer to all treatments with support $\{l,\dots,100\}$, and D$k$ (where $k\in\{1,2,3\}$), without a subscript, to refer to both Dk$\f$ and Dk$\e$.

Within certification treatments (D2 and D3), $c$ was drawn each round from $\{10,15,20,25\}$. In D3, precision $\alpha$ was drawn each round from $\{0.5,0.6,0.7,0.8,0.9\}$. The values of $c$ and $\alpha$ were common knowledge.\footnote{While $c$ and $\alpha$ played a limited role in our theoretical analysis, they directly influence certification and disclosure decisions, which
affects the blurring and differentiation effects discussed in Section~\ref{sec:theory}.
Consequently, we chose to investigate their impact in the experimental setting.} Draws were balanced so that, in D2, each $c$ value occurred equally often, and in D3 each $(c,\alpha)$ pair occurred equally often within a treatment. For example, both D3 treatments had 20 rounds, so each $(c,\alpha)$ pair appeared exactly once.

\subsection{Procedure}
\label{subsec:procedure}

The experiments were programmed and conducted in z-Tree (Fischbacher, 2007). A total of 152 students from Phenikaa University (Vietnam) participated in eight sessions conducted between 2021 and 2023. Each session included 12, 20, or 24 participants and lasted for 20 or 24 rounds, yielding 3{,}064 observations.\footnote{Treatment D$\f$ had only 17 rounds due to a computer crash in round 18.} Table \ref{tab:design2} reports the exact number of participants per treatment. We recruited more participants for the D3 treatments due to their greater complexity.

\begin{center}
[Place Table \ref{tab:design2} Here]
\end{center}

Recruitment was conducted via student associations and course-group announcements. All sessions took place in three computer labs at Phenikaa University. At the start of each session, participants received step-by-step instructions in Vietnamese (English translations are provided in Appendix \ref{appendix:instructions}). The experiment began with two non-incentivized practice rounds, followed by 20 or 24 incentivized rounds. Approximately half of the participants also completed an incentivized loss-aversion task at the end of the session.\footnote{Because individual loss aversion was not common knowledge in the market game, the elicitation was not intended to test equilibrium play. Rather, it provided a check that participants' loss-aversion levels were consistent with values reported in the literature.}

Participants' final earnings equaled the sum of their profits from all rounds, a 20{,}000 VND show-up fee (approximately \$0.86), and, when applicable, the payoff from one of the two loss-aversion questions (see Appendix \ref{appendix:LA}).\footnote{We paid based on the sum across all rounds because earnings can vary substantially by role and realized quality. This approach rewards effort rather than luck and follows prior studies (Cai and Wang, 2006; Jin et al., 2022). To limit potential endowment effects, cumulative earnings were not shown during the session.}

On average, participants earned 53{,}280 VND (approximately \$2.25). Sessions lasted about 75 minutes, so this compensation exceeded the typical hourly wage for student jobs in Vietnam (about 20{,}000 VND, approximately \$0.86).

\subsection{Hypotheses}
\label{subsec:hypotheses}

The theory in Section~\ref{sec:theory} delivers testable predictions for profits and competitive intensity through the interaction of blurring and differentiation effects. Without certification, buyers perceive offers as identical and Bertrand competition drives profits to zero (Corollary \ref{coro:D1}). With certification, accurate signals maximize informativeness and, under risk neutrality, yield higher profits than noisy signals (Proposition \ref{pro:RankingRN}). With loss aversion, however, a modest amount of certification noise can raise sellers' profits (Proposition \ref{pro:2Tsame}, Proposition \ref{pro:LAgeneral}, and Example \ref{ex:same}.)

\begin{assumptionp}{(D1$_\pi$)}
Sellers earn zero profit in the no-certification environment.
\label{hyp:D1}
\end{assumptionp}

\begin{assumptionp}{(D2$_\pi$-RN)}
With risk-neutral buyers, sellers' profits are lower in noisy certification (D3) than in accurate certification (D2).
\label{hyp:profit-RN}
\end{assumptionp}

\begin{assumptionp}{(D2$_\pi$-LA)}
With loss-averse buyers, sellers can earn higher profits in noisy certification (D3) than in accurate certification (D2).
\label{hyp:profit-LA}
\end{assumptionp}

The relative profitability of D2 and D3 may differ for two distinct reasons: \emph{certification expenditures} (if sellers certify at different rates, they pay different total fees) and \emph{competition intensity} (differentiation effect changes how intensely sellers compete). Because the theory does not pin down take-up differences between D2 and D3, we adopt a neutral benchmark:

\begin{assumptionp}{(SameCert)}
D2 and D3 sellers are equally likely to obtain certification.
\label{hyp:propensity}
\end{assumptionp}

Under \ref{hyp:propensity}, any D2--D3 profit difference reflects competition intensity rather than fee outlays. Accordingly, we also posit that the D2--D3 profit ranking should not depend on whether profits are computed before or after subtracting certification fees:

\begin{assumptionp}{(FeeNeutral)}
The D3 versus D2 profit comparison is unchanged whether profits are calculated before or after certification fees are deducted.
\label{hyp:expenditure}
\end{assumptionp}

Regarding the intensity of competition, the theoretical analysis and the differentiation effect imply that environments with noisy certification are less competitive. We assess competitiveness using two indicators: average posted prices and sellers' share of total welfare. In less competitive environments, both measures should be higher.

\begin{assumptionp}{(LowerCompD3)}
Sellers in D3 post higher prices than in D1 and D2, and the sellers' profit share in total welfare is higher in D3 than in D1 and D2.
\label{hyp:competitiveness}
\end{assumptionp}

Comparing D$\f$ and D$\e$, we expect D2$\e$ to be more competitive than D2$\f$. The narrower quality distribution in D2$\e$ makes rivals more likely to be similar in quality and to disclose it, which intensifies price competition. Conversely, we expect D3$\e$ to be less competitive than D3$\f$, because greater similarity in underlying quality strengthens the differentiation effect under noisy certification.

\begin{assumptionp}{(HigherCompD2$\e$)}
Prices and sellers' profit shares in D2$\e$ are lower than in D2$\f$.
\label{hyp:D2-80}
\end{assumptionp}
\vspace{-1cm}

\begin{assumptionp}{(LowerCompD3$\e$)}
Prices and sellers' profit shares in D3$\e$ are higher than in D3$\f$.
\label{hyp:D3-80}
\end{assumptionp}

\section{Experimental Results}
\label{sec:results}

\subsection{Sellers and Profitability Ranking}
\label{subsec:summarystat}

\subsubsection{Sellers' Profits}

Table \ref{tab:summaryD} presents summary statistics for key variables across the six treatments: D1$\f$, D1$\e$, D2$\f$, D2$\e$, D3$\f$, and D3$\e$. Figure \ref{fig:sellersprofit} shows sellers' profits in each treatment. Under risk-neutral theory, accurate certification should dominate noisy certification (Proposition \ref{pro:RankingRN}). Empirically, we observe the opposite: across both supports, D3 is more profitable than D2. D3$\f$ exceeds D2$\f$ (44.83 vs 42.61), and D3$\e$ exceeds D2$\e$ by a wide margin (64.06 vs 45.82). Our loss-averse framework is more consistent with these findings: as shown in Section \ref{sec:theory}, loss aversion can make D3 more profitable than D2, matching the observed higher profitability of D3$_l$ relative to D2$_l$.

\begin{center}
[Place Table \ref{tab:summaryD} here]

[Place Figure \ref{fig:sellersprofit} here]
\end{center}

We also observe a substantial increase in profit from D3$\f$ to D3$\e$ (44.83 to 64.06). This is consistent with the theoretical prediction that, under noisy certification, the differentiation effect is stronger when the quality support is narrower, that is, when qualities are more likely to be close. Notably, this increase is specific to noisy certification: profits in D1 and D2 are similar across supports, indicating that the D3$\e$ gain is driven by noise-induced differentiation rather than by a change in the underlying quality support. Finally, profits in D1 are strictly positive under both supports, contrary to the zero-profit benchmark in Corollary \ref{coro:D1}.\footnote{That a Bertrand-like environment, such as D1, does not necessarily yield the Bertrand outcome in the lab has been documented; see Dufwenberg and Gneezy (2000).}

\begin{description}
\item[Result \refstepcounter{ccc}\arabic{ccc}:] \textit{The profitability ranking among
D$\f$ treatments is D$2\f<$D$3\f<$D$1\f$. The profitability ranking among
D$\e$ treatments is D$2\e<$D$1\e<$D$3\e$. Hypotheses \ref{hyp:D1} and \ref{hyp:profit-RN} are  not supported. Hypothesis \ref{hyp:profit-LA} is consistent with the data.}
\end{description}

\subsubsection{Sellers' Profits across Treatments and Certification Costs}

Sellers' profit is determined by two key factors: certification fees and competition intensity. Examining the first factor, Table \ref{tab:certifytab} and Table \ref{tab:certifyreg} show certification take-up in D2 and D3 as a function of $c, \alpha,$ and $l$. The main result is that D2 sellers certify more often than D3 sellers across $c$ and $l$.
The difference is statistically significant in all but one case. Higher take-up in D2 depresses net profits through larger certification fee payments.

\begin{description}
\item[Result \refstepcounter{ccc}\arabic{ccc}:] \textit{$D2_l$ sellers are more likely to obtain certification than $D3_l$ sellers. Hypothesis \ref{hyp:propensity} is not supported.}
\end{description}

\begin{center}
[Place Table \ref{tab:certifytab} here]
\end{center}

\begin{center}
[Place Table \ref{tab:certifyreg} here]
\end{center}

Table \ref{tab:profits} reports sellers' profits by treatment, fee $c$, and
(for D3) accuracy $\alpha$. It includes both \textit{net profits} (after deducting certification costs) and \textit{gross profits} (before deducting costs).
Gross profits allow us to test whether the observed differences in net profitability between treatments stem from certification expenditures or from underlying differences in seller behavior and competition. 

\begin{center}
[Place Table \ref{tab:profits} here]
\end{center}

In terms of net profits, D3 consistently outperforms D2 across all values of $c$ and for both quality distributions ($l=50$ and $l=80$). This pattern holds for most values of $\alpha$ with
one exception. Although not all pairwise differences are statistically significant, D3 sellers typically earn more than D2 sellers. Similarly, D1 sellers (who have no access to certification) earn higher profits than those in D2 for most parameter values.

Turning to gross profits, a different picture emerges. For both $l$, D2 is more profitable than D1 once fees are removed, indicating that accurate certification can enhance profitability absent costs. Comparing D2 and D3, the pattern depends on the support. When $l=50$, D2$\f$ exceeds D3$\f$ in four parameter combinations ($c\in\{15,25\}$ and $\alpha\in\{0.5,0.6\}$), suggesting that accurate certification can yield higher gross profits than noisy certification when quality dispersion is wider. By contrast, when $l=80$, D3$\e$ strictly exceeds D2$\e$ in gross profits for all $c$ and $\alpha$. Thus, even before fees, D2$\e$ does not benefit from its higher certification propensity, pointing to more intense competition under accurate certification, rather than fee payments, as the main driver of D2$\e$'s lower profits.

\begin{description}
\item[Result \refstepcounter{ccc}\arabic{ccc}:] \textit{In terms of gross profits, D1 is the least profitable treatment.}

\item[Result \refstepcounter{ccc}\arabic{ccc}:] \textit{Comparing D2 and D3, certification costs are not the sole factor responsible for lower profitability of D2. In terms of gross profits, D2$\f$ can outperform D3$\f$ in gross profits but only when $\alpha$ is low. D3$\e$ consistently outperforms D2$\e$. Hypothesis \ref{hyp:expenditure} is not supported, except for D2$\f$ with low $\alpha$.}
\end{description}

\subsubsection{Sellers' Profits across Treatments and Competition Intensity}

Since differences in paid certification costs cannot fully explain the observed profitability ranking, we compare the competitiveness of the environments. As discussed in the theoretical section, noisy certification introduces endogenous product differentiation by adding randomness to buyers' WTP, which can reduce competitive pressure because greater dispersion in WTP allows sellers to charge higher prices and earn higher profits (see Proposition \ref{pro:WTPs}). To test this mechanism in our data, we use two measures of competitiveness: prices and the share of sellers' profits in total welfare, which we call the \textit{profit share}:
\begin{equation*}
\mbox{Profit Share} = \frac{\mbox{Sellers' Profit}}{\mbox{Sellers' Profit} + \mbox{Buyers' Profit}}.
\end{equation*}
Lower prices and lower profit shares reflect more intense competition, all else equal. Table \ref{tab:pricessales} and Table \ref{tab:share} present average prices and profit shares across treatments and parameter values. Figure \ref{fig:prices} displays average prices across the six treatments.

\begin{center}
[Place Tables \ref{tab:pricessales} and \ref{tab:share} here]

[Place Figure \ref{fig:prices} here]
\end{center}

Our first finding is that the treatment with noisy certification, D3, is consistently the least competitive. For both supports ($D\f$ and $D\e$), D3 yields the highest prices and profit shares for all $c$ and $\alpha$. The gap is especially pronounced when product qualities are tightly clustered ($D\e$), which is where the differentiation effect is strongest. A useful analogy is the ``Editor's Choice'' badge on tech review sites: in markets with many near-identical products (such as Bluetooth speakers), such a noisy but influential signal lets the highlighted product command a premium while nearly identical competitors without the badge are forced into fiercer price competition.

A second key finding is that the no-certification treatment (D1) is not always the most competitive. When $l=50$, D2$\f$ rather than D1$\f$ has the lowest prices and profit shares. This contrasts with the theoretical benchmark, which predicts that D1 is most competitive because, without certification, buyers regard the two products as identical, implying Bertrand pricing and zero profits (Corollary \ref{coro:D1}). In the laboratory, however, we cannot expect D1 buyers to hold identical beliefs. Without a public, verifiable signal, beliefs and hence WTPs differ across buyers, creating non-equilibrium product differentiation and breaking the Bertrand logic.\footnote{We call it non-equilibrium because, in equilibrium, buyers hold correct and therefore identical beliefs on the equilibrium path.} By contrast, in D2, once quality is disclosed it becomes common knowledge, aligning beliefs and intensifying price competition, which makes D2$\f$ the most competitive. In D3, although certification is available (as in D2),
disclosure does not make quality common knowledge, so belief heterogeneity persists and competition remains weaker than in D2.

\begin{description} \item[Result \refstepcounter{ccc}\arabic{ccc}:] \textit{For a given $l$, D3$_l$ is the least competitive treatment, as evidenced by higher prices and profit shares, confirming Hypothesis \ref{hyp:competitiveness}. At $l=50$, D2$\f$ is more competitive than D1$\f$, likely due to non-equilibrium belief heterogeneity in D1$\f$. At $l=80$, D1$\e$ is the most competitive treatment.}
\end{description}

We now compare treatments across quality supports, $l=50$ and $l=80$. Products in D$\e$ have a higher average quality of 90 than in D$\f$, which has an average quality of 75, so, all else equal, one would expect higher prices and profits in D$\e$. Strikingly, the opposite holds in D1 and D2. Prices and profit shares in D1$\e$ and D2$\e$ are consistently lower than in their D$\f$ counterparts, while profits barely change. This is evidence of D$\e$ environment being more competitive. The greater product homogeneity in D$\e$ compresses perceived differences between products and makes buyers less willing to pay a premium. As a result, D1$\e$ and D2$\e$ sellers cannot extract extra surplus from buyers despite offering better products.

This contrasts sharply with D3, where both prices and sellers' profits are higher in D3$\e$ than in D3$\f$. Noisy certification sustains product differentiation, softening competitive pressure and allowing sellers to earn higher profits.\footnote{A related empirical pattern appears in commercial real estate: buildings with environmental certifications such as LEED or ENERGY STAR are associated with rent and sale price premia relative to comparable noncertified buildings; see Eichholtz, Kok, and Quigley (2010) and Fuerst and McAllister (2011).} Its impact is strongest when quality dispersion is narrow ($l=80$), consistent with the differentiation effect discussed in Section \ref{sec:theory}.

Figure \ref{fig:profitbyquality} provides further evidence that noisy certification effectively softens competition in the environment with narrow quality dispersion (D$\e$). Typically, when certification is available, higher-quality sellers gain and lower-quality sellers lose relative to the no-certification benchmark; this pattern holds for D2$\f$ versus D1$\f$, for D3$\f$ versus D1$\f$, and for D2$\e$ versus D1$\e$. D3$\e$ is the sole exception: both higher- and lower-quality sellers earn more than in D1$\e$. Because weaker price competition raises profits for all sellers, including lower-quality sellers, these across-the-board gains in D3$\e$ indicate that certification noise makes the D$\e$ environment less competitive.

\begin{center}
[Place Figure \ref{fig:profitbyquality} here]
\end{center}

\begin{description} \item[Result \refstepcounter{ccc}\arabic{ccc}:] \textit{Despite higher average product quality (90 vs.\ 75), D$\e$ is more competitive than D$\f$ in D1 and D2: prices and profit shares are lower in D1$\e$ and D2$\e$ than in D1$\f$ and D2$\f$. This supports Hypothesis \ref{hyp:D2-80}.}

\item[Result \refstepcounter{ccc}\arabic{ccc}:] \textit{Noisy certification softens competition in $D\e$. Prices in D3$\e$ are higher than in D3$\f$, and D3$\e$ sellers earn more than D3$\f$ sellers. Hypothesis \ref{hyp:D3-80} is supported.}
\end{description}

\subsection{Buyers}
\label{subsec:buyers}

This section examines buyers' behavior along three dimensions: the optimality of individual choices, realized profits, and how certification disclosure affects purchasing.

\subsubsection{Optimality of Buyers' Decisions: Complete Information}

We begin by asking whether buyers' choices were individually optimal under complete information. Evaluating optimality under uncertainty is difficult because it depends on unobserved risk preferences, beliefs about quality, and expectations about seller behavior. We therefore restrict attention to cases with complete information about product quality.

The D2 treatment provides such cases. When both sellers obtain and disclose certification, buyers observe the qualities of both products and can compute consumer surplus for each option, $CS_i = v_i - p_i$. Optimal behavior is to purchase the product with the higher $CS$, provided the surplus is nonnegative.

Tables \ref{tab:rational1} and \ref{tab:Optimal} report behavior under complete information. Buyers rarely selected options with negative surplus, 6.1\% in D2$\f$ and 0\% in D2$\e$ (Table \ref{tab:rational1}). They chose the highest-surplus product in 87\% of cases in D2$\f$ and 98.1\% in D2$\e$ (Table \ref{tab:Optimal}). Thus, in the absence of uncertainty, buyers' decisions were largely optimal.

\begin{description}
\item[Result \refstepcounter{ccc}\arabic{ccc}:] \textit{When buyers faced no uncertainty, they made mostly optimal decisions.}
\end{description}

\begin{center}
[Place Table \ref{tab:rational1} here]
\end{center}

\begin{center}
[Place Table \ref{tab:Optimal} here]
\end{center}

\subsubsection{Buyers' profit}

Panels A and B of Table \ref{tab:buyersprofit} report average buyer profits across treatments and parameter values. The main takeaway is that buyers earn significantly more in D2$_l$ (accurate certification) than in D3$_l$ (noisy certification). Note that the profit measure in Table \ref{tab:buyersprofit} is \textit{ex post}: it records realized payoffs given choices and prices and therefore does not capture the \textit{ex ante} benefits of certification, such as reduced uncertainty and the associated utility gains. As a result, it likely understates the total value of accurate certification to buyers.

Table \ref{tab:buyersprofit} also shows that buyers in D$\e$ earn more than those in D$\f$, regardless of $c$ and $\alpha$, which is expected given the higher expected product value to be split between buyers and sellers in D$\e$. Within each treatment, buyer profits in D3 peak at $\alpha = 0.5$, although the relationship between $\alpha$ and profits is generally nonmonotonic in both D3$\f$ and D3$\e$.

\begin{description}
\item[Result \refstepcounter{ccc}\arabic{ccc}:] \textit{Buyer profits are highest under accurate certification and higher in $D\e$ than in $D\f$.}
\end{description}

\begin{center}
[Place Table \ref{tab:buyersprofit} here]
\end{center}

\subsubsection{Certification and Buyers' Purchasing Decisions.}

Table \ref{tab:buyers} examines how purchasing behavior varies with certification disclosure. In D2$\f$ and D3$\f$, when neither seller disclosed certification, purchase rates fell sharply, even below those in D1$\f$ where certification was unavailable. This indicates that non-disclosure, when certification is an option, is interpreted as a negative quality signal. In contrast, this pattern does not appear in D$\e$, where purchase rates remain consistently high (above 90\%) regardless of disclosure, likely reflecting the higher quality of D$\e$ products and the correspondingly lower risk of negative payoff.

Across D$_l$ treatments, buyers are most likely to purchase in D2$_l$, but only when at least one seller discloses certification, underscoring buyers' preference for certainty. Importantly, D2's higher purchase rates cannot be explained solely by lower prices: D1$\e$ buyers purchase less often than D2$\e$ buyers despite even lower average prices (see Table \ref{tab:pricessales}).

\begin{description}
\item[Result \refstepcounter{ccc}\arabic{ccc}:] \textit{If no seller discloses a certification outcome, purchase propensity in $D2\f$ and $D3\f$ falls below the $D1\f$ level.}
\end{description}

\begin{center}
[Place Table \ref{tab:buyers} here]
\end{center}

\subsection{Welfare}
\label{subsec:welfare}

We measure welfare as the average profit of all participants within each group (each group has two sellers and two buyers). This reflects a risk-neutral, ex post perspective and does not capture potential ex ante benefits, such as reduced uncertainty from certification. We report both net welfare, based on sellers' net profits, and gross welfare, based on sellers' gross profits.

Table \ref{tab:welfarecalpha} presents average net and gross welfare across $(c,\alpha)$ pairs. The main finding is that certification availability does not raise net welfare and often reduces it. In both supports, D3$_l$ has significantly lower net welfare than D1$_l$. In the narrow support, D2$\e$ also has significantly lower net welfare than D1$\e$, while in the wide support D2$\f$ shows only a small, statistically insignificant increase over D1$\f$.

Gross welfare paints a different picture for accurate certification: D2 raises gross welfare relative to D1 in both supports, consistent with more efficient allocation under accurate disclosure. By contrast, D3 does not increase gross welfare: in the wide support it is statistically similar to D1, and in the narrow support it is lower. Table \ref{tab:welfarecalpha} further shows that, in the narrow support, the gap between D2 and D3 narrows as $\alpha$ increases and becomes statistically indistinguishable at higher $\alpha$, indicating that greater precision mitigates the welfare penalty of noise.

The lack of net-welfare gains in D2 is naturally explained by certification fees. Once fees are excluded, gross welfare in D2 exceeds D1 in both supports. In D3, fees alone cannot account for the losses, since both net and gross welfare are below D1.

\begin{description}
\item[Result \refstepcounter{ccc}\arabic{ccc}:] \textit{Certification does not improve net welfare. Accurate certification increases gross welfare but fails to raise net welfare due to certification costs. Noisy certification reduces net welfare and does not improve gross welfare.}
\end{description}

\begin{center}
[Place Table \ref{tab:welfarecalpha} Here]
\end{center}

\subsection{Robustness Checks}
\label{subsec:robustness}

\indent \textbf{Subsample Analysis.} To test robustness, we replicate the analyses in subsections \ref{subsec:summarystat} through \ref{subsec:welfare} using only later periods of each treatment. Specifically, for each treatment we restrict the data to observations from period 11 onward, ensuring that all subjects had a ten-period learning phase. Table \ref{tab:robustnesscheck} reports three outcomes of interest: sellers' profit, buyers' profit, and group welfare.\footnote{See the supplementary materials for exact replications of the main tables and figures in subsections \ref{subsec:summarystat} through \ref{subsec:welfare} using this subsample.} Overall, Table \ref{tab:robustnesscheck} mirrors the main results: noisy certification reduces buyers' profits and welfare but can increase sellers' profits, and the treatment rankings and qualitative comparative statics are unchanged when we restrict to later periods.

\smallskip
\indent \textbf{Trends in Data.} Figure \ref{fig:data-trends} traces the dynamics of three variables: sellers' posted prices, sellers' profits, and buyers' profits. Prices and sellers' profits decline over time, while buyers' profits rise. The price decline is consistent with experimental evidence on convergence toward more competitive pricing in Bertrand environments (Abbink and Brandts, 2008; Buchheit and Feltovich, 2011). The increase in buyers' profits likely reflects both lower prices and learning about the environment.

\begin{center}
[Place Table \ref{tab:robustnesscheck} Here]
\end{center}

\begin{center}
[Place Figure \ref{fig:data-trends} Here]
\end{center}

\section{Conclusion}
\label{sec:conclusion}

This paper investigates how inaccurate certification, as a noisy quality signal, affects competition, pricing, and welfare. Our model and experiments reveal a clear pattern: in markets with close substitutes and loss-averse buyers, certification noise widens perceived quality dispersion, weakens price competition, and increases sellers' profits relative to accurate certification. These are markets where objective quality differences are small (e.g., bottled water, generic drugs, Bluetooth speakers, and other common electronics accessories), and imprecise signals create noise-induced differentiation among otherwise similar products. Our experiments confirm that when offerings are close substitutes, profits are higher under noisy certification than under accurate certification. In terms of welfare, accurate disclosure improves allocation and gross welfare, though certification fees can offset these gains. By contrast, noisy disclosure consistently reduces welfare, whether or not certification costs are included.

These findings extend beyond third-party certification to other noisy signals, such as platform-assigned badges like ``Amazon's Choice'' on Amazon, ``Top Rated Seller'' on eBay, and ``Superhost'' on Airbnb. These badges are algorithmically assigned and function as noisy proxies for quality rather than verified guarantees. Sellers can strategically pursue them, and earning one has an effect similar to a favorable certification: it tilts demand toward the tagged option among close substitutes, reduces fine-grained comparisons, and softens price competition.

The contribution of this paper is to analyze how imprecision in quality signals affects market outcomes. For scholars, we identify the conditions and mechanisms through which noise increases sellers' profits. For practitioners, we specify when noisy signals sustain margins versus when precise, common-knowledge disclosure compresses them. For policymakers, our experimental results show that imprecision can be privately profitable yet harmful to buyers and overall welfare, highlighting precision, disclosure, and fees as key regulatory levers to protect consumers and improve allocation.

Several questions remain for future research. First, endogenizing certifiers' policies, including thresholds, fees, and disclosure rules, would refine the model's implications (Dranove and Jin 2010). Second, while assuming homogeneous loss aversion is common in the literature (e.g., De and Nabar 1991; Strausz 2005; Marinovic et al. 2018), relaxing this assumption could uncover richer heterogeneity in market responses. Third, field experiments that vary precision, fees, and disclosure defaults across categories with different degrees of substitutability would shed light on the conditions under which noise alters margins and welfare. Finally, because ratings, badges, and other platform signals operate as noisy disclosure under platform governance, a model that allows for strategic certifiers and basic platform objectives would provide a useful bridge across the literatures on disclosure, certification, and marketplace design.

\section*{References}
\begin{hangparas}{.9cm}{1}
\linespread{1}\selectfont
\setlength{\parskip}{.3cm}

Abbink, K., and Brandts, J. (2008). Pricing in Bertrand competition with increasing marginal costs. \emph{Games and Economic Behavior}, \emph{63(1)}, 1-31.



Anderson, L. R., Freeborn, B. A., and Hulbert, J. P. (2012). Risk aversion and tacit collusion in a Bertrand duopoly experiment. \emph{Review of Industrial Organization}, \emph{40}, 37-50.

Beaver, W. H., Shakespeare, C., and Soliman, M. T. (2006). Differential Properties in the Ratings of Certified Versus Non-Certified Bond-Rating Agencies. \textit{Journal of Accounting and Economics}, \textit{42}(3), 303-334.

Becker, G. S., and Milbourn, T. T. (2011). How Did Increased Competition Affect Credit Ratings? \textit{Journal of Financial Economics}, \textit{101}(3), 493-514.

Benabou, R., and Laroque, G. (1992). Using Privileged Information to Manipulate Markets: Insiders, Gurus, and Credibility. \textit{The Quarterly Journal of Economics}, \textit{107}(3), 921-958.

Bloomberg (2008). Bringing Down Wall Street as Ratings Let Loose Subprime Scourge. \emph{Bloomberg, September, 24}.


Bottega, L. and Freitas, J. D. (2019). Imperfect certification in a Bertrand duopoly. \emph{Economics
Letters}, \emph{178}, 33-36.


Buchheit, S., and Feltovich, N. (2011). Experimental Evidence of a Sunk-Cost Paradox: A Study of Pricing Behavior in Bertrand–Edgeworth Duopoly.  \emph{International Economic Review},  \emph{52}(2), 317-347.

Cai, H., and Wang, J. T. Y. (2006). Overcommunication in strategic
information transmission games. \emph{Games and Economic Behavior},
\emph{56}(1), 7-36.

Cain, D. M., Loewenstein, G., and Moore, D. A. (2005). The Dirt on Coming Clean: Perverse Effects of Disclosing Conflicts of Interest. \textit{The Journal of Legal Studies}, \textit{34}(1), 1-25.

Card, D., and DellaVigna, S. (2013). Nine facts about top journals in economics. \emph{Journal of Economic Literature}, \emph{51}(1), 144--161.

Chen, L. and Lee, H. L. (2017). Sourcing Under Supplier Responsibility Risk: The Effects of
Certification, Audit, and Contingency Payment. \emph{Management Science}, \emph{63(9)}, 2795–2812.

Chen, Y., Narasimhan, C., and Zhang, Z. J. (2001). Individual marketing with imperfect targetability.
\emph{Marketing Science}, \emph{20}(1), 23–41.

Chen, Y. J., and Deng, M. (2013). Supplier certification and quality
investment in supply chains. \emph{Naval Research Logistics (NRL)},
\emph{60}(3), 175--189. https://doi.org/10.1002/NAV.21527


De, S. and Nabar, P. (1991). Economic implications of imperfect quality certification. \emph{Economics
Letters}, \emph{37(4)}, 333–37.


Dranove, D., and Jin, G. Z. (2010). Quality Disclosure and Certification:
Theory and Practice. \emph{Journal of Economic Literature},
\emph{48}(4), 935--963. https://doi.org/10.1257/JEL.48.4.935

Dufwenberg, M., and Gneezy, U. (2000). Price competition and market concentration: an experimental study. \emph{International Journal of Industrial Organization}, \emph{18(1)}, 7-22.



Eichholtz, P., Kok, N., and Quigley, J. M. (2010). Doing well by doing good? Green office buildings. \textit{American Economic Review}, 100(5), 2492-2509.

Elfenbein, D. W., Fisman, R., and McManus, B. (2015). Market structure, reputation, and the value of quality certification.  \emph{American Economic Journal: Microeconomics},  \emph{7(4)}, 83-108.

Feltovich, N. (2019). The interaction between competition and unethical behaviour. \emph{Experimental Economics}, \emph{22(1)}, 101-130.

Feinstein, J. S. (1989). The safety regulation of US nuclear power plants: Violations, inspections, and abnormal occurrences. \emph{Journal of Political Economy}, \textit{97}(1), 115-154.

Fischbacher, U. (2007). Z-Tree: Zurich toolbox for ready-made economic
experiments. \emph{Experimental Economics}, \emph{10}(2), 171--178.
https://doi.org/10.1007/s10683-006-9159-4

Fuerst, F., and McAllister, P. (2011). Green noise or green value? Measuring the effects of environmental certification on office values. \textit{Real estate economics}, 39(1), 45-69.

Gul, F. (1991). A theory of disappointment aversion. \emph{Econometrica}, \emph{59}(3), 667--686.


Hong, H., and Kubik, J. D. (2003). Analyzing the Analysts: Career
Concerns and Biased Earnings Forecasts. \emph{The Journal of Finance},
\emph{58}(1), 313--351.

Hu, N., Bose, I., Koh, N. S., and Liu, L. (2012). Manipulation of online reviews: An analysis of ratings, readability, and sentiments. \emph{Decision support systems}, \textit{52}(3), 674-684.

Huh, S., Shapiro, D. A., and Sung Ham (2023).
Profitability Of Noisy Certification In The Presence Of Loss-Averse Buyers.
\emph{Journal of Industrial Economics}, \emph{71}(3), 770--813.

Hui, X., Saeedi, M., Spagnolo, G., and Tadelis, S. (2023). Raising the bar: Certification thresholds and market outcomes. \emph{American Economic Journal: Microeconomics}, \emph{15(2)}, 599-626.


Jin, G. Z., and Leslie, P. (2003). The effect of information on product quality: Evidence from restaurant hygiene grade cards.  \emph{The Quarterly Journal of Economics},  \emph{118(2)}, 409-451.

Jin, G. Z., Luca, M., and Martin, D. (2021). Is No News (Perceived As)
Bad News? An Experimental Investigation of Information Disclosure.
\emph{American Economic Journal: Microeconomics}, \emph{13}(2),
141--173. https://doi.org/10.1257/MIC.20180217

Jin, G. Z., Luca, M., and Martin, D. (2022). Complex disclosure.  \emph{Management Science},  \emph{68}(5), 3236-3261.

Jovanovic, B. (1982). Truthful Disclosure of Information.  \emph{Bell Journal of Economics}, \emph{13}(1),


Kalayci, K. (2015). Price complexity and buyer confusion in markets. \emph{Journal of Economic Behavior \& Organization}, \emph{111}, 154-168.

Karik\'{o}, K., Buckstein, M., Ni, H., and Weissman, D. (2005). Suppression of RNA recognition by Toll-like receptors: the impact of nucleoside modification and the evolutionary origin of RNA. \textit{Immunity}, 23(2), 165-175.

Klemperer, P. (1987). The competitiveness of markets with switching costs. \emph{The RAND Journal of
Economics}, 138-150.

Koszegi, B., and Rabin, M. (2006). A model of reference-dependent preferences. \emph{The Quarterly Journal of Economics}, \emph{121}(4), 1133-1165.

Kov\'{a}cs, B., Lehman, D. W., and Carroll, G. R. (2020). Grade inflation in restaurant hygiene inspections: Repeated interactions between inspectors and restaurateurs. \emph{Food Policy}, \emph{97}, 101960.


Levin, D., Peck, J., and Ye, L. (2009). Quality Disclosure and Competition. \emph{The Journal of Industrial Economics}, \emph{57}(1), 167--196. doi:10.1111/j.1467-6451.2009.00366.x

Lim, T., Kothari, S. P., Stein, J., Stulz, R., Wang, J., and Womack, K.
(2001). Rationality and Analysts' Forecast Bias. \emph{The Journal of
Finance}, \emph{56}(1), 369--385.

Intermediaries. \emph{The RAND Journal of Economics}, \emph{30}(2), 214.
https://doi.org/10.2307/2556078


Marinovic, I., Skrzypacz, A., and Varas, F. (2018). Dynamic certification and reputation for quality. \emph{American Economic Journal: Microeconomics}, \emph{10}(2), 58-82.

Masatlioglu, Y., and Raymond, C. (2016). A Behavioral Analysis of
Stochastic Reference Dependence. \emph{American Economic Review},
\emph{106}(9), 2760--2782.







Peyrache, E., and Quesada, L. (2011). Intermediaries, credibility and incentives to collude. \emph{Journal of Economics \& Management Strategy}, \emph{20(4)}, 1099-1133.

Pollrich, M., and Wagner, L. (2016). Imprecise information disclosure and truthful certification. \emph{European Economic Review}, \emph{89}, 345-360.

Resnick, P., and Zeckhauser, R. (2002). Trust Among Strangers in Internet Transactions: Empirical Analysis of eBay's Reputation System. In M. R. Baye (Ed.), The Economics of the Internet and E-Commerce (pp. 127-157). Emerald Group Publishing Limited.

Sheth, J. D. (2021). Disclosure of information under competition: An experimental study. \emph{Games and Economic Behavior}, \emph{129}, 158-180.

Shin, J. and Yu, J. (2021). Targeted advertising and consumer inference. \emph{Marketing Science}, \emph{40}(5), 900-922.

Shin, J. and Sudhir, K. (2010). A customer management dilemma: When is it profitable to reward
one's own customers? \emph{Marketing Science}, \emph{29}(4), 671-689.

Stahl, K., and Strausz, R. (2017). Certification and market transparency. \emph{The Review of Economic Studies}, \emph{84(4)}, 1842-1868.

Strausz, R. (2005). Honest certification and the threat of capture.
\emph{International Journal of Industrial Organization},
\emph{23}(1--2), 45--62. https://doi.org/10.1016/j.ijindorg.2004.09.002

Suzumura, K. (1990). Strategic information revelation. \emph{Review of
Economic Studies}, \emph{57}(1), 25--47. https://doi.org/10.2307/2297541

Tversky, A., and Kahneman, D. (1992). Advances in prospect theory:
Cumulative representation of uncertainty. \emph{Journal of Risk and Uncertainty}, \emph{5}, 297–323.

Van Der Schaar, Mihaela, and Simpson Zhang. (2015). A dynamic model of certification and reputation. \emph{Economic Theory}, \emph{58(3)}, 509–41.

Villas-Boas, J. M. (1999). Dynamic competition with customer recognition. \emph{The Rand Journal of Economics}, \emph{30(4)}, 604-631.

Wang, J. T. Y., Spezio, M., and Camerer, C. F. (2010). Pinocchio's
Pupil: Using Eyetracking and Pupil Dilation to Understand Truth Telling
and Deception in Sender-Receiver Games. \emph{American Economic Review},
\emph{100}(3), 984--1007.

Wang, M., Rieger, M. O., and Hens, T. (2017). The impact of culture on loss aversion. \emph{Journal of Behavioral Decision Making}, \emph{30(2)}, 270-281.

Wired (2019) What Does It Mean When a Product Is 'Amazon's Choice'? \textit{Wired, June 4, 2019}.

Zhang, J., and Li, K. J. (2020). Quality Disclosure Under Consumer Loss
Aversion. \emph{Management Science}, \emph{67(8)}, 4643-5300.
https://doi.org/10.1287/mnsc.2020.3745

\end{hangparas}

\newpage
\section{Figures}

\begin{figure}[htbp!]
\centering
\caption{Sellers' Profits}
\includegraphics[scale=0.8]{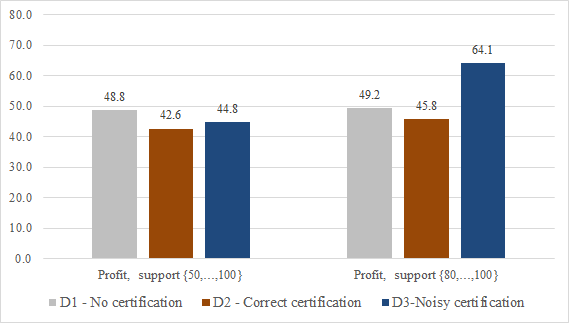}
\label{fig:sellersprofit}
    \captionsetup{font=it,singlelinecheck=off}
\end{figure}

\begin{figure}[!]
\centering
\caption{Prices}
\includegraphics[scale=0.8]{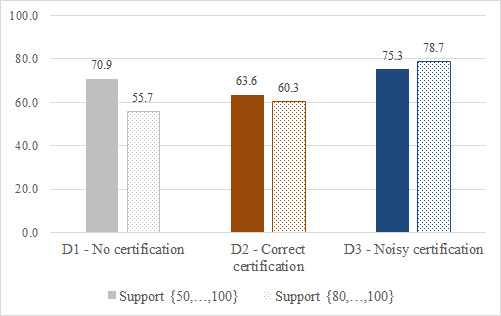}
\label{fig:prices}
    \captionsetup{font=it,singlelinecheck=off}
\end{figure}

\newpage
\begin{figure}[htbp!]
\centering
\caption{Sellers' net profit by quality}
\includegraphics[scale=0.8]{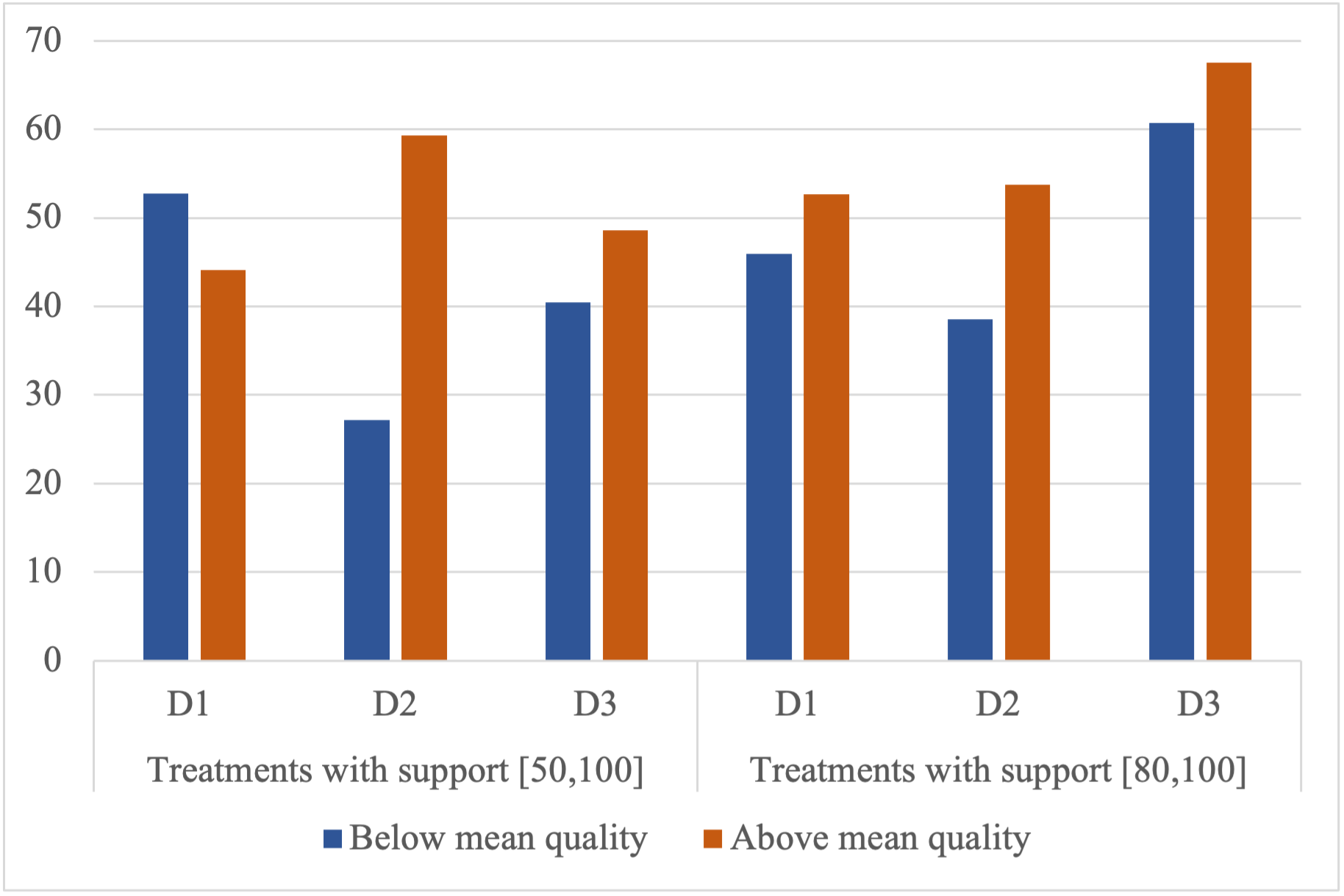}
\label{fig:profitbyquality}
    \captionsetup{font=it,singlelinecheck=off}
\end{figure}

\begin{figure}[ht!]
  \caption{Trends in data}
  \label{fig:data-trends}

  \centering

  \begin{subfigure}[b]{0.45\textwidth}
    \centering
    \includegraphics[width=\linewidth]{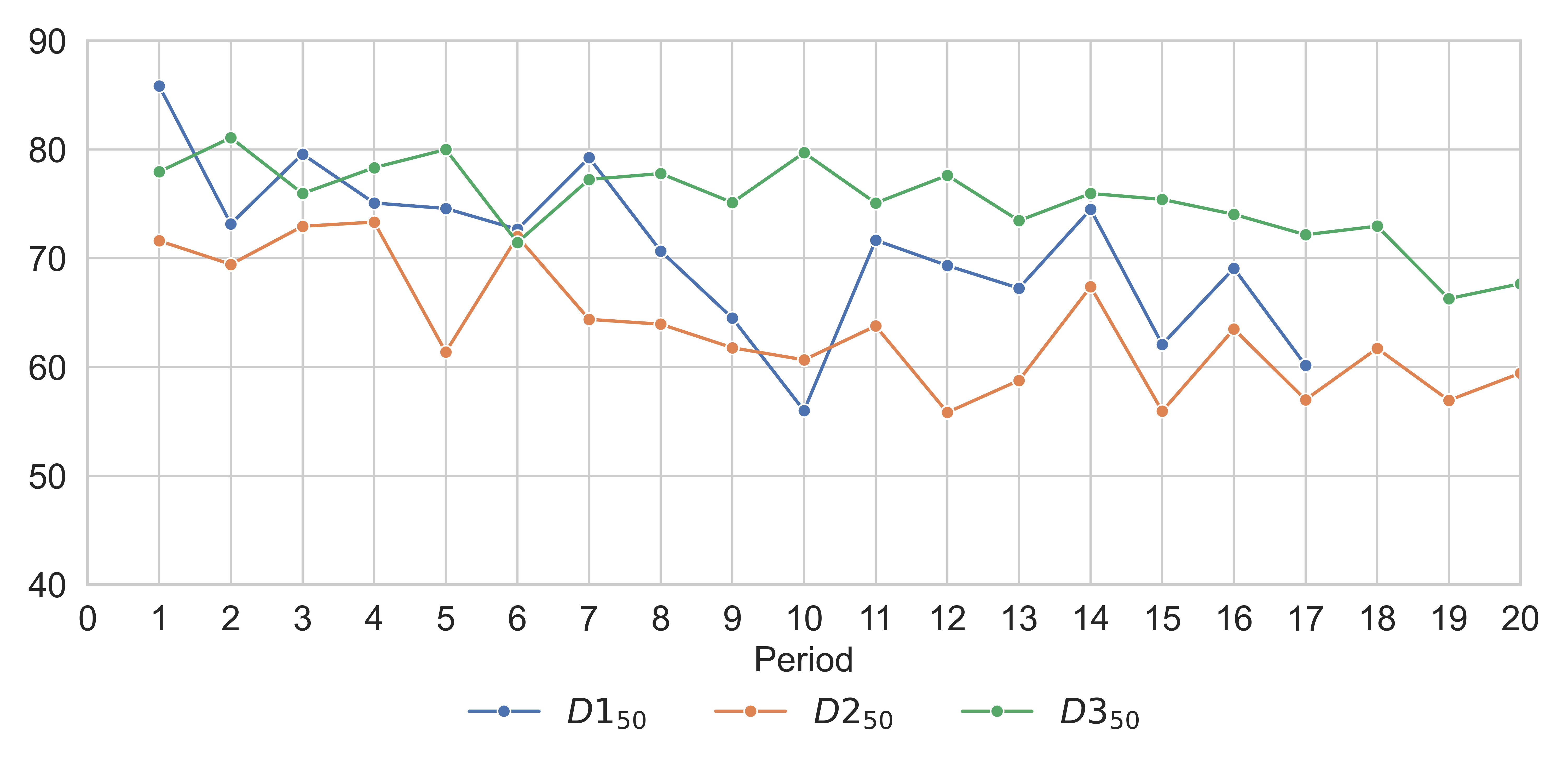}
    \caption{Sellers' price in \(D_{50}\)}
    \label{fig:SellerPrice50}
  \end{subfigure}%
  \hfill
  \begin{subfigure}[b]{0.45\textwidth}
    \centering
    \includegraphics[width=\linewidth]{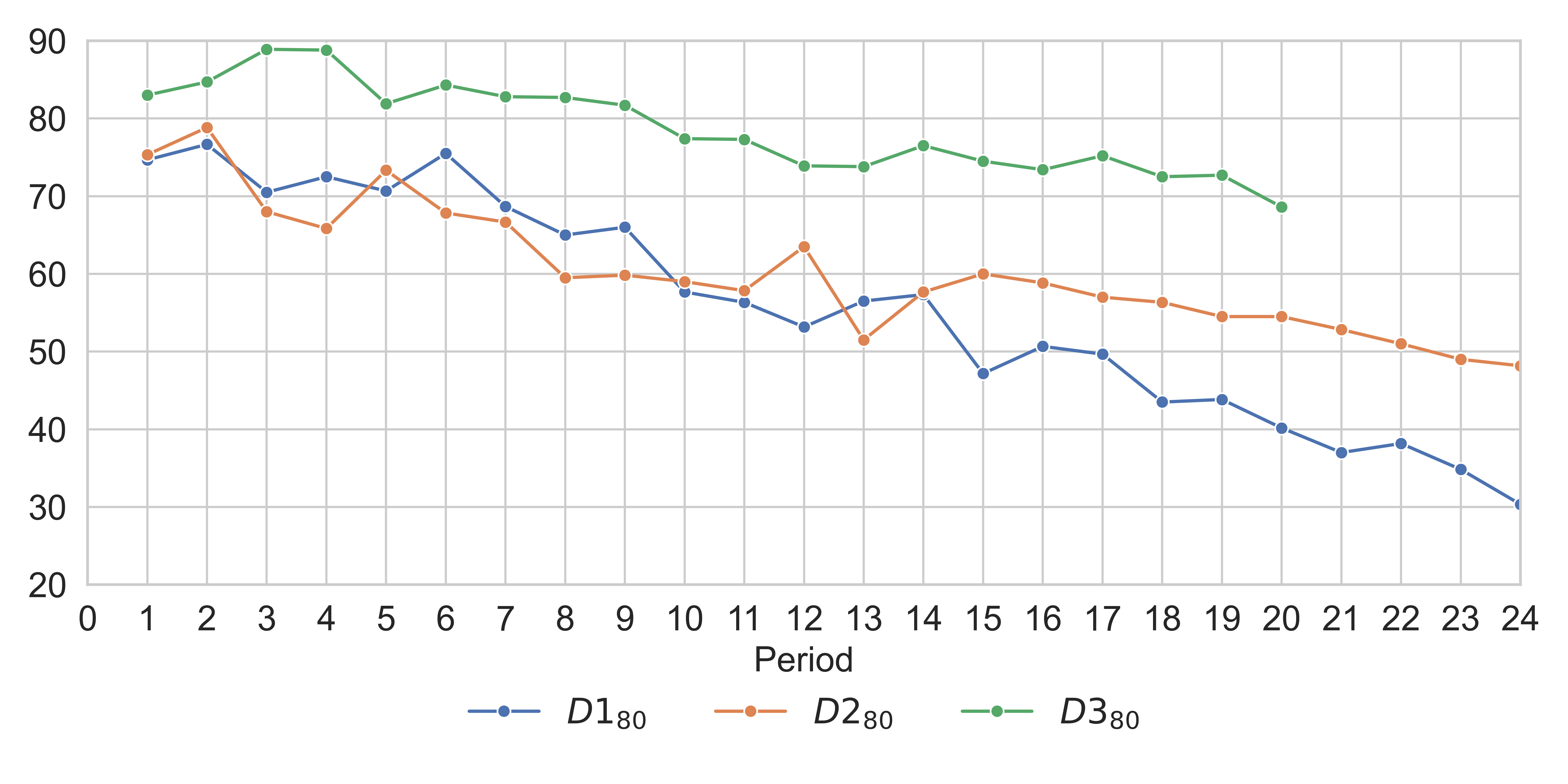}
    \caption{Sellers' price in \(D_{80}\)}
    \label{fig:SellerPrice80}
  \end{subfigure}

  \vspace{2em}

  \begin{subfigure}[b]{0.45\textwidth}
    \centering
    \includegraphics[width=\linewidth]{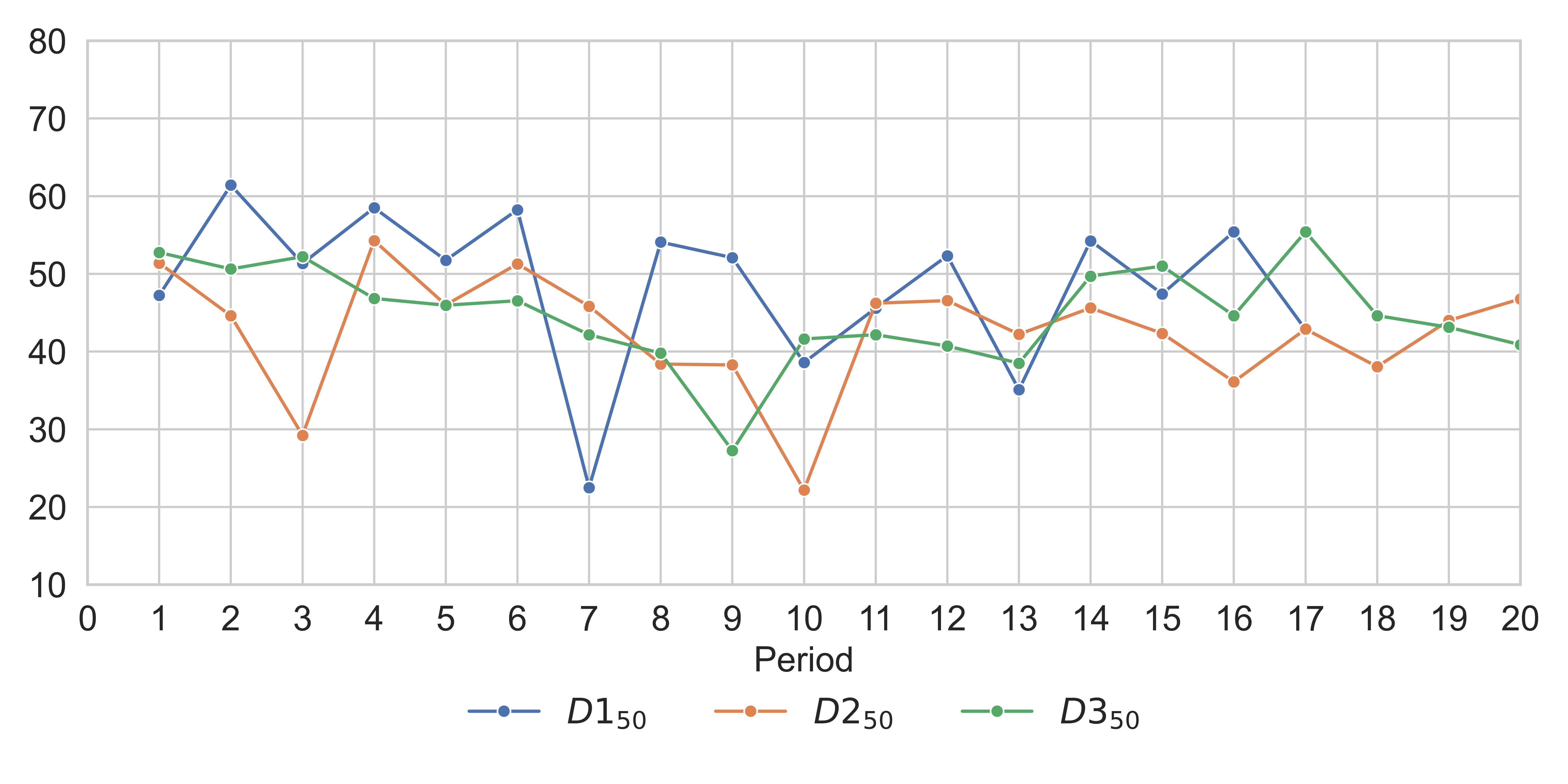}
    \caption{Sellers' profit in \(D_{50}\)}
    \label{fig:SellerProfit50}
  \end{subfigure}%
  \hfill
  \begin{subfigure}[b]{0.45\textwidth}
    \centering
    \includegraphics[width=\linewidth]{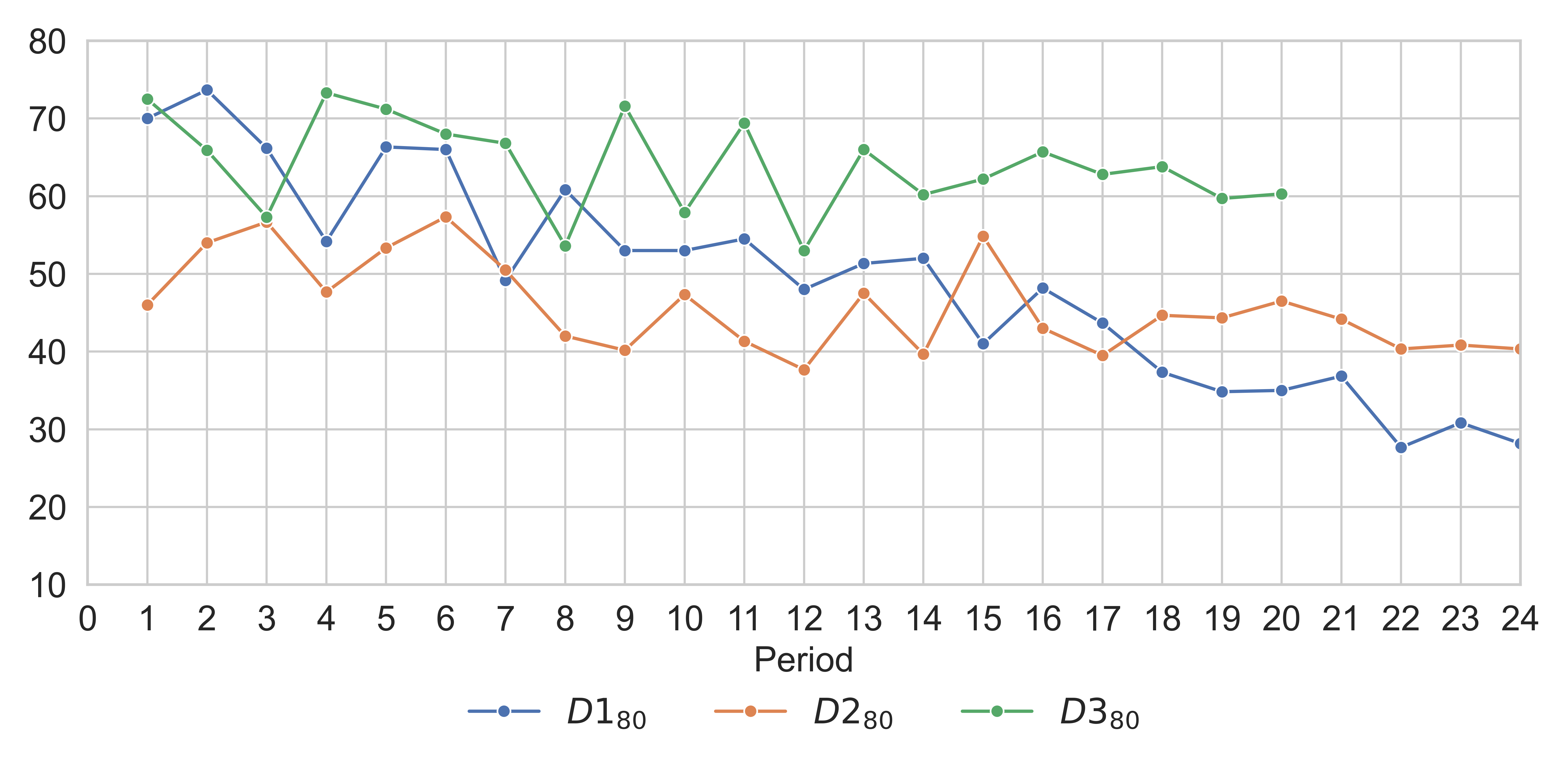}
    \caption{Sellers' profit in \(D_{80}\)}
    \label{fig:SellerProfit80}
  \end{subfigure}

  \vspace{2em}

  \begin{subfigure}[b]{0.45\textwidth}
    \centering
    \includegraphics[width=\linewidth]{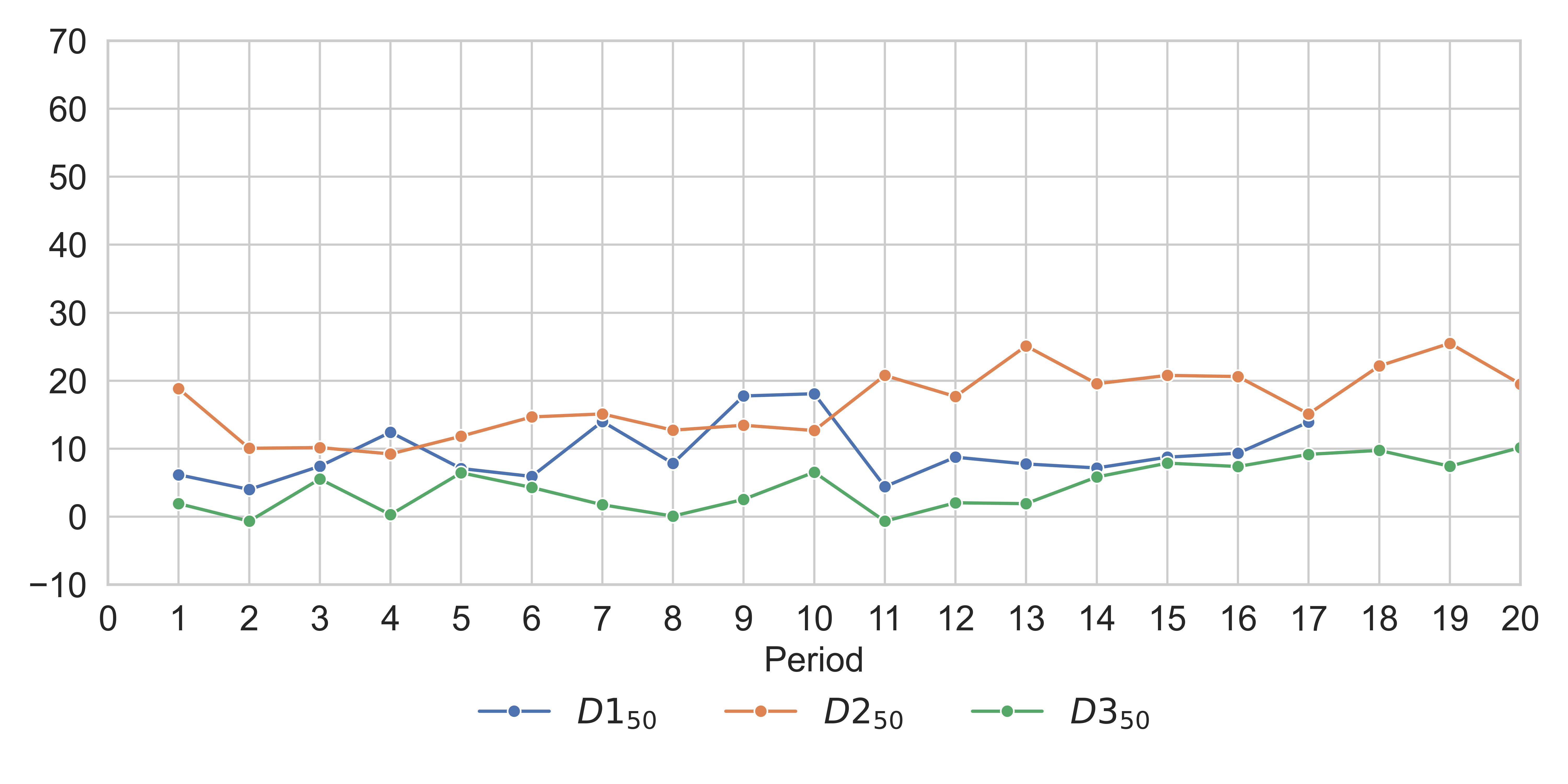}
    \caption{Buyers' profit in \(D_{50}\)}
    \label{fig:BuyerProfit50}
  \end{subfigure}%
  \hfill
  \begin{subfigure}[b]{0.45\textwidth}
    \centering
    \includegraphics[width=\linewidth]{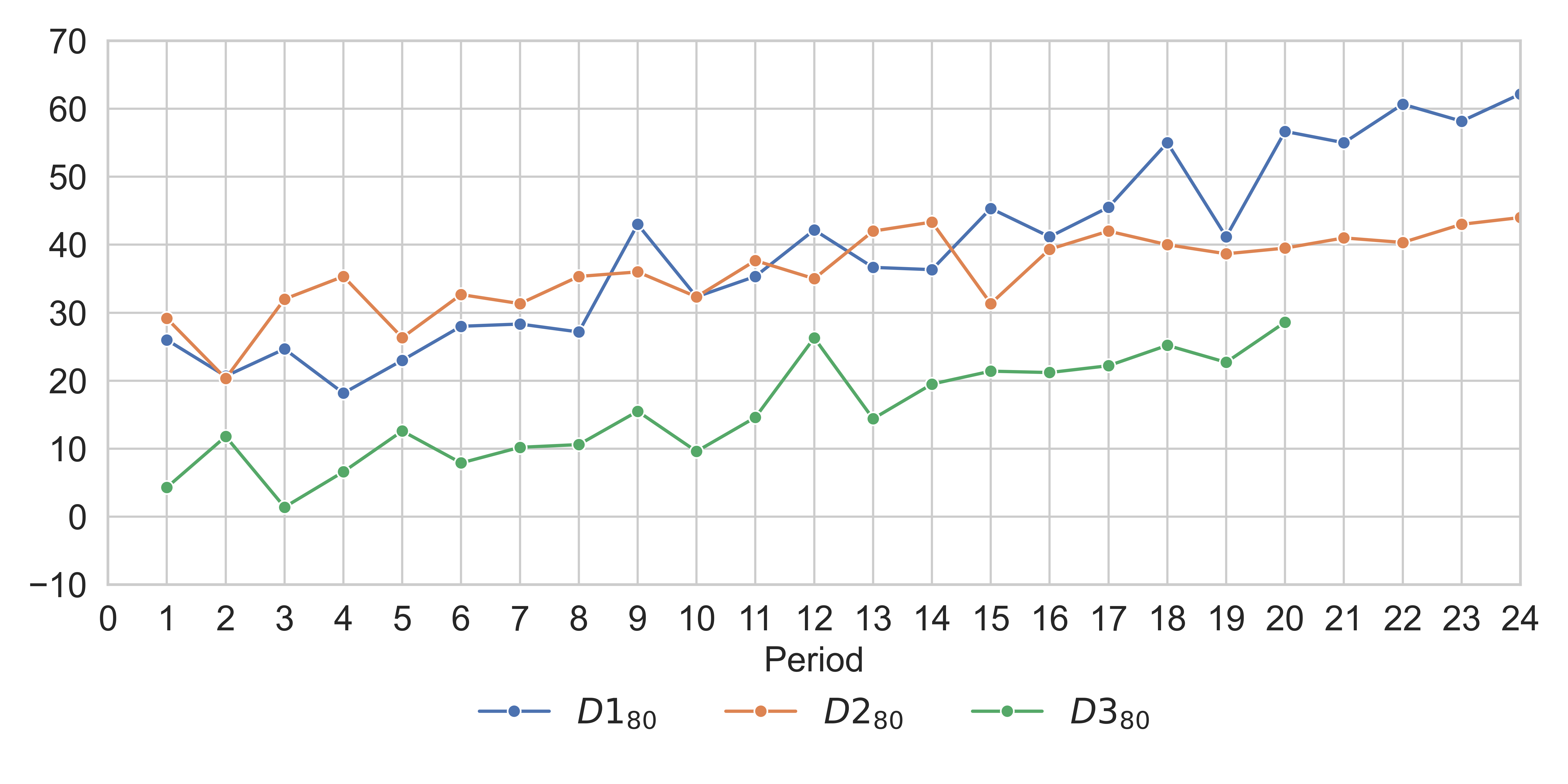}
    \caption{Buyers' profit in \(D_{80}\)}
    \label{fig:BuyerProfit80}
  \end{subfigure}
\end{figure}

\newpage
\section*{\vspace{2cm}\quad}

\section{Tables}
\begin{table}[htbp]
\fontsize{9pt}{11pt}\linespread{1.2}\selectfont
  \centering
  \caption{Number of participants and observations in each treatment}
\begin{tabular}{lcc|cc|cc}
\hline
\hline
\multicolumn{1}{c}{} &\mycomment{ \multicolumn{3}{c}{\textbf{Monopoly}} &       &} \multicolumn{6}{c}{\textbf{Treatments}} \mycomment{ &       & \multicolumn{3}{c}{\textbf{Differences}}} \bigstrut\\

\cline{2-7}    & \multicolumn{1}{c}{\textbf{D1}$\f$} & \textbf{D1}$\e$  & \textbf{D2}$\f$ & \textbf{D2}$\e$    & \textbf{D3}$\f$ & \textbf{D3}$\e$ \bigstrut\\
\hline
\textbf{No. of rounds} &       17 & 24     & 20 & 24     & 20 & 20 \bigstrut[t]\\
\textbf{No. of participants} & 24 & 12   & 36 & 12 &    48 & 20 \bigstrut[b]\\
\textbf{No. of observations} & 408 & 288   &  720 & 288   &  960 & 400 \bigstrut[b]\\
\hline
\hline
\end{tabular}%
\label{tab:design2}
\begin{flushleft}
\emph{Notes}: Data for treatment D1$\f$ is limited to the first 17 rounds due to
a computer crash during round 18.
\end{flushleft}
\end{table}

\mycomment{
\begin{table}[htbp]
\fontsize{9pt}{11pt}\linespread{1.2}\selectfont
  \centering
  \caption{Demographics}
\begin{tabular}{llcc}
\hline
\hline
      &       & \textbf{N} & \textbf{\%} \bigstrut\\
\hline
\multicolumn{1}{l}{\textbf{Major}} & Economics or business & 128   & 90.14\% \bigstrut[t]\\
      & Information technology & 5     & 3.52\% \\
      & Language & 6     & 4.23\% \\
      & Others & 3     & 2.11\% \\
\multicolumn{1}{l}{\textbf{Year}} & Freshman & 140   & 98.59\% \\
      & Sophomore and above & 2     & 1.41\% \\
\multicolumn{1}{l}{\textbf{Gender}} & Male  & 41    & 28.87\% \\
      & Female & 101   & 71.13\% \bigstrut[b]\\
\hline
\hline
\end{tabular}%
\label{tab:design3}
\end{table}
}

\begin{table}[h!]
\fontsize{9pt}{11pt}\linespread{1.2}\selectfont
  \centering
  \caption{Summary statistics }
\begin{tabular}{cl|rr|rr|rr|rr|rr|rr}
\hline
\hline
\multicolumn{2}{l}{\textbf{Treatments}} & \multicolumn{2}{c}{\textbf{Mean}} & \multicolumn{2}{c}{\textbf{Median}} & \multicolumn{2}{c}{\textbf{SD}} & \multicolumn{2}{c}{\textbf{Min}}& \multicolumn{2}{c}{\textbf{Max}} &
\multicolumn{2}{c}{\textbf{N}} \bigstrut\\
&\multicolumn{1}{c}{} &
\multicolumn{1}{c}{\textbf{D$_{50}$}} & \multicolumn{1}{c}{\textbf{D$_{80}$}} & \multicolumn{1}{c}{\textbf{D$_{50}$}} & \multicolumn{1}{c}{\textbf{D$_{80}$}} & \multicolumn{1}{c}{\textbf{D$_{50}$}} & \multicolumn{1}{c}{\textbf{D$_{80}$}} & \multicolumn{1}{c}{\textbf{D$_{50}$}} & \multicolumn{1}{c}{\textbf{D$_{80}$}} & \multicolumn{1}{c}{\textbf{D$_{50}$}} & \multicolumn{1}{c}{\textbf{D$_{80}$}} & \multicolumn{1}{c}{\textbf{D$_{50}$}} & \multicolumn{1}{c}{\textbf{D$_{80}$}} \\
\hline
\multicolumn{2}{l|}{\textit{D1 -- No certification}} &     &       &       &       &       &   &&&&&& \bigstrut[t]\\
      & Product quality &  74.92&90.15   &  74.5&90   &  15.00&5.78     &  50&80     &  100&100    &  204&144  \\
      & Announced quality &  87.78&96.45   &  90&99     &  9.43&4.40    &  55&80     &  100&100    &  204&144  \\
      & Offered price &  70.91&55.69   &  70&55     &  21.80&17.97   &  18&25     &  195&99    & 204&144  \\
      & Purchasing likelihood &  0.79&0.98   &  1&1      &  0.41&0.14   &  0&0      & 1&1      &  204&144  \\
      & Sellers' profit &  48.75&49.24   &  50&50     &  50.41&47.82   &  0&0      &  200&164    & 204&144  \\
      & Buyers' profit &  9.46&39.28    &  3&40      &  21.44&16.68   &  -36&-5    &  79&71     & 204&144  \bigstrut[b]\\
\hline
\cline{1-14}
\multicolumn{2}{l|}{\textit{D2 -- Accurate certification}} &       &       &       &       &       &  &&&&&&\bigstrut[t]\\
      & Product quality &  74.39&90.25   &  74.5&90   &  14.90&5.91   &  50&80     & 100 & 100   & 360&144  \\
      & Probability of certifying &  0.55&0.56   &  1&1      &  0.50&0.50  &  0&0      &  1&1      &  360&144  \\
      & Quality | certifying &  77.96&92.25   &  79&93.5     &  14.09&5.29   &  50&80     & 100&100    &  197&80  \\
      & Quality | not certifying &  70.07&87.75   &  68&87     &  14.74&5.74   &  50&80     & 100&99    &  163&64  \\
      & Offered price &  63.59&60.28   &  64&59  &  16.90&13.33   &  30&25     &  130&96    & 360&144  \\
      & Purchasing likelihood &  0.89&1   &  1&1      &  0.32&0   &  0&0      &  1&1      &  360&144  \\
      & Certification expenses & 9.26 &  9.06 & 10 &  10&9.42 & 9.03 & 0& 0 & 25 & 25 &360 & 144  \bigstrut[b]\\
      & Sellers' profit &  42.61&45.82   &  40&49.5   &  49.80&53.13   &  -25&-25    &  170&150    &  360&144  \\
      & Buyers' profit &  16.78&36.17   &  19.5&39   &  13.39&11.98     &  -34&7    &  46&60     & 360&144 \\
\hline
\cline{1-14}
\multicolumn{2}{l|}{\textit{D3 -- Noisy certification}} &       &       &       &       &       &  &&&&&&\bigstrut[t]\\
      & Product quality &  76.35&90.26   &  77&90     &  14.97&5.94     &  50&80     &  100&100    &  480&200 \\
      & Probability of certifying &  0.45&0.27  &  0&0      &  0.50&0.45    &  0&0      &  1&1      &  480&200  \\
      & Quality | certifying &  80.19&91.69  &  82&92     &  13.49&5.58   &  50&80     &  100&100    &  214&54  \\
      & Quality | not certifying &  73.26&89.73   &  72.5&90   &  15.41&6.00   &  50&80     & 100&100    &  266&146  \\
      & Disclosure | certifying &  0.86&0.93   &  1&1      &  0.35&0.26   &  0&0      &  1&1      & 214&54  \\
      & Disclosed outcome | certifying & 83.67 & 91.86   & 87  & 92     & 12.62 & 5.34  & 52 &  80  &100  & 100  & 183 &50  \\
      & Offered price &  75.27&78.73   &  76&80     &  15.67&11.19   &  34&50     &  123&100    & 480&200  \\
      & Purchasing likelihood &  0.74&0.93  &  1&1      &  0.44&0.26  &  0&0      &  1&1      &  480&200  \\
      & Certification expenses & 7.49 & 4.35 & 0 & 0 & 9.16 & 7.67 & 0& 0 & 25 & 25&480  & 200  \bigstrut[b]\\
      & Sellers' profit &  44.83&64.06   &  46&70     &  52.99&61.13   &  -25&-20    &  190&196   & 480&200  \\
      & Buyers' profit &  4.48&15.33    &  0&14    &  13.94&11.29   &  -38&-13  &  47&40     & 480&200  \\
\hline
\hline
\end{tabular}%
\label{tab:summaryD}
\begin{flushleft}
\end{flushleft}
\end{table}

\newpage
\begin{table}[htbp]
\fontsize{9pt}{11pt}\linespread{1.2}\selectfont
  \centering
  \caption{Sellers' propensity to obtain product certification}
\begin{tabular}{p{3.715em}ccccccccccc}
\hline
\hline
\multicolumn{12}{c}{\textbf{Panel A. D$\f$}}\\
\hline
\multicolumn{1}{r}{} & \multirow{2}[4]{*}{\textbf{Mean}} & \multicolumn{4}{c}{\textbf{Certification cost}} &       & \multicolumn{5}{c}{\textbf{Certification precision}} \bigstrut\\
\cline{3-6}\cline{8-12}\multicolumn{1}{r}{} &       & \textbf{10} & \textbf{15} & \textbf{20} & \textbf{25} &       & \textbf{0.5} & \textbf{0.6} & \textbf{0.7} & \textbf{0.8} & \textbf{0.9} \bigstrut\\

\cline{3-6}\cline{8-12}D2    & 0.55  & 0.63  & 0.58   & 0.48   & 0.50  &       &       &       &       &       & \multicolumn{1}{c}{} \bigstrut[t]\\
D3    & 0.45  & 0.54   & 0.43 & 0.45  & 0.37  &       & 0.32   & 0.38  & 0.54  & 0.47  & 0.52 \\
\cline{2-6}\cline{8-12}
D2-D3 &  ***    & *   & **   & -    & **     &       &       &       &       &       & \multicolumn{1}{c}{} \bigstrut[b]\\

\hline
\hline
\multicolumn{12}{c}{\textbf{Panel B. D$\e$}}\\
\hline
\multicolumn{1}{r}{} & \multirow{2}[4]{*}{\textbf{Mean}} & \multicolumn{4}{c}{\textbf{Certification cost}} &       & \multicolumn{5}{c}{\textbf{Certification precision}} \bigstrut\\
\cline{3-6}\cline{8-12}\multicolumn{1}{r}{} &       & \textbf{10} & \textbf{15} & \textbf{20} & \textbf{25} &       & \textbf{0.5} & \textbf{0.6} & \textbf{0.7} & \textbf{0.8} & \textbf{0.9} \bigstrut\\

\cline{3-6}\cline{8-12}D2    & 0.56  & 0.67  & 0.64   & 0.58   & 0.33  &       &       &       &       &       & \multicolumn{1}{c}{} \bigstrut[t]\\
D3    & 0.27  & 0.32   & 0.38  & 0.20  & 0.18  &       & 0.18   & 0.28  & 0.23  & 0.30  & 0.38 \\
\cline{2-6}\cline{8-12}
D2-D3 &  ***    & ***   & ***   & ***    & *    &       &       &       &       &       & \multicolumn{1}{c}{} \bigstrut[b]\\
\hline
\hline
\end{tabular}%
\label{tab:certifytab}
\begin{flushleft}
\emph{Notes:}  Asterisks indicate significance levels (* p\textless0.1, ** p\textless0.05, *** p\textless0.01), based on one-sided $t$-test with unequal variances. 
\end{flushleft}
\end{table}

\begin{table}[h!]
\fontsize{9pt}{11pt}\linespread{1.2}\selectfont
  \centering
  \caption{Sellers' propensity to obtain product certification. Regression analysis}
\begin{tabular}{llllll}
\hline
\hline
      & \multicolumn{2}{c}{\textbf{D$\f$}} & &\multicolumn{2}{c}{\textbf{D$\e$}}\bigstrut\\
\cline{2-3} \cline {5-6}& \multicolumn{1}{p{4.215em}}{\textbf{D2}} & \multicolumn{1}{p{4.215em}}{\textbf{D3}}& &\multicolumn{1}{p{4.215em}}{\textbf{D2}} & \multicolumn{1}{p{4.215em}}{\textbf{D3}} \bigstrut\\
\hline
True quality & \multicolumn{1}{p{4.215em}}{0.054***} & \multicolumn{1}{p{4.215em}}{0.059***} &&\multicolumn{1}{p{4.215em}}{0.26***} & \multicolumn{1}{p{4.215em}}{0.09***} \bigstrut[t]\\
      & \multicolumn{1}{p{4.215em}}{(5.60)} & \multicolumn{1}{p{4.215em}}{(6.22)} & &\multicolumn{1}{p{4.215em}}{(4.63)} & \multicolumn{1}{p{4.215em}}{(2.47)} \\
Certification cost & \multicolumn{1}{p{4.215em}}{-0.059**} & \multicolumn{1}{p{4.215em}}{-0.084***} & &\multicolumn{1}{p{4.215em}}{-0.163**} & \multicolumn{1}{p{4.215em}}{-0.103***} \\
      & \multicolumn{1}{p{4.215em}}{(-2.53)} & \multicolumn{1}{p{4.215em}}{(-3.65)}& &\multicolumn{1}{p{4.215em}}{(-3.47)} & \multicolumn{1}{p{4.215em}}{(-2.76)} \\
Certification precision &             & \multicolumn{1}{p{4.215em}}{4.39***} &             && \multicolumn{1}{p{4.215em}}{3.38***} \\      & & \multicolumn{1}{p{4.215em}}{(4.61)} & & & \multicolumn{1}{p{4.215em}}{(2.25)} \bigstrut[b]\\
\cline{2-3} \cline{5-6}No. of observations (N) & \multicolumn{1}{p{4.215em}}{342}   & 430 & &\multicolumn{1}{p{4.215em}}{144}   & 152 \bigstrut[t]\\
log-likelihood & \multicolumn{1}{p{4.215em}}{-131.3} & \multicolumn{1}{p{4.215em}}{-140.3} & &\multicolumn{1}{p{4.215em}}{-40.80} & \multicolumn{1}{p{4.215em}}{--55.41} \bigstrut[b]\\
\hline
\hline
\end{tabular}%
\label{tab:certifyreg}
\begin{flushleft}
\emph{Notes}: Asterisks indicate significance levels (* p\textless0.1, ** p\textless0.05, *** p\textless0.01), with $t$-statistics shown in parentheses. A fixed-effect logistic model (xtlogit) is used. 
Displayed estimates are odds ratios.
\end{flushleft}
\end{table}


\newpage
\begin{table}[htbp]
\fontsize{9pt}{11pt}\linespread{1.2}\selectfont
  \centering
  \caption{Sellers' profit by certification cost and precision}
\begin{tabular}{p{12.715em}ccccccccccc}
\hline
\hline
\multicolumn{1}{r}{} & \multirow{2}[4]{*}{\textbf{Mean}} & \multicolumn{4}{c}{\textbf{Certification cost}} &       & \multicolumn{5}{c}{\textbf{Certification precision}} \bigstrut\\
\cline{3-6}\cline{8-12}\multicolumn{1}{r}{} &       & \textbf{10} & \textbf{15} & \textbf{20} & \textbf{25} &       & \textbf{0.5} & \textbf{0.6} & \textbf{0.7} & \textbf{0.8} & \textbf{0.9} \bigstrut\\

\hline
\cline{1-12}

\multicolumn{1}{l}{\textbf{Panel A. Net Profit. D$\f$}}\\
D1 & 48.75\\
D2    & 42.61  & 44.22  & 45.54   & 42.14   & 38.53  &       &       &       &       &       & \multicolumn{1}{c}{} \bigstrut[t]\\
D3    & 44.83  & 45.38   & 46.34  & 46.96  & 40.63  &       & 41.58   & 44.70  & 46.59  & 45.75  & 45.50 \\
\cline{2-6}\cline{8-12}
D1-D2 &  *    &    &    &    &    &       &       &       &       &       & \multicolumn{1}{c}{} \bigstrut[b]\\
D1-D3 &  -    &    &    &   &    &       &       &       &       &       & \multicolumn{1}{c}{} \bigstrut[b]\\
D2-D3 &  -    & -   & -   & -    & -    &       &       &       &       &       & \multicolumn{1}{c}{} \bigstrut[b]\\

\hline
\cline{1-12}

\multicolumn{1}{l}{\textbf{Panel B. Net Profit. D$\e$}}\\
D1& 49.24\\
D2  & 45.82  & 50.19  & 45.33   & 44.53  & 43.22  &       &       &       &       &       & \multicolumn{1}{c}{} \bigstrut[t]\\
D3  & 64.06  & 65.76   & 60.54  & 62.94  & 67.00  &       & 57.72   & 61.50  & 69.10  & 65.13  & 66.85 \\
\cline{2-6}\cline{8-12}
D1-D2 &  -    &    &    &    &    &       &       &       &       &       & \multicolumn{1}{c}{} \bigstrut[b]\\
D1-D3 &  ***    &    &    &   &    &       &       &       &       &       & \multicolumn{1}{c}{} \bigstrut[b]\\
D2-D3 &  ***    & -   & -   & *    & **    &       &       &       &       &       & \multicolumn{1}{c}{} \bigstrut[b]\\
\hline
\cline{1-12}

\multicolumn{1}{r}{\textbf{Panel C. Gross Profit. D$\f$}}\\
D1 & 48.75\\
D2   & 51.88  & 50.56  & 54.21   & 51.70   & 51.03  &       &       &       &       &       & \multicolumn{1}{c}{} \bigstrut[t]\\
D3   & 52.31  & 50.79   & 52.72  & 55.96  & 49.79  &       & 46.95   & 50.69  & 55.50  & 53.72  & 54.72 \\
\cline{2-6}\cline{8-12}
D1-D2 &  -    &    &    &    &    &       &       &       &       &       & \multicolumn{1}{c}{} \bigstrut[b]\\
D1-D3 &  -    &    &    &   &    &       &       &       &       &       & \multicolumn{1}{c}{} \bigstrut[b]\\
D2-D3 &  -    & -   & -   & -    & -    &       &       &       &       &       & \multicolumn{1}{c}{} \bigstrut[b]\\

\hline
\cline{1-12}
\multicolumn{1}{l}{\textbf{Panel D. Gross Profit. D$\e$}}\\
D1 & 49.24\\
D2 & 54.88  & 56.86  & 54.91   & 56.19  & 51.56  &       &       &       &       &       & \multicolumn{1}{c}{} \bigstrut[t]\\
D3  & 68.41  & 68.96   & 66.24  & 66.94  & 71.50  &       & 60.73   & 65.38  & 72.85  & 70.25  & 72.85 \\
\cline{2-6}\cline{8-12}
D1-D2 &  -    &    &    &    &    &       &       &       &       &       & \multicolumn{1}{c}{} \bigstrut[b]\\
D1-D3 &  ***   &    &    &   &    &       &       &       &       &       & \multicolumn{1}{c}{} \bigstrut[b]\\
D2-D3 &  ***    & -   & -   & *    & **    &       &       &       &       &       & \multicolumn{1}{c}{} \bigstrut[b]\\
\hline
\hline
\end{tabular}%
\label{tab:profits}
\begin{flushleft}
\emph{Notes:} Asterisks indicate significance levels (* p\textless0.1, ** p\textless0.05, *** p\textless0.01), based on one-sided $t$-test with unequal variances.  Non-significant results are marked with a dash ("-"). Sellers' gross profits exclude certification costs. Net profits are calculated by subtracting certification costs from gross profits.
\end{flushleft}
\end{table}

\newpage
\begin{table}[htbp]
\fontsize{9pt}{11pt}\linespread{1.2}\selectfont
  \centering
  \caption{Prices by certification cost and precision}
\begin{tabular}{p{12.715em}ccccccccccc}
\hline
\hline
\multicolumn{1}{r}{} & \multirow{2}[4]{*}{\textbf{Mean}} & \multicolumn{4}{c}{\textbf{Certification cost}} &       & \multicolumn{5}{c}{\textbf{Certification precision}} \bigstrut\\
\cline{3-6}\cline{8-12}\multicolumn{1}{r}{} &       & \textbf{10} & \textbf{15} & \textbf{20} & \textbf{25} &       & \textbf{0.5} & \textbf{0.6} & \textbf{0.7} & \textbf{0.8} & \textbf{0.9} \bigstrut\\

\hline
\cline{1-12}

\multicolumn{1}{l}{\textbf{Panel A. Prices. D$\f$}}\\
D1 & 70.91 \\
D2    & 63.59   &  66.23 &  61.20  &   65.70&  61.23 &       &       &       &       &       & \multicolumn{1}{c}{} \bigstrut[t]\\
D3    & 75.27   & 76.11    & 73.20 & 77.22 &  74.54 &       & 72.82   & 75.72 & 76.63  & 76.08   &  75.08 \\
\cline{2-6}\cline{8-12}
D2-D3 & ***    & ***  & ***  &  ***  &  ***  &   &       &       &       &       & \multicolumn{1}{c}{} \bigstrut[b]\\

\hline
\cline{1-12}

\multicolumn{1}{l}{\textbf{Panel B. Prices. D$\e$}}\\
D1& 55.69 \\
D2  &  60.28 &   61.19 &  59.72&  63.00&  57.22 &       &       &       &       &       & \multicolumn{1}{c}{} \bigstrut[t]\\
D3  &   78.73 &  80.86 &   74.70 &  80.86 &  78.50 &       & 73.15    & 80.58   & 81.78  & 79.53  & 78.63 \\
\cline{2-6}\cline{8-12}
D2-D3 & ***  &  *** &  *** & *** & ***  &  &       &       &       &       & \multicolumn{1}{c}{} \bigstrut[b]\\
%
%
\hline
\hline
\end{tabular}%
\label{tab:pricessales}
\begin{flushleft}
\emph{Notes:} Asterisks indicate significance levels (* p\textless0.1, ** p\textless0.05, *** p\textless0.01), based on one-sided $t$-test with unequal variances.
\end{flushleft}
\end{table}

\newpage
\begin{table}[htbp]
\fontsize{9pt}{11pt}\linespread{1.2}\selectfont
  \centering
  \caption{Sellers' profit shares by cost and precision}
\begin{tabular}{p{12.715em}ccccccccccc}
\hline
\hline
\multicolumn{1}{r}{} & \multirow{2}[4]{*}{\textbf{Mean}} & \multicolumn{4}{c}{\textbf{Certification cost}} &       & \multicolumn{5}{c}{\textbf{Certification precision}} \bigstrut\\
\cline{3-6}\cline{8-12}\multicolumn{1}{r}{} &       & \textbf{10} & \textbf{15} & \textbf{20} & \textbf{25} &       & \textbf{0.5} & \textbf{0.6} & \textbf{0.7} & \textbf{0.8} & \textbf{0.9} \bigstrut\\

\hline
\cline{1-12}

\multicolumn{1}{l}{\textbf{Panel A. Net profit. D$\f$}} \\
D1 & 0.87\\
D2    & 0.73  & 0.77   &  0.74 &  0.75   & 0.68  &       &       &       &       &       & \multicolumn{1}{c}{} \bigstrut[t]\\
D3    &  1.08 &  0.98 & 0.93   & 0.97 & 1.46  &       &  0.93  & 0.98  & 0.96  & 1.62  & 0.92\\
\cline{2-6}\cline{8-12}
D2-D3 &  ***  & ***  & *** & ***  &  ****  &       &       &       &       &       & \multicolumn{1}{c}{} \bigstrut[b]\\

\hline
\cline{1-12}

\multicolumn{1}{l}{\textbf{Panel B. Net profit. D$\e$}} \\
D1& 0.56 \\
D2  & 0.56  & 0.59 & 0.55  &  0.57 &  0.53&       &       &       &       &       & \multicolumn{1}{c}{} \bigstrut[t]\\
D3  & 0.81  & 0.86 & 0.74 &  0.85 &  0.80 &       & 0.74   & 0.85 &  0.86 & 0.79  & 0.83 \\
\cline{2-6}\cline{8-12}
D2-D3 & ***   & *** &  *** &  ***  & *** &       &       &       &       &       & \multicolumn{1}{c}{} \bigstrut[b]\\
%
%
%
\hline
\hline
\end{tabular}%
\label{tab:share}
\begin{flushleft}
\emph{Notes:} Asterisks indicate significance levels (* p\textless0.1, ** p\textless0.05, *** p\textless0.01), based on one-sided $t$-test with unequal variances. Sellers' profit shares can exceed 1 if buyers' profits are negative.
\end{flushleft}
\end{table}

\newpage
\begin{table}[ht!]
\fontsize{9pt}{11pt}\linespread{1.2}\selectfont
\centering
\caption{\textbf{Rationality of buyers' participation under complete information}}
\begin{tabular}{lccccc}
\hline
\hline
& \multicolumn{2}{c}{\textbf{D2$\f$ treatment}} & & \multicolumn{2}{c}{\textbf{D2$\e$ treatment}} \\
& Rational & Irrational  & & Rational & Irrational \\
\hline
$\max\{CS_i,CS_j\}<0$ & 6 & 2 &&0 &0 \\
$CS_i> 0>CS_j$ & 29 & 3 &&2 &0\\
$\min\{CS_i,CS_j\}>0$ & 57 & 1 && 50 &0 \\
\hline
All situations &92 & 6 &&52 &0 \\
Percentages &$93.9\%$& $6.1\%$& & $100\%$ &$0\%$ \\
\hline
\hline
\end{tabular}
\label{tab:rational1}
 \begin{flushleft}
 \emph{Notes}: Consumer surplus ($CS_i$) is defined as the difference between a product's value ($v_i$) and its price ($p_i$). This analysis focuses solely on situations where both D2 sellers have certified, enabling buyers to determine the consumer surpluses associated with each seller. Rational participation is defined as buying a product with positive $CS$ if available, and not buying if no product has positive $CS$. In the first row, rational participation dictates not buying any product. In the second row, it means buying from seller $i$, and in the third row, it means buying from either seller.
\end{flushleft}
\end{table}

\begin{table}[ht!]
\fontsize{9pt}{11pt}\linespread{1.2}\selectfont
\centering
\caption{\textbf{Optimality of buyers' purchasing behavior under complete information}}
\begin{tabular}{lccccc}
\hline
\hline
& \multicolumn{2}{c}{\textbf{D2$\f$ treatment}} & & \multicolumn{2}{c}{\textbf{D2$\e$ treatment}} \\
& Optimal & Not optimal  & & Optimal & Not optimal \\
\hline
$\max\{CS_i,CS_j\}<0$ & 6 & 2 && 0&0 \\
$\max\{CS_i,CS_j\}>0$ & 81 & 11 &&51 &1 \\
\hline
All situations &87 & 13 &&51 &1 \\
Percentages &$87\%$& $13\%$& & $98.1\%$ &$1.9\%$ \\
\hline
\hline
\end{tabular}
\label{tab:Optimal}
 \begin{flushleft} \emph{Notes}: Consumer surplus ($CS_i$) is defined as the difference between a product's value ($v_i$) and its price ($p_i$). This analysis focuses solely on situations where both D2 sellers have certified, enabling buyers to determine the consumer surpluses associated with each seller. Buyers behave optimally when their choice maximizes their consumer surplus. When $\max\{CS_i,CS_j\}<0$, it is optimal not to purchase. When $\max\{CS_i,CS_j\}>0$, the optimal decision is to buy the product with the higher surplus.
 \end{flushleft}
\end{table}


\newpage
\begin{table}[htbp]
\fontsize{9pt}{11pt}\linespread{1.2}\selectfont
  \centering
  \caption{Buyers' profits and purchasing decisions by certification cost and precision}
\begin{tabular}{p{11.715em}ccccccccccc}
\hline
\hline
\multicolumn{1}{r}{} & \multirow{2}[4]{*}{\textbf{Mean}} & \multicolumn{4}{c}{\textbf{Certification cost}} &       & \multicolumn{5}{c}{\textbf{Certification precision}} \bigstrut\\
\cline{3-6}\cline{8-12}\multicolumn{1}{r}{} &       & \textbf{10} & \textbf{15} & \textbf{20} & \textbf{25} &       & \textbf{0.5} & \textbf{0.6} & \textbf{0.7} & \textbf{0.8} & \textbf{0.9} \bigstrut\\

\hline
\cline{1-12}
\multicolumn{1}{l}{\textbf{Panel A. Profit. D$\f$}}\\
D1 & 9.46\\
D2    & 16.78  & 16.07 &  16.47 & 15.44  & 19.12  &       &       &       &       &       & \multicolumn{1}{c}{} \bigstrut[t]\\
D3    & 4.48  &  3.26  & 7.41 &  2.15 &  5.11 &       & 6.54   & 3.08  & 3.43 & 4.22  & 5.14 \\
\cline{2-6}\cline{8-12}
D1-D2 &  ***    &    &    &    &    &       &       &       &       &       & \multicolumn{1}{c}{} \bigstrut[b]\\
D1-D3 &  ***    &    &    &   &    &       &       &       &       &       & \multicolumn{1}{c}{} \bigstrut[b]\\
D2-D3 &  ***    &  *** &  ***  & *** &  ***  &       &       &       &       &       & \multicolumn{1}{c}{} \bigstrut[b]\\

\hline
\cline{1-12}

\multicolumn{1}{l}{\textbf{Panel B. Profit. D$\e$}}\\
D1& 39.28\\
D2  & 36.17  & 35.06  &  36.64  &  34.14 & 38.83 &       &       &       &       &       & \multicolumn{1}{c}{} \bigstrut[t]\\
D3  & 15.33  &  11.64  &  21.60  &  11.64 & 16.44   &       & 21.08   &  11.95 & 11.58 & 17.25  & 14.08 \\
\cline{2-6}\cline{8-12}
D1-D2 &  **   &    &    &    &    &       &       &       &       &       & \multicolumn{1}{c}{} \bigstrut[b]\\
D1-D3 &  ***    &    &    &   &    &       &       &       &       &       & \multicolumn{1}{c}{} \bigstrut[b]\\
D2-D3 &  ***    &  *** &  *** &  ***   &  ***  &       &       &       &       &       & \multicolumn{1}{c}{} \bigstrut[b]\\
%
%
%
%
\hline
\hline
\end{tabular}%
\label{tab:buyersprofit}
\begin{flushleft}
\emph{Notes:} Asterisks indicate significance levels (* p\textless0.1, ** p\textless0.05, *** p\textless0.01), based on one-sided $t$-test with unequal variances.
\end{flushleft}
\end{table}

\newpage
\begin{table}[htbp]
\fontsize{9pt}{11pt}\linespread{1.2}\selectfont
  \centering
  \caption{Buyers' purchasing decisions}
    \begin{tabular}{lccccccc}
    \hline
    \hline
  \multicolumn{8}{l}{\textbf{Panel A. Buyers' purchasing decisions, D$\f$}}\\
  & {\textit{Mean}} & {\textit{Two Cert.}} & {\textit{One Cert.}} & {\textit{Zero Cert.}} &\multicolumn{3}{c}{\textit{Differences}}\\
\cline{6-8}&& \textit{(2)} & \textit{(1)} & \textit{(0)} & \textit{(2)-(1)} & \textit{(2)-(0)} & \textit{(1)-(0)}   \bigstrut\\
\cline{2-8}
    D1    & 0.79  &  no obs     &   no obs    &0.79       &       &  &  \\
    D2    & 0.89 & 0.92 & 0.94  &0.69 &  -     &*** &***  \\
    D3    & 0.74  & 0.88  & 0.79  &0.63 &   **   &***    &***   \\
    \cline{2-8}
    D1-D2 &***   &       &       &       &       &    &       \\
    D1-D3 &  *   &       &       &       &       &     &        \\
    D2-D3 &***   & -  & ***     &-      &       &    &    \\
    \hline
    \cline{1-8}
      \multicolumn{8}{l}{\textbf{Panel B. Buyers' purchasing decisions, D$\e$}}\\
  & {\textit{Mean}} & {\textit{Two Cert.}} & {\textit{One Cert.}} & {\textit{Zero Cert.}} &\multicolumn{3}{c}{\textit{Differences}}\\
\cline{6-8}&& \textit{(2)} & \textit{(1)} & \textit{(0)} & \textit{ (2)-(1)} & \textit{(2)-(0)} & \textit{(1)-(0)}   \bigstrut\\
\cline{2-8}
    D1    & 0.98  &  no obs     &   no obs    &0.98       &       &  &  \\
    D2    & 1.00 & 1.00 & 1.00  &1.00 &  -     &- &-  \\
    D3    & 0.93  & 1.00  & 0.92  &0.93 &   ***   &***   &-   \\
     \cline{2-8}
    D1-D2 &**   &       &       &       &       &    &       \\
    D1-D3 &  ***   &       &       &       &       &     &        \\
    D2-D3 &***   & -  & ***     &***      &       &    &    \\
    \hline
\hline
    \end{tabular}%
  \label{tab:buyers}%
 \begin{flushleft} \emph{Notes}: Asterisks indicate significance levels (* p\textless0.1, ** p\textless0.05, *** p\textless0.01), based on one-sided $t$-test with unequal variances.  Non-significant results are marked with a dash ("-"). "Two Cert.," "One Cert.," and "Zero Cert." represent scenarios where two sellers, one seller, or no sellers disclose certification outcomes, respectively.
\end{flushleft}
\end{table}%

\newpage
\begin{table}[htbp]
\fontsize{9pt}{11pt}\linespread{1.2}\selectfont
  \centering
  \caption{Welfare by certification cost and precision}
\begin{tabular}{p{9.715em}ccccccccccc}
\hline
\hline
\multicolumn{1}{r}{} & \multirow{2}[4]{*}{\textbf{Mean}} & \multicolumn{4}{c}{\textbf{Certification cost}} &       & \multicolumn{5}{c}{\textbf{Certification precision}} \bigstrut\\
\cline{3-6}\cline{8-12}\multicolumn{1}{r}{} &       & \textbf{10} & \textbf{15} & \textbf{20} & \textbf{25} &       & \textbf{0.5} & \textbf{0.6} & \textbf{0.7} & \textbf{0.8} & \textbf{0.9} \bigstrut\\

\hline
\cline{1-12}

\multicolumn{1}{l}{\textbf{Panel A. Net Welfare. D$\f$}}\\
D1   &29.10 \\
D2    & 29.69  & 30.14 &  31.01&  28.79 &  28.83&       &       &       &       &       & \multicolumn{1}{c}{} \bigstrut[t]\\
D3    & 24.65  &  24.32 & 26.88 & 24.55 &  22.87&       &24.06    &23.89   &25.01 & 24.98 &25.32  \\
\cline{2-6}\cline{8-12}
D1-D2 &  -   &    &    &    &    &       &       &       &       &       & \multicolumn{1}{c}{} \bigstrut[b]\\
D1-D3 &  ***   &    &    &   &    &       &       &       &       &       & \multicolumn{1}{c}{} \bigstrut[b]\\
D2-D3 & ***    & ** &  ** & ** &  *** &       &       &       &       &       & \multicolumn{1}{c}{} \bigstrut[b]\\

\hline
\cline{1-12}

\multicolumn{1}{l}{\textbf{Panel B. Net Welfare. D$\e$}}\\
D1   &  44.26 \\
D2    &  40.99 & 42.63&  40.99 &  39.33&  41.03&       &       &       &       &       & \multicolumn{1}{c}{} \bigstrut[t]\\
D3    &  39.70 &  38.70 & 41.07 & 37.29 & 41.72 &       & 39.76   & 36.73  & 40.34 & 41.19 & 40.46 \\
\cline{2-6}\cline{8-12}
D1-D2 &  ***   &    &    &    &    &       &       &       &       &       & \multicolumn{1}{c}{} \bigstrut[b]\\
D1-D3 & ***  &    &    &   &    &       &       &       &       &       & \multicolumn{1}{c}{} \bigstrut[b]\\
D2-D3 & * &  **& -  & -&  - &       &       &       &       &       & \multicolumn{1}{c}{} \bigstrut[b]\\
\hline
\cline{1-12}

\multicolumn{1}{l}{\textbf{Panel C. Gross Welfare. D$\f$}}\\
D1 & 29.10\\
D2    & 34.33  & 33.31& 35.34&  33.57&  35.08 &       &       &       &       &       & \multicolumn{1}{c}{} \bigstrut[t]\\
D3    & 28.40 &  27.03 & 30.06& 29.05& 27.45 &       & 26.74  & 26.89  & 29.46& 28.97 & 29.93 \\
\cline{2-6}\cline{8-12}
D1-D2 &  ***   &    &    &    &    &       &       &       &       &       & \multicolumn{1}{c}{} \bigstrut[b]\\
D1-D3 &  -   &    &    &   &    &       &       &       &       &       & \multicolumn{1}{c}{} \bigstrut[b]\\
D2-D3 &   ***  & *** & ***  &** &  *** &       &       &       &       &       & \multicolumn{1}{c}{} \bigstrut[b]\\

\hline
\cline{1-12}

\multicolumn{1}{l}{\textbf{Panel D. Gross Welfare. D$\e$}}\\
D1 & 44.26 \\
D2    &  45.52 & 45.96&  45.78&  45.17&  45.19 &       &       &       &       &       & \multicolumn{1}{c}{} \bigstrut[t]\\
D3    &  41.87 &40.30  & 43.92& 39.29&  43.97&       &41.26    & 38.66  & 42.21& 43.75  &43.46  \\
\cline{2-6}\cline{8-12}
D1-D2 &   **  &    &    &    &    &       &       &       &       &       & \multicolumn{1}{c}{} \bigstrut[b]\\
D1-D3 &    ** &    &    &   &    &       &       &       &       &       & \multicolumn{1}{c}{} \bigstrut[b]\\
D2-D3 &   ***  & *** &  - & ***& -  &       &       &       &       &       & \multicolumn{1}{c}{} \bigstrut[b]\\
\hline
\hline
\end{tabular}%
\label{tab:welfarecalpha}
\begin{flushleft}
\emph{Notes}: Asterisks indicate significance levels (* p\textless0.1, ** p\textless0.05, *** p\textless0.01), based on one-sided $t$-test with unequal variances.
\end{flushleft}
\end{table}

\begin{table}[ht!]
\fontsize{9pt}{11pt}\linespread{1.2}\selectfont
  \centering
  \caption{Robutness check for sellers' profit, buyers' profit, and welfare}
\begin{tabular}{p{8.715em}lllll}
\hline
\hline
\multicolumn{1}{r}{} & \multirow{2}[4]{*}{\textbf{Mean}} & \multicolumn{4}{c}{\textbf{Certification cost}}  \bigstrut\\
\cline{3-6}\multicolumn{1}{r}{} &       & \textbf{10} & \textbf{15} & \textbf{20} & \textbf{25} \bigstrut\\
\hline

\multicolumn{1}{l}{\textbf{Net Profit. D$\f$}}\\

D1   & (47.6, 8.6, 28.1) \\
D2   & (43.1, 20.7, 31.9) & (43.2, 23.1, 33.2) & (45.3, 18.5, 31.9) & (38.1, 22.2, 30.1) & (42.6, 20.7, 31.7) \bigstrut[t]\\
D3   & (45.1, 6.1, 25.6) & (45.9, 2.6, 24.3) & (48.0, 9.2, 28.6) & (39.6, 2.0, 20.8) & (43.9, 7.4, 25.6) \\
\cline{2-6}
D1-D2 &  (-,***,**)   &    &    &    &     \\
D1-D3 &  (-,***,-)   &    &    &   &   \\
D2-D3 &  (-,***,***)  & (-,***,***) & (-,***,*)  &(-,***,*) & (-,***,**)  \bigstrut[b]\\
\hline

\multicolumn{1}{l}{\textbf{Net Profit. D$\e$}}\\
D1   & (40.7, 48.0, 44.3) \\
D2   & (43.2, 39.8, 41.5) & (46.5, 36.8, 41.6) & (42.0, 40.5, 41.3) & (42.5, 42.0, 42.3) & (42.4, 40.1, 41.3) \bigstrut[t]\\
D3   & (62.3, 21.6, 42.0) & (64.8, 17.1, 40.9) & (62.3, 24.4, 43.3) & (59.5, 20.4, 39.9) & (62.7, 22.0, 42.3) \\
\cline{2-6}
D1-D2 &  (-,***,***)   &    &    &    &     \\
D1-D3 &  (***,***,**)   &    &    &   &   \\
D2-D3 &  (***,***,-)  & (-,***,-) & (-,***,**)  &(-,***,-) & (-,***,-)  \bigstrut[b]\\
\hline
\multicolumn{1}{l}{\textbf{Gross Profit. D$\f$}}\\

D1   & (47.6, 8.6, 28.1) \\
D2   & (52.9, 20.7, 36.8) & (49.3, 23.1, 36.2) & (54.2, 18.5, 36.3) & (46.9, 22.2, 34.6) & (55.1, 20.7, 37.9) \bigstrut[t]\\
D3   & (51.6, 6.1, 28.9) & (51.4, 2.6, 27.0) & (54.7, 9.2,  32.0) & (48.4, 2.0, 25.2) & (49.1, 7.4, 28.2) \\
\cline{2-6}
D1-D2 &  (-,***,***)   &    &    &    &     \\
D1-D3 &  (-,***,-)   &    &    &   &   \\
D2-D3 &  (-,***,***)  & (-,***,***) & (-,***,*)  &(-,***,*) & (-,***,***)  \bigstrut[b]\\
\hline

\multicolumn{1}{l}{\textbf{Gross Profit. D$\e$}}\\
D1   & (40.7, 48.0, 44.3) \\
D2   & (50.2, 39.8, 45.0) & (52.1, 36.8, 44.4) & (49.5, 40.5, 45.0) & (47.5, 42.0, 44.8) & (50.7, 40.1, 45.4 ) \bigstrut[t]\\
D3   & (67.1, 21.6, 44.3) & (68.8, 17.1, 42.9) & (67.9, 24.4, 46.1) & (65.5, 20.4, 42.9) & (65.2, 22.0, 43.6) \\
\cline{2-6}
D1-D2 &  (*,***,-)   &    &    &    &     \\
D1-D3 &  (***,***,-)   &    &    &   &   \\
D2-D3 &  (**,***,-)  & (-,***,-) & (-,***,-)  &(-,***,-) & (-,***,-)  \bigstrut[b]\\
\hline
\hline
\end{tabular}%
\label{tab:robustnesscheck}
\begin{flushleft}
\emph{Notes}: Asterisks indicate significance levels (* p\textless0.1, ** p\textless0.05, *** p\textless0.01), based on one-sided $t$-test with unequal variances. In each triple of numbers, the first, second, and third entries represent sellers' profit, buyers' profit, and welfare, respectively. In a tripleton of significance levels, the first, second, and third asterisks correspond to significance levels for comparisons of sellers' profit, buyers' profit, and welfare comparison, respectively. The analysis is restricted to the subsample of data from period 11 onward.
\end{flushleft}
\end{table}

\clearpage
\noindent \Large{\textbf{APPENDIX}}
\normalsize
\begin{appendices}


\section{Proofs}

\textbf{Proof of Proposition \ref{pro:WTPs}:}
Let $p(v|s_i)>p(v|s_j)$.  By the standard argument, in the equilibrium of this subgame, seller $j$ sets a price $0$, while seller $i$ sets a price $p(v|s_i)-p(v|s_j)>0$. All buyers
purchase from seller $i$, and the sellers' profits are $\pi_i(s_i,s_j)=p(v|s_i) - p(v|s_j)$ and
$\pi_j(s_i,s_j)=0$. Note that even though buyers are indifferent between the two sellers, the equilibrium outcome involves all buyers purchasing from seller $i$. A scenario where buyers randomly choose between sellers would not constitute an equilibrium, as seller $i$  would have a profitable deviation of setting a price slightly below  $p(v|s_i)-p(v|s_j)$.
Let $p(v|s_i)=p(v|s_j)$, then both sellers set a price of 0.
Regardless of how buyers split between the two sellers, both sellers earn zero profit.$~\blacksquare$

\smallskip
\smallskip
\noindent \textbf{Proof of Proposition \ref{pro:2Tsame}:} Let $\Pi_i^\alpha$ and $\Pi_i^1$ denote the ex ante expected profit of seller $i$ under noisy certification with precision $\alpha$ and
accurate certification, respectively.
Due to firms' symmetry, it is sufficient to compare the profits of seller $i$.

Under accurate certification, $\Pi^1_i=q_H(1-q_H)\Delta v-cq_H$. Player $i$ earns positive profit only when his type is $v_H$ and the competitor's type is $v_L$, which happens with probability $q_H(1-q_H)$. The profit earned in this $(s_H,ND)$-subgame is $\Delta v=v_H-v_L$, and the expected certification cost is $cq_H$.

Under noisy certification,
\[
\Pi_i^\alpha=\Pr(s_i=s_H)\Pr(s_j=ND)(p(v|s_H)-p(v|ND))-cq_H.
\]
Given the sellers' strategy, $\Pr(s_i=s_H)=q_H(\alpha+q_H(1-\alpha))$. Either, with probability $\alpha$, the certification succeeds and results in $s_H$ or, with probability $1-\alpha$, the certification fails and it returns the certification outcome $s_H$ with the prior probability of $q_H$. Additionally, $\Pr(s_j=ND)=(1-q_H)[1+q_H(1-\alpha)]$. Either, with probability $1-q_H$, the competitor does not certify as his type is low or, with probability $q_H(1-q_H)(1-\alpha)$, the competitor certifies but the certification fails and results in outcome $s_L$.

Buyers' WTP conditional on $s_H$ and $ND$ can be calculated as follows:
\begin{eqnarray*}
\Pr(v_L|ND)&=&\frac{1}{1+q_H(1-\alpha)}\\
E(v|ND)&=&v_H-\Pr(v_L|ND)\cdot \Delta v\\
p(v|s_H) &=& v_H\\
p(v|ND)&=&E(v|ND)+\Pr(v_L|ND)\cdot b\cdot (v_L-E(v|ND)).
\end{eqnarray*}
Combining all the terms, we get:
\[
\Pi^\alpha_i=q_H(\alpha+q_H(1-\alpha))\left[\Delta v (1-q_H) \Big(1+b\cdot \frac{q_H(1-\alpha)}{1+q_H(1-\alpha)}\Big)\right]-cq_H.
\]
The inequality $\Pi^\alpha_i> \Pi^1_i$ is equivalent to
\[
(\alpha+q_H(1-\alpha))\left[1+b\cdot \frac{q_H(1-\alpha)}{1+q_H(1-\alpha)}\right] > 1,
\]
which, as one can show, is equivalent to (\ref{eq:sameEq}). $\blacksquare$

\noindent\textbf{Proposition A1 [Equilibrium Existence. Two Types]} \textit{Let $\displaystyle \Lambda=\frac{c}{\Delta v(1-q_H)}$ where $\Delta v=v_H-v_L$. Let $X_1= (0,1]$ and $X_2= (\underline{\gamma}, \bar{\gamma}]$, where}
\[
\underline{\gamma}=(1-\alpha)q_H\left[1+b\cdot \frac{(1-\alpha)q_H}{1+(1-\alpha)q_H}\right]; \text{ and }\bar{\gamma}=[\alpha+(1-\alpha)q_H]\left[1+b\cdot \frac{(1-\alpha)q_H}{1+(1-\alpha)q_H}\right].
\]
\noindent\textit{Then:}

\textit{ i) $(\Co,\Do)=(\{v_H\},\{s_H\})$ is an equilibrium in the 100\%-case if and only if $\Lambda\in X_1$;}

\textit{ ii) $(\Ca,\Da)=(\{v_H\},\{s_H\})$ is an equilibrium in the $\alpha$-case if and only if $\Lambda\in X_2$;}

\begin{proof}
\noindent \textit{\underline{Part i):}} Type $v_L$ will never find it optimal to deviate and certify when certification is accurate. Following equilibrium strategy, type $v_H$ earns
profit of $(1 - q_H)(v_H - v_L)- c$. Hence, type $v_H$ will find it optimal to certify when $\Lambda=\frac{c}{\Delta v(1-q_H)}\le 1$ or $\Lambda\in X_1$.

\bigskip
\noindent \textit{\underline{Part ii):}} Equilibrium profits of type $v_H$ and  $v_L$ are as follows:
\[
\Pi^\alpha(v_H)=\Pr(s_H|v_H) \Pr(ND)(p(v|s_H)-p(v|ND))-c,  \text{ and   } \Pi^\alpha (v_L)=0.
\]
If type $v_H$ deviates and does not certify, its profit is zero. If type $v_L$ deviates,
certifies and discloses $s_H$, its profit is:
\[
\Pi^{\alpha, dev}(v_L)=\Pr(s_H|v_L) \Pr(ND)(p(v|s_H)-p(v|ND))-c.
\]
In equilibrium, it must be that $\Pi^\alpha (v_H)\ge 0> \Pi^{\alpha, dev}(v_L)$. Given the sellers' strategy, $\Pr(s_H|v_H) =\alpha+(1-\alpha)q_H$, and $\Pr(s_H|v_L)=(1-\alpha)q_H$. The remaining terms have been calculated above. One can show that $\Pi^\alpha (v_H)\ge 0> \Pi^{\alpha, dev}(v_L)$ is equivalent to $\Lambda \in X_2$.\end{proof}

\bigskip
\textbf{Proof of Proposition \ref{pro:D2RN}:} Note that in any equilibrium of the environment with the 100\%-certification, the disclosing decision is trivial as
all certifying types disclose: $v_C=v_D$. Since sellers face no uncertainty about the certification outcome, they pay the certification fee if and only if they plan to disclose the certification outcome.

First, we prove that $(\Crn,\Drn)$ is an equilibrium. Let $\qND$ denote the probability that seller $j$ does not disclose in equilibrium. Given
$(\Crn,\Drn)$, it is equal to $q_1+\dots+q_{m-1}$. Let $p(v|ND)$ be buyers' WTP conditional
on non-disclosure. Note that $p(v|ND)<p(v|s=v_k)=v_k$ for any $k\ge m$ and, therefore, $E\pi_i(ND)=0$. All types $v_k\ge v_m$ find it optimal to certify:
\begin{eqnarray*}
E\Pi_i(v_k)&=&\sum_{j=m}^k q_j (v_k-v_j)+ \qND(v_k-E(v|ND))-c\\
&=&\sum_{j=m}^k q_j (v_k-v_j)+ \qND\left(v_k-\sum_{j=1}^{m-1}\frac{1}{\qND} q_jv_j\right)-c\\
&=&\sum_{j=m}^k q_j (v_k-v_j)+ \sum_{j=1}^{m-1} q_j(v_k-v_j)-c=\sum_{j=1}^k q_j (v_k-v_j)-c\ge 0.
\end{eqnarray*}

All types $v_k<v_m$ find it optimal not to certify. The equilibrium profit of $v_k$ is zero. If type $v_k$ deviates and certifies then its expected deviation profit is $E\Pi^{dev}_i(v_k)=\sum_{j=1}^k q_j (v_k-v_j)-c<0.$ Thus, $(\Crn,\Drn)$ is an equilibrium.

Next, we prove that this is a unique equilibrium. The proof is based on the observation that the expression for $E\Pi_i(v_k)$, derived above, does not depend on $m$. That is, it does not depend
on the lowest certifying type. It only depends on whether type $v_k$ certifies or not.

Consider an equilibrium where $v_C=v_D=v_l$. If $l>m$, then it cannot be an equilibrium since $v_{l-1}$ has a strictly profitable deviation of certifying
and earning a positive profit. Note that we use here the off-equilibrium beliefs refinement, which is that disclosing message $s=v_{l-1}$ leads to buyers' beliefs $\Pr(v|s=v_{l-1})=v_{l-1}$.
If $l<m$, then type $l$ would earn a negative profit in such an equilibrium
and has a profitable deviation of not certifying.

The last part of the Proposition follows immediately from the expression for $E\Pi_i(v_k)$: if such $v_m$ does not exist, then no certifying type could earn a non-negative payoff in equilibrium.$~\blacksquare$


\bigskip
\noindent \textbf{Proof of Proposition \ref{pro:RankingRN}}: The proof is based on two lemmas.
Lemma \ref{le:higherprofit} shows that equilibrium $(\Crn,\Drn)$ is more profitable under the 100\%-certification than any (not necessarily equilibrium) strategy profile $(C,D)$ such
that $v_C=v_D$. Lemma \ref{le:RNsame} shows that any such strategy profile $(C,D)$
is more profitable under the 100\%-certification than under the $\alpha$-certification.
The combination of the two Lemmas proves the Proposition statement.

Throughout the proof, when we say that a strategy profile, e.g., $(C, D)$, is not necessarily an equilibrium, we mean that sellers' certification and disclosure strategies are not necessarily optimal. Nevertheless, we do assume that buyers' beliefs and WTP are determined by the sellers' strategies, and that buyers make optimal purchasing decisions, ensuring that Proposition \ref{pro:WTPs} holds.

\begin{lemma}
Let $(C,D)$ be a (not necessarily equilibrium) strategy profile such that $v_C=v_D=v_l$. Under the 100\%-certification $(C,D)$ is weakly less profitable than $(\Crn,\Drn)$.
\label{le:higherprofit}
\end{lemma}

\begin{proof}
First, suppose $v_l>v_m$. Types with $v_k<v_m$ earn zero profit under both strategy profiles
as they do not certify. According to Proposition \ref{pro:D2RN}, types with $v_k\ge v_l$ also earn equal profits under both strategy profiles: $E\Pi_i(v_k)=\sum_{j=1}^k q_j (v_k-v_j)-c$. Finally, consider types $v_k$ such that $m\le k<l$. Under $(C,D)$, type $v_k$ earns zero profit as it does not certify, while under $(\Crn,\Drn)$, it earns a positive profit. Thus, the sellers' aggregated profit is higher under $(\Crn,\Drn)$. Next, suppose $v_l<v_m$. All types such that $v<v_l$ and $v\ge v_m$ earn equal expected profits under both strategy profiles: 0 and $E\pi(v_k)=\sum_{j=1}^k q_j (v_k-v_j)-c$, respectively. Consider type $v_k$ such that $l\le k<m$. Under $(C,D)$, type $v_k$ earns negative profit
since it certifies and $k<m$. Under $(\Crn, \Drn)$, it earns zero profit as it does not certify.
Thus, the sellers' aggregated profit is higher under $(\Crn,\Drn)$.
\end{proof}

\begin{lemma}
Let $(C,D)$ be a (not necessarily equilibrium) strategy profile such that $v_C=v_D=v_l$. Then, sellers' aggregate expected profit under the $\alpha$-certification is a strictly increasing function of $\alpha$.
\label{le:RNsame}
\end{lemma}
\begin{proof}
Let $\qD=q_l+\dots+q_n$ and $\qND=1-\qD$. Define $\rho_i=\Pr(s=v_i)=\alpha q_i+(1-\alpha)\qD q_i$, where $i\ge l$, as the \textit{unconditional} probability of observing the certification outcome $s=v_i$. Let $\pND=\Pr(ND)=\qND+(1-\alpha)\qD \qND$ be the \textit{unconditional} probability of non-disclosure. Non-disclosure occurs if either $v=v_k$, for $k<l$, which happens with probability $q_k$; or if $v=v_k$, for $k\ge l$, and the certification outcome is wrong with a realization
below $v_l$, which happens with probability $q_k(1-\alpha)\qND$. The sum of these probabilities yields $\pND$. Lastly, let $\pD=1-\pND$ denote the \textit{unconditional} probability of disclosure.

Buyers' WTP given $s=v_k$, where $k\ge l$, is the expected quality conditional on $s=v_k$:
\begin{equation}
p(v|s=v_k)=E(v|s=v_k)=\alpha\frac{q_k}{\rho_k}v_k+\sum_{j=l}^n (1-\alpha)\frac{q_k}{\rho_k}q_jv_j.
\label{eq1:Lemma1}
\end{equation}
Buyers' WTP conditional on non-disclosure is:
$$
p(v|ND)=E(v|ND)=\sum_{j=1}^{l-1} \frac{1}{\pND}q_j v_j+\sum_{j=l}^n (1-\alpha)\frac{\qND}{\pND}q_jv_j.
$$
Each seller sends messages $(ND, s=v_l,\dots, s=v_n)$ with
probabilities $(\pND,\rho_l,\dots, \rho_n)$. For the sake of brevity, let $p(v|s_i)$ denote
$p(v|s=v_i)$. By Proposition \ref{pro:WTPs}, the joint aggregated profit of the sellers can be expressed as:
\begin{equation*}
\underbrace{\pND^2 (p(v|ND)-p(v|ND))}_{\mbox{\footnotesize neither seller discloses}} + \underbrace{\sum_{i,j=l}^n \rho_i \rho_j |p(v|s_i)-p(v|s_j)|}_{\mbox{\footnotesize both sellers disclose}} +\underbrace{2\sum_{j=l}^n \rho_j \pND(p(v|s_j)-p(v|ND))}_{\mbox{\footnotesize one seller discloses}}.
\label{eq:jointprofitRN}
\end{equation*}
In this expression, the first term corresponds to the joint profit when neither seller discloses, the second term is the joint profit when both sellers disclose, and the third term is the joint profit when exactly one seller discloses. For the second term, we took into account that Proposition \ref{pro:WTPs} implies that $\pi_i(s_i,s_j)+\pi_j(s_i,s_j)=|p(v|s_i)-p(v|s_j)|$. For the last term, we took into account that $p(v|s_j)-p(v|ND)\ge 0$, meaning that a non-disclosing seller earns zero.

\textbf{Term 1:} Neither seller discloses, the joint profit is zero.

\textbf{Term 2:} Both sellers disclose. We will prove that Term 2 is an
increasing function of $\alpha$. We can re-write Term 2 as:
\begin{equation*}
\sum_{i,j=l}^n \rho_i \rho_j |p(v|s_i)-p(v|s_j)|=\pD^2 \sum_{i,j} \frac{\rho_i}{\pD} \frac{\rho_j}{\pD} |p(v|s_i)-p(v|s_j)|
=\pD^2 \sum_{i,j} \frac{q_i}{\qD} \frac{q_j}{\qD}|p(v|s_i)-p(v|s_j)|,
\label{eq2: Lemma1}
\end{equation*}
where the second equality uses $\rho_i/q_i=\rho_j/q_j=\pD/\qD$ for all $i,j\ge l$ and $\alpha$.
The terms $q_i/\qD$ and $q_j/\qD$ do not depend on $\alpha$. The term $\pD=1-\pND$ is
an increasing function of $\alpha$ because $\pND$ is a decreasing function of $\alpha$. Finally, the term $|p(v|s_i)-p(v|s_j)|$ is also an increasing function of $\alpha$. Indeed, from
(\ref{eq1:Lemma1}), and the fact that $\rho_i/q_i=\rho_j/q_j=\pD/\qD$ we have
\[
|p(v|s_i)-p(v|s_j)|=\Bigg|\alpha\frac{q_i}{\rho_i}v_i-\alpha\frac{q_j}{\rho_j}v_j\Bigg|=\alpha \frac{\qD}{\pD}|v_i-v_j|,
\]
which is a strictly increasing function of $\alpha$, whenever $i\ne j$. Thus, Term 2 is a strictly increasing function of $\alpha$ unless $l=n$.

\textbf{Terms 3}: Only one seller discloses. We will prove that Term 3 is a strictly increasing function of $\alpha$. Ignoring the constant 2, we can re-write it as:
\begin{eqnarray*}
\sum_{j=l}^n \rho_j \pND(p(v|s_j)-p(v|ND))&=&\pND^2 (p(v|ND)-p(v|ND))+\sum_{j=l}^n \rho_j \pND(p(v|s_j)-p(v|ND))\\
&=&\pND(Ev - p(v|ND)).
\end{eqnarray*}
The first equality holds because we just added a zero. The second equality holds due to
the fact that certification is unbiased, $\pND E(v|ND)+\sum_{j=l}^n \rho_j E(v|s_j) = Ev$,
and that buyers are risk neutral, so $p(v|ND)=E(v|ND)$ and $p(v|s_j)=E(v|s_j)$. We can re-write $\pND Ev$ as:
$$
\pND Ev =  \sum_{j=1}^{l-1} q_j Ev + (1-\alpha)\qND \sum_{j=l}^n q_j Ev,
$$
and we can re-write $\pND p(v|ND)$ as:
\begin{eqnarray*}
\pND p(v|ND)&=&\pND\sum_{j=1}^n \Pr(v_j|ND)v_j=\pND\sum_{j=1}^n \frac{\Pr(ND|v_j)\Pr(v=v_j)}{\Pr(ND)}v_j= \\
&=&\sum_{j=1}^{l-1} q_j v_j + (1-\alpha)\qND \sum_{j=l}^n q_j v_j,
\end{eqnarray*}
where we use that  $\pND=\Pr(ND)$, $\Pr(ND|v_j)=1$ if $j<l$, and $\Pr(ND|v_j)=(1-\alpha)\qND$ otherwise.
Thus,
$$
\pND(Ev - p(v|ND))=\sum_{j=1}^{l-1} q_j (Ev-v_j) + (1-\alpha)\qND \sum_{j=l}^n q_j (Ev-v_j).
$$
The first term does not depend on $\alpha$. The second term is $(1-\alpha)\qND$ times a negative term:
$$
\sum_{j=l}^n q_j (Ev-v_j)=Ev\sum_{j=l}^nq_j -\sum_{j=l}^nq_jv_j=\qD Ev - \qD E(v|v\ge v_l)\le 0.
$$
The inequality is strict unless $l=1$. Recall that Term 2 is an increasing function of $\alpha$ and is strictly increasing unless $l=n$. Thus, the sum of Terms 2 and 3 is a strictly increasing function of $\alpha$.\end{proof}

Combining the two Lemmas proves the Proposition statement. $\blacksquare$

\medskip
\textbf{Proof of Proposition \ref{pro:LAgeneral}:}
Let $q_D=q_l+\dots+q_{n}$ denote the probability of disclosure, and let $\qND=1-q_D$ denote the probability of non-disclosure. Define $E_C=E[v|v\ge v_l]=\frac{1}{q_D}\sum_{k=l}^{n} q_kv_k$ and $\END=E[v|v<v_l]=\frac{1}{\qND}\sum_{k=1}^{l-1} q_kv_k$ as the average quality of certified and non-certified types, respectively.

As in the proof of Lemma \ref{le:RNsame}, let $\rho_i=\Pr(s=v_i)=\alpha q_i+(1-\alpha)\qD q_i$, where $i\ge l$, be the \textit{unconditional} probability of observing the certification outcome $s=v_i$. Let $\pND=\Pr(ND)=\qND+(1-\alpha)\qD \qND$ be the \textit{unconditional} probability of non-disclosure. Let $\pD=1-\pND$ denote the \textit{unconditional} probability of disclosure. We use the following calulations in the proof.
\begin{eqnarray*}
\Pr(v=v_j|s_j)&=& \frac{\alpha+(1-\alpha)q_j}{\alpha+(1-\alpha)q_D}\quad \text{ when $j\ge l$} \\
\Pr(v=v_k|s_j)&=& \frac{(1-\alpha)q_k}{\alpha+(1-\alpha)q_D}\quad \text{ when $j\ge l$ and $k\ge l$ and $k\ne j$} \\
\Pr(v=v_k|s_j)&=& 0 \quad \text{ when $j\ge l$ and $k< l$} \\
\Pr(v=v_k|ND)&=& \frac{(1-\alpha)\qND q_k}{\qND+(1-\alpha)\qD \qND} \quad \text{ when $k\ge l$} \\
\Pr(v=v_k|ND)&=& \frac{q_k}{\qND+(1-\alpha)\qD \qND}\quad \text{ when $k< l$} \\
E(v|s_j)&=&\alpha\frac{q_j}{\rho_j}v_j+\sum_{k=l}^n (1-\alpha)\frac{q_j}{\rho_j}q_kv_k \quad \text{ when $j\ge l$} \\
E(v|ND)&=&\sum_{k=1}^{l-1} \frac{1}{\pND}q_k v_k+\sum_{k=l}^n (1-\alpha)\frac{\qND}{\pND}q_kv_k = \END+\sum_{k=l}^n (1-\alpha)\frac{\qND}{\pND}q_kv_k
\end{eqnarray*}

Let $D_j=\frac{\partial EL_j}{\partial\alpha}\big|_{\alpha=1}$, where $j\in\{ND,l,\dots,n\}$.
If $v_j\ge E_C$ and $\alpha$ is sufficiently close to 1, then $v_{j-1}<E(v|s_j)<v_j$. Consequently, conditional on observing signal $s_j$, buyers experience a loss if and only if $v_l\le v\le v_{j-1}$, where $v$ is the quality of the purchased product. It follows
\begin{eqnarray*}
D_j=\frac{\partial}{\partial\alpha}\left(\sum_{k=l}^{j-1} \Pr(v=v_k|s_j)(v_k-E(v|s_j))\right)\Bigg|_{\alpha=1}=\sum_{k=l}^{j-1} q_k(v_j-v_k).
\end{eqnarray*}
The second equation comes from the product rule of differentiation and the fact that when $j,k\ge l$ and $k<j$, we have
\begin{eqnarray*}
\Pr(v=v_k|s_j)\big |_{\alpha=1}= 0, \quad \text{ and } \quad \frac{\partial \Pr(v=v_k|s_j)}{\partial \alpha}\bigg |_{\alpha=1}= -q_k, \quad \text{ and } \quad E(v|s_j)\big |_{\alpha=1} = v_j.
\end{eqnarray*}

Similarly, if $v_j< E_C$ and $\alpha$ is sufficiently close to 1, then $v_{j}<E(v|s_j)<v_{j+1}$.
Buyers experience a loss when the product quality is $v_j$ or lower, $v_l\le v\le v_j$.
Thus,
\begin{eqnarray*}
D_j&=&\sum_{k=l}^{j-1} q_k(v_j-v_k)+(q_D E_C-q_D v_j)=\sum_{k=j+1}^n q_k (v_k-v_j).
\end{eqnarray*}
The term $q_D E_C-q_D v_j$ is indeed $\frac{\partial}{\partial\alpha}(\Pr(v=v_j|s_j)(v_j-E(v|s_j)))\Big|_{\alpha=1}$, which results from the facts that when $j\ge l$, we have
\begin{eqnarray*}
\Pr(v=v_j|s_j)\big |_{\alpha=1}= 1, \quad \text{ and } \quad \frac{\partial \Pr(v=v_j|s_j)}{\partial \alpha}\bigg |_{\alpha=1}&=& q_D-q_j, \quad \text{ and } \quad E(v|s_j)\big |_{\alpha=1}=v_j,
\end{eqnarray*}
and $\frac{\partial E(v|s_j)}{\partial \alpha}\big |_{\alpha=1}= q_Dv_j-\sum_{k=l}^n q_kv_k=q_Dv_j-q_DE_C$.

The derivative of $K_1(\alpha)$ at $\alpha=1$ can thus be expressed as the sum of two terms:
\begin{eqnarray*}
\frac{\partial K_1(\alpha)}{\partial\alpha}\Big|_{\alpha=1}=\underbrace{\frac{\partial}{\partial\alpha}\sum_{l\le i<j\le n} \rho_i\rho_j (EL_j-EL_i)\Big|_{\alpha=1}}_{S1}+\underbrace{\frac{\partial}{\partial\alpha} \sum_{k=l}^n \rho_k\rho_{ND} (EL_k-EL_{ND})\Big|_{\alpha=1}}_{S2}
\end{eqnarray*}

\paragraph{Term $S1$:} Given that $EL_j(1)=0$ for all $j\ge l$, as the product's quality is known when $\alpha=1$, and $\rho_j\big |_{\alpha=1}=q_j$ for all $j\ge l$, we have:
\begin{eqnarray*}
S1=\sum_{j=l}^n \sum_{i=l}^{j-1} q_i q_j (D_j-D_i).
\end{eqnarray*}

\paragraph{Term $S2$:}
Let $K$ be such that $v_K\le \END<v_{K+1}$. It is straightforward to see that $K<l$, as $\END$ is the average quality of non-certified types, and all types with quality $v\ge v_l$ certify.

When $\alpha<1$, $E(v|ND)>\END$ and as $\alpha\rightarrow 1$, $E(v|ND)\rightarrow \END$.
Thus, when $\alpha$ is close to 1, buyers purchasing a product without disclosure
experience a loss if and only if  $v\le v_K<v_l$. Then:
\begin{equation}
EL_{ND}(\alpha)=\sum_{k=1}^{K}\underbrace{\frac{q_k}{\qND+(1-\alpha)q_D\qND}}_{\Pr(v=v_k|ND)}
\Bigg(\underbrace{v_k-\END-\sum_{j=l}^n (1-\alpha)\frac{\qND}{\pND}q_jv_j}_{v_k-E(v|ND)}\Bigg).
\label{eq:EL_ND}
\end{equation}
Note that in the expression of $EL_{ND}(\alpha)$ above, we compute $\Pr(v=v_k|ND)$ for $k=1,..,K$, meaning that $k<l$ as $K<l$ argued earlier. Taking the derivative at $\alpha=1$, we obtain:
\begin{eqnarray*}
D_{ND}=\frac{q_D}{\qND}\sum_{k=1}^{K}q_k\left(E_C+(v_k-2\END)\right).
\end{eqnarray*}
To obtain the expression of $D_{ND}$ above, for $k\le K<l$, we used
\[
\Pr(v=v_k|ND)\big|_{\alpha=1}=\frac{q_k}{\qND}, \quad \quad E(v|ND)\big|_{\alpha=1}=\END, \quad \text{and} \quad \frac{\partial \Pr(v=v_k|ND)}{\partial \alpha} \big|_{\alpha=1}=\frac{q_kq_D}{\qND},
\]
and
\[
\frac{\partial E(v|ND)} {\partial \alpha} \big|_{\alpha=1}=\frac{q_D}{\qND} \sum_{k=1}^{l-1} q_kv_k -\sum_{k=l}^n q_kv_k=\frac{q_D}{\qND} \qND \END-q_DE_C =q_D(\END-E_C).
\]


So, term $S2$ becomes
\begin{eqnarray*}
S2&=&\sum_{k=l}^n (\rho_k\rho_{ND})^\prime\big|_{\alpha=1} (EL_k(1)-EL_{ND}(1)) + \sum_{k=l}^n q_k\qND (D_k-\DND)\\
&=&(1-2q_D)\sum_{k=l}^n q_k\left(-\sum_{i=1}^K q_i(v_i-\END)\right) + \sum_{k=l}^n q_k\qND (D_k-\DND),
\end{eqnarray*}
where we used $\rho_k\big|_{\alpha=1}=q_k$, $\pND\big|_{\alpha=1}=\qND$, $EL_k(1)=0$ for all $k\ge l$ and
$$
\ELND(1)=\sum_{i=1}^K \frac{q_i}{\qND}(v_i-\END)
$$
following the calculation of $\ELND(\alpha)$ in $\ref{eq:EL_ND}$.

To prove the Proposition, we must show $S_1+S_2<0$. The proof proceeds as follows.
First, we consider the case $l=n$, where only the highest type certifies, and show that in this case $S_1+S2<0$. Then, we consider the case $l<n$ and show that $S_1+S_2$ is smaller when $l<n$ than when $l=n$.

\begin{lemma}
$E[v|v\le v_k]>v_{k-1}$ for every $k>1$, and $E[v|v\ge v_k]>v_{n-1}$ for every $k\ge 1$.
\label{le:Evk}
\end{lemma}
\begin{proof}
We prove the first statement by induction. When $k=2$ it is satisfied. Assume the statement
holds for $k-1$. Let $Q_k=q_1+\dots+q_k$. Then:
\begin{eqnarray*}
E[v|v\le v_k]&=&\sum_{i=1}^{k-1} \frac{q_i}{Q_k}v_i+\frac{q_k}{Q_k}v_k=\frac{1}{Q_k}
\left(Q_{k-1}\sum_{i=1}^{k-1} \frac{q_i}{Q_{k-1}}v_i+q_kv_k \right)\\
&>&\frac{1}{Q_k}\left(Q_{k-1}v_{k-2}+q_kv_k \right)\\
&=&\frac{1}{Q_k}(Q_k v_{k-1}- Q_{k-1}(v_{k-1}-v_{k-2}) + q_k(v_k-v_{k-1}))>v_{k-1}.
\end{eqnarray*}

The second part follows by applying the first part to the conditional distribution
$q^{cond}=\left(\frac{q_k}{1-Q_{k-1}},\dots,\frac{q_n}{1-Q_{k-1}}\right)$ over support $\{v_k,\dots,v_n\}$. $q^{cond}$ satisfies the requirement $q_j^{cond}>2Q_{j-1}^{cond}$
and, therefore, the first statement is applicable.
\end{proof}

\begin{lemma}
If $l=n$ then $S1=0$ and $S2<0$.
\label{le:l=n}
\end{lemma}

\begin{proof}
When only the highest type certifies, $D_n=0$ and thus $S_1=0$. Since $q_D=q_n$, term $S2$ becomes:
$$
S_2=(1-2q_n)q_n\left(-\sum_{k=1}^K q_k(v_k-\END)\right)+q_n\qND(0-\DND),
$$
and
\begin{eqnarray*}
\DND=\frac{q_n}{\qND}\sum_{k=1}^K q_k(v_n+v_k-2\END).
\end{eqnarray*}
Plugging it into $S2$ we get,
\begin{eqnarray*}
S_2&=&-q_n\sum_{k=1}^K q_k [(1-2q_n)(v_k-\END)+q_n(v_n+v_k-2\END)]=-q_n\sum_{k=1}^K q_k[\qND v_k-\END+v_nq_n].
\end{eqnarray*}
To prove that it is negative, we need to prove that
$$
\sum_{k=1}^K q_k[\qND v_k-\END+v_nq_n]=Q_K(\qND E_K-\END +v_nq_n)>0,
$$
where $E_K=E[v|v\le v_K]$. By Lemma \ref{le:Evk}, $E_K>\max\{v_{n-3},v_1\}$,
and $\END<v_{n-1}$. If $n\ge 4$, the inequality holds iff
$$
q_n>\frac{v_{n-1}-v_{n-3}}{v_n-v_{n-3}}.
$$
From $q_n>2(q_1+\dots+q_{n-1})$ follows that $q_n>2/3$. The RHS on the other hand is less than 2/3 because $v$'s have increasing differences:
\begin{eqnarray*}
&&\frac{v_{n-1}-v_{n-3}}{v_{n}-v_{n-3}}=1-\frac{v_{n}-v_{n-1}}{v_{n}-v_{n-3}}<1-
\frac{v_{n}-v_{n-1}}{3(v_{n}-v_{n-1})}=\frac{2}{3}.
\end{eqnarray*}
The cases $n=2,3$ are analogous.
\end{proof}

\begin{lemma}
Suppose $l\le n-1$. Let $v_l^\prime=v_l+\eps$, $v_n^\prime=v_n-\tau$, where $\eps>0$, $\tau>0$ and $E_C=E_C^\prime$. Then $S_1(v^\prime)+S_2(v^\prime)>S_1(v)+S_2(v)$.
\label{le:DeltaS}
\end{lemma}

\begin{proof}
Let $\Delta S = (S_1(v^\prime)+S_2(v^\prime)) - (S_1(v)+S_2(v))$. We aim to show that $\Delta S > 0$. The condition $E_C=E_C^\prime$ implies $q_l \eps = q_n \tau$, and therefore $\tau = \frac{\eps q_l}{q_n}$.

First, we compute $\Delta D_j = D_j(v^\prime) - D_j(v)$, which depends on $j\ge l$:
\begin{itemize}
    \item \textbf{Case $v_j=v_l$:} Since $v_l<E_C$, $D_l=\sum_{k=l+1}^{n}q_{k}(v_{k}-v_{l})$ and
    \begin{align*}
        \Delta D_l &= -\eps \sum_{k=l+1}^{n}q_k - \tau q_n=-\eps q_D.
    \end{align*}

\item \textbf{Case $v_l<v_j<E_C$:} Here $D_j=\sum_{k=j+1}^{n}q_{k}(v_{k}-v_{j})$, so
$\Delta D_j = -q_n\tau = -\eps q_l$.

        \item \textbf{Case $E_C\le v_j<v_n$:} Here $D_j=\sum_{k=l}^{j-1}q_{k}(v_{j}-v_{k})$, so again $\Delta D_j = -\eps q_l$.

    \item \textbf{Case $v_j=v_n$:} Since $v_n>E_C$, $D_n=\sum_{k=l}^{n-1}q_{k}(v_{n}-v_{k})$ and:
    \begin{align*}
        \Delta D_n = \sum_{k=l}^{n-1}q_{k}(-\tau) - q_l(-\eps) = -\frac{\eps q_l}{q_n}q_D.
    \end{align*}
\end{itemize}

By definition, $\Delta S_1 = \sum_{l \le i < j \le n} q_i q_j(\Delta D_j - \Delta D_i)$. The terms where $l<i<j<n$ are zero, since $\Delta D_j = \Delta D_i = -\eps q_l$. Thus,
\begin{align*}
    \Delta S_1 &= \sum_{j=l+1}^n q_j q_l(\Delta D_j - \Delta D_l) + \sum_{i=l+1}^{n-1} q_n q_i(\Delta D_n - \Delta D_i) \\
    &= q_l \eps \left[ (q_D-q_l)(q_D-q_l-q_n) + q_D(q_n-q_l) + (q_n-q_D)(q_D-q_l-q_n) \right] \\
    &= q_l\eps(q_n-q_l)(2q_D-q_l-q_n).
\end{align*}

As for $\Delta S_2$, note that $E_{ND}$, $K$, and $D_{ND}$ remain unchanged under the perturbation. Thus, the change in $S_2$ comes only from the change in the $D_k$ terms:
\begin{align*}
    \Delta S_2 = q_{ND} \sum_{k=l}^{n}q_k \Delta D_k=q_{ND} \cdot \eps q_l(q_l+q_n-3q_D).
\end{align*}

We now verify that $\Delta S_1 + \Delta S_2 > 0$.  Substitute $q_{ND}=1-q_D$ into the expression:
\begin{align*}
    \Delta S_1+\Delta S_2 &= \eps q_l \left[ (q_n-q_l)(2q_D-q_l-q_n) + (1-q_D)(q_l+q_n-3q_D) \right]\\
     &= \eps q_l \left[ 3q_D^2 + (q_n - 3q_l - 3)q_D + (q_l + q_n)(1 + q_l - q_n) \right].
\end{align*}
Let $f(q_D)$ denote a quadratic function inside the brackets. The valid domain for $q_D$ is inside the interval $[q_n+q_l, 1]$. When $q_n>2/3$, which holds since $q_n>2\sum_{j=1}^{n-1}q_j$,
its minimum of $f(q_D)$ over this interval is attained at $q_D = q_n+q_l$.
Evaluating at this point gives:
$$
f(q_n+q_l) = (3q_n^2 - 2q_n) + (q_l^2 + (4q_n-2)q_l)>0,
$$
when $q_n>2/3$. Thus, $\Delta S_1+\Delta S_2>f(q_n+q_l)>0$, completing the proof.
\end{proof}

We can now prove the Proposition. If $l=n$, then by Lemma \ref{le:l=n}, we are done. If $l<n$, then by Lemma \ref{le:Evk}, $E_C>v_{n-1}$. Let $\eps_l=v_{l+1}-v_l$, and choose $\tau$ so that
$E_C$ does not change. Then $v_l^\prime=v_l+\eps=v_{l+1}$
and $v_n^\prime>v_n-\tau$. Having $E_C>v_{n-1}$, guarantees that $v_{n}^\prime>E_C$.
By Lemma \ref{le:DeltaS}, $S_1^\prime+S_2^\prime>S_1+S_2$. We continue perturbations until
all quality types become equal to each other and equal to $E_C$. Each perturbation step increases $S_1+S_2$ by Lemma \ref{le:DeltaS}. When all quality values are the same, $S_1^{final}+S_2^{final}<0$ by
Lemma \ref{le:l=n}. Therefore, for the original distribution, we must have $S_1+S_2<0$. Thus,
$\frac{\partial K_1(\alpha)}{\partial\alpha}\Big|_{\alpha=1}<0$ and, therefore,
$\frac{\partial \Pi^\alpha}{\partial\alpha}\Big|_{\alpha=1}<0$ in a left neighborhood of $\alpha=1$ when $b$ is sufficiently large. $\blacksquare$

\newpage
\section{Translated Experimental Instructions}
\label{appendix:instructions}

Welcome to the experiment! Thank you very much for taking the time to participate. I will be your experimenter for today. This experiment is expected to last approximately 1 hour and 15 minutes. You will receive payment upon completing the entire session. Your payment will depend on your decisions and performance during the experiment. I will provide a detailed explanation of how your payment is calculated later. The experiment consists of three completely separate parts. In Part 1, you will play a game with others. In Part 2, you will participate in two individual decision problems. Finally, in Part 3, you will be asked to answer a simple questionnaire. After completing Part 1, I will provide instructions for Part 2, and once we finish Part 2, we will proceed to Part 3.

\begin{center}
\textbf{PART 1 - OVERVIEW}\footnote{These instructions are for a D3 treatment with 20 rounds and 20 subjects, where seller qualities range from 50 to 100. For other treatments with different parameters (number of subjects, rounds, or quality distribution), we adapt the instructions accordingly and provide them to the subjects.}
\end{center}

First, let me provide you with an overview of the first part of the experiment. You will be playing a game with others, comprising 20 rounds numbered from 1 to 20. In each round, you will be randomly assigned the role of either a buyer or a seller. The computer will inform you of your role for each round. It's important to note that the role assignments are completely random, meaning that your role in the next round is not influenced by your previous assignments. Also, given that there are only two roles, the likelihood of being assigned either the seller or buyer role is equal in every round.

There are 20 participants in total, and the computer will randomly divide you into five groups of four people each. Within each group, there will be two sellers and two buyers. Importantly, your group composition may change every round, and you will not be aware of the identities of your group members.

There will be two practice rounds before we proceed to the main twenty-round game. Throughout the two practice rounds and twenty main rounds, you can ask questions at any time. However, during the main rounds, please refrain from speaking out loud and instead raise your hand for assistance. I will be available to help you. The twenty main rounds will be played consecutively without any breaks. In Practice Round 1, some of you will assume the roles of sellers, while others will be buyers. In Practice Round 2, the roles will be reversed. If you are a seller in Practice Round 1, you will become a buyer in Practice Round 2, and vice versa. The purpose of the practice rounds is to help you become acquainted with the program. It is important to note that in the main rounds, your roles may or may not switch after every round. As explained earlier, role assignments in the main rounds are completely random. Additionally, please be aware that the outcomes of the practice rounds are completely unrelated to those in the main rounds.

\begin{center}
    \textbf{PART 1 - SELLER ROLE}
\end{center}

Second, let's talk about the seller role. As a seller, you will be provided with two units of a product, and your objective is to sell them. The quality of your product is represented by an integer ranging from 50 to 100 in each round. The computer will randomly select this integer, with each value having an equal chance of being chosen. The product quality is your private information, meaning that the two buyers and the other seller in the market are unaware of it (although everyone knows it's an integer between 50 and 100).

A third party in the market provides an imperfect certification at a cost. You can choose to utilize this certification to improve the credibility of the message you send to buyers. The certification cost is 10, 15, 20, or 25, depending on the specific round. The precision level of the technology is 0.5, 0.6, 0.7, 0.8, or 0.9, also depending on the round. To illustrate the meaning of precision level, consider the following example:

\medskip
\noindent \textit{Example 1:} Suppose the computer informs you that your true quality level is 70 and you decide to purchase the certification. Suppose in the considered round, the precision level of the certification technology is 0.8. This implies:
\begin{itemize}
    \item There is an 80\% chance that the certification succeeds. In this case, the certification technology will produce an outcome identical to your true quality level (so the outcome will be 70).

    \item However, with the remaining 20\%, the certification fails. In this case, the computer will randomly select an integer between 50 and 100 as your certification outcome, with each integer having an equal chance of being chosen. Hence, under this circumstance, the certification outcome you receive may or may not equal your true quality level.
\end{itemize}

Please note that the precision level of the certification outcome may or may not change after each round, as it is randomly selected by the computer.

After purchasing certification, you have the choice to disclose the outcome to buyers. If you choose to reveal the certification outcome, buyers will be notified that your product is certified by a third party. They will be able to view your certification outcome in the following form:

\bigskip
\hspace{2cm} I sell a product with a CERTIFIED quality of \_\_\_\_\_.
\bigskip

\noindent Alternatively, you may choose not to purchase certification technology or withhold the certification outcome after purchasing quality certification. In such instances, you can send a message to the buyers in the following form:

\bigskip
\hspace{2cm} I sell a product with a quality of\_\_\_\_\_.
\bigskip

You are free to choose a quality level (an integer) between 50 and 100 in the message above. You can inflate or deflate the quality of your product, or you can choose to be honest. The decision is entirely yours.

Once you finalize your disclosure decision, you will observe the disclosure decision of the competing seller. After that, you and your competing seller will simultaneously set prices. Your price must be a positive integer, and unlike the quality message, you are free to choose any price you like between 1 and 10,000. After setting a price, your task is complete, and you must wait to see if buyers purchase your products. Please note that you have a zero production cost and your profit in each round is calculated as follows:

\begin{table}[ht!]
\centering
\footnotesize
\begin{tabular}{|c|c|c|}
\hline
 & If you purchase certification & If you don't purchase certification \\
\hline
No buyer purchases your product & - certification cost & 0 \\
\hline
One buyer purchases your product & Your price – certification cost & Your price \\
\hline
Two buyers purchase your product & 2*Your price – certification cost & 2*Your price \\
\hline
\end{tabular}
\end{table}

Your profits from each round will be accumulated, and the total profit over twenty rounds will be converted to payment in Vietnamese dong. The conversion rate is set at 1,000 Vietnamese dongs per 20 units of profit. It's important to note that, as a seller, it's possible to make a negative profit, particularly when you purchase the certification technology (and pay the certification cost) but no buyers choose your products. In such cases, the same conversion rate will be applied. For instance, if you have a profit of -20, your payment will decrease by 1,000 Vietnamese dongs.

\begin{center}
\textbf{PART 1 - BUYER ROLE}
  \end{center}

Finally, let me talk about the buyer role. If you are a buyer, you will decide whether to buy a product. The information available to you includes: (i) there are two sellers in the market, referred to as Seller 1 and Seller 2; (ii) the quality of available products in the market varies from 50 to 100, with no exceptions--this range is fixed and known by everyone; and (iii) the quality messages from the sellers, including whether their quality messages are the result of certification or not.

If a seller does not purchase a quality certification or does not disclose the certification outcome (despite purchasing it), the message you will receive has the following form:

\bigskip
\hspace{2cm} I sell a product with a quality of \_\_\_\_\_\_.
\bigskip

With this message, sellers can inflate or deflate their product qualities, or they can choose to be honest. However, if the seller purchases a quality certification and discloses the certification outcome, you will receive the following message:

\bigskip
\hspace{2cm}  I sell a product with a CERTIFIED quality of \_\_\_\_\_\_.
\bigskip

In this case, the seller's quality was certified with a certification technology that has a precision of 0.5, 0.6, 0.7, 0.8, or 0.9, depending on the specific round. You will be informed of this precision level. To illustrate the meaning of a precision level, consider the following example.

\noindent \textit{Example 2:} Suppose you received the following message:

\bigskip
\hspace{2cm} I sell a product with a CERTIFIED quality of 70.
\bigskip

Also, suppose that the certification technology has a precision level of 0.8. This implies:
\begin{itemize}
\item There is an 80\% chance that the certification succeeds and the true quality is 70.

\item However, with the remaining 20\% probability, the certified quality you observed is randomly selected by the computer from all integers between 50 and 100, each with an equal likelihood of being chosen. You see 70 because the computer selected it. Hence, the observed certified quality (70) may or may not match the actual quality.
\end{itemize}

After observing the certification precision, certified or non-certified quality messages, and prices of both sellers, you will have three options to consider with corresponding profits as follows:

\begin{table}[ht!]
\centering
\footnotesize
\begin{tabular}{|l|c|}
\hline
 Options & Your profit  \\
\hline
Buy from Seller 1 &  Seller 1's true quality minus her price \\
\hline
Buy from Seller 2 &  Seller 2's true quality minus her price \\
\hline
Do not buy from any sellers & 0 \\
\hline
\end{tabular}
\end{table}

You can buy a maximum of one product, and you can always choose the product you prefer. Your profit from each round will be accumulated, and the total profit over twenty rounds will be converted to payment in Vietnamese dongs. The conversion rate is set at 1,000 Vietnamese dongs per 20 units of profit. As a buyer, you can also make a negative profit, especially when the actual quality of the purchased product is lower than the price paid. In such cases, the same conversion rate will be applied.

\begin{center}
\textbf{PART 1 - RESULT SCREEN - ALL ROLES}
\end{center}

Once all members in a group have made their decisions, a result screen will be displayed summarizing the outcomes. At this stage, you will know everything about the market. As a seller, you will be able to view the number of products you have sold, your profit, as well as the actual product quality, certification decision, certification outcome (if any), quality message (if any), price, and number of products sold by the competing seller. As a buyer, you will see your profit and the true product qualities of the two sellers, as well as the purchasing decision and profit of the other buyer in the market.

\begin{center}
    \textbf{EXPERIMENT PART 2}
\end{center}

Thank you for completing the first part of the experiment. In the second part, you will participate in two independent decision problems, one after the other, with no interaction with others. The two decision problems are unrelated to each other.

In Decision Problem X, you will be presented with 11 lotteries. In all lotteries, you have a 50\% chance of winning 10,000 Vietnamese dongs. However, the lotteries differ in potential losses. In lottery $i$, where $i$ is an integer between 0 and 10, you face a potential loss of $i$ thousand Vietnamese dongs. For instance, in lottery number 2, you have a 50\% chance of winning 10,000 Vietnamese dongs and a 50\% chance of losing 2,000 Vietnamese dongs. Your task is to choose whether to accept or reject each lottery.

In Decision Problem Y, the setup is similar to Decision Problem X, except now you will have 21 lotteries. In all lotteries, there is a $50\%$ of winning 20,000 Vietnamese dongs, which is double the winning amount in Decision Problem X. The lotteries also differ in terms of potential losses. In lottery $i$, where $i$ is an integer between 0 and 20, there is a $50\%$ chance of losing $i$ thousand Vietnamese dongs. For example, in lottery number 13, there is a $50\%$ chance of winning 10,000 Vietnamese dongs and a $50\%$ chance of losing 13,000 Vietnamese dongs. Your task is once again to decide whether to accept or reject each lottery.

Let's explain how your payment is determined in the second part of the experiment. You'll be paid based on your decision in one randomly chosen lottery. A fair coin toss determines the decision problem: heads for X, tails for Y. Then, a lottery is randomly selected from 11 (for X) or 21 (for Y) corresponding lotteries. Your payment will be determined based on your decision in the selected lottery. If you reject the lottery, your payment will be zero. If you accept the lottery, another fair coin will be tossed. If it lands on heads, you will receive the winning amount, and that will be added to your final payment. If it lands on tails, the losing amount will be deducted from your final payment.

\begin{table}[ht]
\footnotesize
\fontsize{9pt}{11pt}\linespread{1.2}\selectfont
\captionsetup{labelformat=empty}
\caption{\textbf{Table \ref{tab:X}: Decision Problem X}}
\begin{tabular}{|c|l|l|c|c|}
\hline
\textbf{Lottery \#} & \textbf{Winning chances and amount} &  \textbf{Losing chances and amount} & \textbf{Accept} & \textbf{Reject} \\
\hline
0 & 50\% chance of winning 10,000 VND & 50\% chance of losing 0 VND & & \\
\hline
1 & 50\% chance of winning 10,000 VND & 50\% chance of losing 1,000 VND & & \\
\hline
2 & 50\% chance of winning 10,000 VND & 50\% chance of losing 2,000 VND & & \\
\hline
3 & 50\% chance of winning 10,000 VND & 50\% chance of losing 3,000 VND & & \\
\hline
4 & 50\% chance of winning 10,000 VND & 50\% chance of losing 4,000 VND & & \\
\hline
5 & 50\% chance of winning 10,000 VND & 50\% chance of losing 5,000 VND & & \\
\hline
6 & 50\% chance of winning 10,000 VND & 50\% chance of losing 6,000 VND & & \\
\hline
7 & 50\% chance of winning 10,000 VND & 50\% chance of losing 7,000 VND & & \\
\hline
8 & 50\% chance of winning 10,000 VND & 50\% chance of losing 8,000 VND & & \\
\hline
9 & 50\% chance of winning 10,000 VND & 50\% chance of losing 9,000 VND & & \\
\hline
10 & 50\% chance of winning 10,000 VND & 50\% chance of losing 10,000 VND & & \\
\hline
\end{tabular}
\label{tab:X}
\begin{flushleft}
\emph{Notes}: For Decision Problem Y, subjects were offered a table similar to Table \ref{tab:X} except that, as described above, it had 21 lotteries, and the winning amount was 20,000 VND. We do not present it here for brevity.
\end{flushleft}
\end{table}

\newpage
\begin{center}
    \textbf{EXPERIMENT PART 3}
\end{center}

\begin{flushleft}
Thank you for completing the first two parts of the experiment. Before we finish, please complete the following brief questionnaire:

\begin{enumerate}
\item What is your age?

\item What is your gender?  $\square$ Male    $\square$ Female    $\square$ Others    $\square$ Prefer not to say

\item What is your major?

 $\square$ Economics or Business   $\quad$ $\square$ STEM (Science, Technology, Engineering, Mathematics)

 $\square$  Languages $\quad$ $\quad \quad$ $\quad$ $\quad$ $\square$ Others

\item What is your year in the program?  $\square$ Freshman     $\square$ Sophomore     $\square$ Junior     $\square$ Senior
\end{enumerate}
\end{flushleft}

\begin{center}
    \textbf{EXPERIMENT PAYMENT}
\end{center}
Thank you for your participation in the entire experiment. Your compensation consists of a show-up fee of 20,000 Vietnamese dongs plus additional payments earned in parts 1 and 2. Please allow me a few moments to calculate your total payment. I will then individually notify each of you and distribute the payments in cash. To ensure confidentiality, your payment will be discreetly placed in a sealed envelope. Upon receiving your envelope, please kindly exit the computer room.

\newpage
\section{Loss Aversion Estimation}
\label{appendix:LA}

We use two Multiple Price Lists (MLPs) to elicit the loss aversion of subjects referred to as  Decision Problem X and Decision Problem Y (see Table \ref{tab:X}). The subjects' choices were incentivized as described in Appendix \ref{appendix:instructions}.
If a subject displays a unique and normal switching point in both decision problems (i.e., accepting lotteries with small potential losses and rejecting lotteries with large potential losses), the loss aversion parameter is computed as:
$b=\frac{1}{2}\big(10/\theta_X+20/\theta_Y\big)$, where $\theta_X\ne 0$ and $\theta_Y\ne 0$ are the switching points in decision problems X and Y, respectively. If a subject displays a unique and normal switching point in only one decision problem, loss aversion is estimated based on that problem alone. Subjects who did not display unique and normal switching points in either X or Y were excluded from the analysis. In total, 27.5\% of subjects in decision problem X and 28.8\% of subjects in decision problem Y exhibited multiple switching points. Table \ref{tab:LA} provides estimates for loss aversion. For comparison, the median degree of loss aversion is estimated at
2.25 by Kahneman and Tversky (1992), while Wang et al. (2017) estimate it at 1.8 for Vietnam.

%

\begin{table}[ht!]
\fontsize{9pt}{11pt}\linespread{1.2}\selectfont
\centering
\caption{\textbf{Estimated values of loss aversion}}
\begin{tabular}{l|c|c|c|ccc|c|c|c}
\hline
\hline
& \multicolumn{4}{c}{\textbf{Panel A}} && \multicolumn{4}{c}{\textbf{Panel B}} \\
\cline{2-5} \cline{7-10}
 & Median & 25\% perc. &  75\% perc. & N &  & Median & 25\% perc. &  75\% perc. & N \\
\hline
$D2_{50}$  & 4.50 & 3.67 & 6.67 & 5 && 4.08 & 1.91 & 6.67 & 6\\
$D3_{50}$ &  2.06 & 1.97 & 3.05 & 12 && 2.00 & 1.94 & 2.36 & 17 \\
$D1_{80}$ & 2.34 & 1.67 & 8.3 & 6 && 1.83 & 1.50 & 5.51 & 8 \\
$D2_{80}$ & 2.25 & 2.00 & 2.67 & 9 && 2.13 & 2.00 & 2.67 & 10 \\
$D3_{80}$  & 3.00 & 2.31 & 3.67 & 12 && 3.00 & 2.00 & 3.67 & 15 \\
\hline
All treatments &2.58 & 2.00 & 3.67 & 44 && 2.25 & 1.93 & 3.17 & 56\\
\hline
\hline
\end{tabular}
\label{tab:LA}
\begin{flushleft}
\emph{Notes:} Panel A presents results for subjects exhibiting a unique switching point in both decision problems. Panel B expands upon Panel A by including subjects who exhibited a unique switching point in one of the decision problems. Loss aversion for these subjects is estimated based on the decision problem in which they had the unique switching point.
\end{flushleft}
\end{table}
\FloatBarrier

\end{appendices}
\end{document}